\documentclass[aoas]{imsart}

\RequirePackage{amsthm,amsmath,amsfonts,amssymb}
\RequirePackage[authoryear]{natbib}
\RequirePackage[colorlinks,citecolor=blue,urlcolor=blue]{hyperref}
\RequirePackage{graphicx}

\usepackage{psfrag,epsf}
\usepackage{booktabs}
\usepackage{multirow}
\usepackage{bbm}
\usepackage{verbatim}
\usepackage{color}
\usepackage{xr}
\usepackage{enumitem}

\newtheorem{theorem}{Theorem}[subsection]

\newtheorem{remark}{Remark}[subsection]
\newtheorem{lemma}{Lemma}

\newtheorem{proposition}{Proposition}[subsection]

\usepackage{newtxtext}
\usepackage[subscriptcorrection]{newtxmath}
\usepackage[plain,noend]{algorithm2e}

\newcommand{\Var}{\mathop{\rm Var}}
\newcommand{\Cov}{\mathop{\rm Cov}}

\graphicspath{{./art/}}

\makeatletter
\renewcommand{\algocf@captiontext}[2]{#1\algocf@typo. \AlCapFnt{}#2} 
\def\@algocf@capt@plain{top}
\renewcommand{\algocf@makecaption}[2]{%
  \addtolength{\hsize}{\algomargin}%
  \sbox\@tempboxa{\algocf@captiontext{#1}{#2}}%
  \ifdim\wd\@tempboxa >\hsize
    \hskip .5\algomargin%
    \parbox[t]{\hsize}{\algocf@captiontext{#1}{#2}}
  \else%
    \global\@minipagefalse%
    \hbox to\hsize{\box\@tempboxa}
  \fi%
  \addtolength{\hsize}{-\algomargin}%
}
\makeatother

\def\dis{\displaystyle}

\newcommand{\bigCI}{\mathrel{\text{\scalebox{1.07}{$\perp\mkern-10mu\perp$}}}}

\begin{document}

\begin{frontmatter}
\title{Subgroup analysis in multi level hierarchical\\ cluster randomized trials}
\runtitle{Subgroup analysis by S. Chakraborty et~al.}

\begin{aug}
\author[A]{\fnms{Shubhadeep}~\snm{Chakraborty} \thanks{\textbf{Corresponding author}}\ead[label=e1]{schakraborty.stat@gmail.com}},
\author[B]{\fnms{Bo}~\snm{Wang}\ead[label=e2]{Bo.Wang@umassmed.edu}}
\author[C]{\fnms{Ram}~\snm{Tiwari}\ead[label=e3]{ram.c.tiwari21@gmail.com}}
\and
\author[D]{\fnms{Samiran}~\snm{Ghosh}\ead[label=e4]{Samiran.Ghosh@uth.tmc.edu}}
\address[A]{Bristol Myers Squibb Company, USA \printead[presep={ ,\ }]{e1}}

\address[B]{Department of Population and Quantitative Health Sciences, University of Massachusetts Medical School, USA \printead[presep={ ,\ }]{e2}}

\address[C]{Global Stat Solutions, USA \printead[presep={ ,\ }]{e3}}

\address[D]{Department of Biostatistics and Data Science, University of Texas School of Public Health, USA \printead[presep={,\ }]{e4}}
\end{aug}

\begin{abstract}
Cluster or group randomized trials (CRTs) are increasingly used for both behavioral and system-level interventions, where entire clusters are randomly assigned to a study condition or intervention. Apart from the assigned cluster-level analysis, investigating whether an intervention has a differential effect for specific subgroups remains an important issue, though it is often considered an afterthought in pivotal clinical trials. Determining such subgroup effects in a CRT is a challenging task due to its inherent nested cluster structure. Motivated by a real-life HIV prevention CRT, we consider a three-level cross-sectional CRT, where randomization is carried out at the highest level and subgroups may exist at different levels of the hierarchy. We employ a linear mixed-effects model to estimate the subgroup-specific effects through their maximum likelihood estimators (MLEs). Consequently, we develop a consistent test for the significance of the differential intervention effect between two subgroups at different levels of the hierarchy, which is the key methodological contribution of this work. We also derive explicit formulae for sample size determination to detect a differential intervention effect between two subgroups, aiming to achieve a given statistical power in the case of a planned confirmatory subgroup analysis. The application of our methodology is illustrated through extensive simulation studies using synthetic data, as well as with real-world data from an HIV prevention CRT in The Bahamas.
\end{abstract}

\begin{keyword}
\kwd{Cluster randomized trials}
\kwd{consistent test}
\kwd{sample size determination}
\kwd{three-level CRT}
\kwd{variance component estimation}
\end{keyword}

\end{frontmatter}

\section{Introduction}\label{sec:intro}
For evaluating the efficacy of interventions for hierarchically structured data, typically arising from epidemiological or community-based clinical studies (see, for example, \cite{burns1995, platt2010, vicens2011}, among others), 
multi-level cluster randomized trials (CRTs) may be designed, where randomization, instead of at the subject or individual level, can be done at different cluster levels, such as the patients' provider, a group of providers, several practices in a health system or a specific geographical area, etc. \citep{cunningham2016}.  
This type of trials are preferred in situations when randomization at the subject level is impossible, cost-prohibitive from the perspective of group dynamics, or can pose significant challenges in trial validity due to contamination and other biases. A large volume of literature is available in this area that deals with CRTs, for example, two-level CRT \citep{Donner1998}, three-level CRTs \citep{Moerbeek}, correlated data \citep{laird1982random, hedeker, ahn2014sample}, three-level longitudinal CRTs \citep{heo_leon_2008}, multi-level CRTs \citep{Hedges2014}, four-level CRTs  \citep{Ghosh2022} and the references therein. Heterogeneity of treatment effect (HTE) can be defined as the non-random, explainable variability of treatment effects observed for individuals within a population. Subgroup analysis is one of the most common analytic approaches for examining HTE. Notably, subgroup analysis can serve either exploratory or confirmatory purposes, with the ultimate goal being to ascertain the differential effect of treatment, if any, among subgroups of trial participants. From a confirmatory perspective, the randomized controlled trial (RCT) must be designed with appropriate sample size and statistical power to facilitate confirmatory conclusions about individual subgroup(s). Several authors \citep{cui2002, pocock2002, lagakos2006} have discussed the challenges associated with such conclusions and have noted that the results of subgroup analysis are seldom regarded as definitive evidence. They recommend that primary emphasis should be placed on the overall primary outcome of the study, with appropriate planning for subgroup analysis in the event of evidence suggesting treatment heterogeneity across plausible subgroups. The issues become further complicated when the trial has a natural hierarchy, such as in CRT. This is because the heterogeneity of treatment effects can originate from multiple levels, and modeling such effects is not discussed in any of these papers to the best of our knowledge. We would like to point out that both \cite{heo_leon_2008} and \cite{Ghosh2022} did mention the existence of subgroups in the context of CRTs, but no formal analysis was carried out. Some recent works \citep{wang2023, li2023} develop formal methods for testing subgroup-specific treatment effects and analyzing treatment effect heterogeneity; however, they either consider a two-level CRT or consider subgroups only at the lowest or participant level.

\subsection{{Motivating example: Identification of differential subgroup effects in a three-level CRT arising from a HIV prevention program in the Bahamas}}\label{intro:motiv_eg}
{The methodology we develop in this article is motivated by the following real-world problem.} The Bahamas has the second highest prevalence rate of HIV in the Caribbean, with an estimated 3.3\% of Bahamian adults infected with HIV. AIDS has been the leading cause of death among young adults 15 to 44 years of age in the Bahamas \citep{figueroa2014}. To reverse the escalating rates of HIV in the Bahamas, the government embarked on an inter-agency approach, targeting children and adolescents. An education-based HIV prevention program {was} developed, which {was} delivered by teachers to middle-school students in a cluster randomized trial (CRT). The natural hierarchy of the trial consists of three levels: at the top level {(level three)} are schools, followed by classes/teachers {(level two)}, and finally, at the lowest level {(level one)}, are individual students, {with the randomization being done at the school level}. Information was collected on students' HIV/AIDS knowledge, condom use skills, condom use self-efficacy, intention to use a condom, and sexual behaviors. Condom use self-efficacy was one of the primary outcome variables for students considered in the present paper. {Subgrouping is possible at various levels of the hierarchy, for example, among students and/or classes/teachers. For subgrouping at levels one and two, whether gender of the students (Male vs. Female) and class size (Large vs. Small class), respectively, have any differential subgroup effect on the outcome, serves as the motivating question of interest.} 

\subsection{Our contributions}\label{intro:our_contri}
In the two-level CRT context, typical subgroup analysis is targeted at the lowest level (e.g., patient/student level) for a pre-defined relatively small number of key demographics or clinical covariates \citep{varadhan2013, lipkovich2024}. Albeit, in a multi-level CRT subgroup can exist at any level. In this article, we consider a three-level hierarchical cluster randomized trial design in which the level one units (for example, patients/students) are nested within level two units (for example, clinicians/teachers), which are in turn nested within level three units (for example, clinics/schools); and the randomization to the experimental or control intervention is done at the highest cluster level, that is, level three. We further consider that the effect of the intervention may vary across subgroups of the level one units (based on characteristics such as gender, age, prior therapies, disease severity, genotype, hormonal levels, etc. of the patients) or at the level two units (for example, if the problem of interest is to compare the effect of two different teaching methods, effects of large and small classes, etc. on the students' performances). The statistical term to describe differential intervention effects across subgroups, is the interaction between the intervention and the subgroup variable, as depicted later in models (\ref{Model:Subg1}) and (\ref{Model:Subg2}). In this work, we consider a pairwise comparison between two subgroups (both at level one or two) for the intervention effect. The main methodological innovation of this work is to develop a consistent test for testing the significance of differential intervention effect between the two subgroups. We first obtain an explicit expression for the maximum likelihood estimator (MLE) of the differential intervention effect, as well as its variance. Consequently we standardize the MLE by its standard error (SE) to obtain a test statistic to test for the significance of the differential intervention effect. Towards this, we first establish that explicit closed forms exist for the MLEs of the different variance components of the underlying models. Subsequently, we derive the explicit expressions of the variance component estimators, and rigorously establish their consistency. These results serve key theoretical contributions of this work, play instrumental role in establishing that the resulting test statistic is asymptotically normal, and lead to a consistent test for testing the significance of the differential intervention effect. We also derive explicit formulae for determining the sample sizes required for detecting a differential intervention effect between the two subgroups to achieve a given statistical power. Our methodology is supported by extensive numerical investigations conducted on synthetic data. To the best of our knowledge, we are unaware of earlier works in the literature that conduct a formal analysis to test for differential subgroup effects in a multi-level CRT, with subgroups existing at different cluster levels.

The rest of the article is outlined as follows. In Section \ref{sec:model}, we present the mixed-effects models for the three-level hierarchical cluster randomized trial with subgrouping at level one and two. In Section \ref{sec:methods}, we present the main methodological innovation of this work on developing a consistent test for the significance of the differential intervention effect between the two subgroups. Simulation studies with varying design parameters are outlined in Section \ref{sec:simulation}. Finally, in Section \ref{sec:real}, {we apply the proposed methodology} on the 
{motivating  example of the HIV prevention program in the Bahamas} and present our findings.

\section{Statistical Model}\label{sec:model}

\subsection{The unified framework}\label{subsec:general_model}
Consider a three-level hierarchical cluster randomized study design in which, for example, patients/students (level one units) are nested within physicians/teachers/classes (level two units), and physicians/teachers/classes are nested within hospitals/schools (level three units). Further, consider the level one or the level two units are categorized into two subgroups according to some attribute. Examples include grouping the patients/students (level one units) according to gender into males and females, or grouping the classes (level two units) into large and small. 
The goal of our work is two-fold:
\begin{itemize}
    \item[i)] To test for whether there is a significant difference in the intervention effect across the two subgroups; i.e., whether there is a differential subgroup effect.
    \item[ii)] To determine the sample size necessary to detect the differential subgroup effect to achieve a given statistical power.
\end{itemize}
\noindent This is important for a confirmatory analysis of subgrouping at any hypothesized level. In general, there might be more than two subgroups of interest; in this paper we focus on conducting a pairwise comparison of the intervention effect on two subgroups (we name them Subgroup 1 and Subgroup 2). Our results can be further extended for multiple subgroups with appropriate modification of hypothesis under investigation.

We begin with introducing some notations. For subgrouping at level one, let $Y_{ijgk}$ denotes a continuous response recorded for the $k^{th}$ patient/student ($1\leq k \leq N_1$) belongs to the $g^{th}$ subgroup
who was treated/instructed by the $j^{th}$ physician/teacher ($1\leq j \leq N_2$) within the $i^{th}$ hospital/school ($1\leq i \leq 2N_3$), with $g=1,2$. Similarly for level two,  $Y_{igjk}$ denotes the $k^{th}$ patient/student and the $j^{th}$ physician/teacher belongs to the $g^{th}$ subgroup, within the $i^{th}$ hospital/school. In other words, $g$ is the index for the subgroups defined either at level one or two.
Let $X_{ijgk}$ and $X_{igjk}$ denote the intervention assignment indicator for the two cases (subgroups defined at level two or one) such that \,$X_{ijgk}$ or $X_{igjk} = 1$ or $0$\, depending on whether the $i$th level three unit is assigned to an experimental intervention or a control intervention, respectively. We consider the randomization to the experimental or control intervention is done at level three, and therefore  $X_{ijgk}$ or $X_{igjk}=X_i$, $\forall$ $g, j$ and $k$. 

For notational simplicity we consider a balanced design where we assume $N_3$ level three units have been assigned to both the experimental and control interventions. Without any loss of generality, we assume that the first $N_3$ level three units ($i=1, \dots, N_3$) are assigned to receive the experimental intervention, and the remaining $N_3$ level three units ($i=N_3+1, \dots, 2N_3$) receive the control intervention.
In the following two subsections we present mixed-effects linear models for the continuous outcome $Y$ in case the subgroups are defined at levels two and one, respectively. We begin with introducing some notations.

\emph{Notations}. Denote by `$\otimes$' the Kronecker product of matrices. We use the notations `$\overset{P}{\to}$' and `$\overset{d}{\to}$' to denote convergence in probability and convergence in distribution, respectively. `O' and `o' stand for the usual notations in mathematics : `is no larger than' and `is ultimately smaller than'. Throughout the paper, we denote $1_a = (\underbrace{1, \dots, 1}_{a \, \textrm{times}})$ and $0_a = (\underbrace{0, \dots, 0}_{a \, \textrm{times}})$, for any generic positive \textit{a}.

\subsection{Subgrouping at level one}\label{subsec:subg1}
In this case, a three-level hierarchical mixed effects linear model (with the subgroup information) can be considered as follows:
\begin{align}\label{Model:Subg1}
	Y_{ijgk} = \beta_0 \,+\, \tau\,L_g \,+\, \xi\,X_{i} \,+\, \delta\, L_g\,X_{i} \,+\, u_i \,+\, u_{j(i)} \,+\, u_{g(ij)} \,+\, \epsilon_{ijgk}\,, 
\end{align}
for $k=1,\dots,N_1$, $j=1,\dots,N_2$, $i=1,\dots,2N_3$\, and\, $g=1,2$. The variable $L_g$ is the indicator for subgroup memberships, taking the value 1 or 0 depending on whether the patient/student (level one unit) under consideration falls under Subgroup 1 or 2. Suppose that out of the $N_1$ patients/students (level one units) treated/instructed by a common physician/teacher within a hospital/school, $N_{11}$ and $N_{12}$ of them belong to Subgroup 1 and Subgroup 2, respectively, with\, $N_{11} + N_{12} = N_1$.

The parameters $\beta_0$, $\tau$, $\xi$ and $\delta$ are the fixed effect parameters in the model. The parameter $\beta_0$ represents an overall fixed intercept; $\tau$ denotes the (slope) coefficient of the subgroup effect; $\xi$ denotes the coefficient of the intervention effect; and $\delta$ stands for the intervention by subgroup effect. We assume that the random effects and the error term in the model are normally distributed: $u_i \sim N(0, \sigma_3^2)$, $u_{j(i)} \sim N(0, \sigma_{2}^2)$, $u_{g(ij)} \sim N(0, \sigma_{\text{grp}}^2)$ and $e_{ijgk} \sim N(0, \sigma_e^2)$ for all $i, j, g$ and $k$. We further assume that $u_i \bigCI u_{j(i)} \bigCI u_{g(ij)} \bigCI e_{ijgk}$, that is, all the random components are mutually independent. From the model in (\ref{Model:Subg2}), the mean of the response $Y_{ijgk}$ can be expressed as:
\begin{equation}\label{eq_mean}
      E(Y_{ijgk}) \;=\; 
      \begin{cases}
      \beta_0 \,+\, \tau \,+\, \xi \,+\, \delta & \qquad \text{for}\;\; \text{Subgroup 1} \;(\text{i.e.,}\; L_g=1) \;\; \text{and}\;\; X_i=1\,, \\
      \beta_0 \,+\, \tau  & \qquad \text{for}\;\; \text{Subgroup 1}\;(\text{i.e.,}\; L_g=1) \;\; \text{and}\;\; X_i=0\,, \\
      \beta_0 \,+\, \xi  & \qquad \text{for}\;\; \text{Subgroup 2}\;(\text{i.e.,}\; L_g=0) \;\; \text{and}\;\; X_i=1\,,\\ 
      \beta_0  & \qquad \text{for}\;\; \text{Subgroup 2}\;(\text{i.e.,}\; L_g=0) \;\; \text{and}\;\; X_i=0\,.
    \end{cases}
\end{equation}
From (\ref{eq_mean}), note that\, $\beta_0 = E(Y_{ijgk}\,|\, L_g=0,\, X_i=0)$,\, $\beta_0 + \tau = E(Y_{ijgk}\,|\, L_g=1,\, X_i=0)$,\, $\beta_0 + \xi = E(Y_{ijgk}\,|\, L_g=0,\, X_i=1)$\, and\, $\beta_0 + \tau + \xi + \delta = E(Y_{ijgk})$. This essentially implies that
\begin{align}\label{eq_interpretations}
    \begin{split}
        \beta_0 \;&=\; E(Y_{ijgk}\,|\, L_g=0,\, X_i=0)\,,\\
        \tau \;&=\; E(Y_{ijgk}\,|\, L_g=1,\, X_i=0) \,-\, E(Y_{ijgk}\,|\, L_g=0,\, X_i=0)\,,\\
        \xi \;&=\; E(Y_{ijgk}\,|\, L_g=0,\, X_i=1) \,-\,  E(Y_{ijgk}\,|\, L_g=0,\, X_i=0)\,,\\
        \text{and} \qquad \delta \;&=\; \big[ E(Y_{ijgk}\,|\, L_g=1,\, X_i=1) \,-\, E(Y_{ijgk}\,|\, L_g=1,\, X_i=0)\big] \\
        & \;\;\;\; -\, \big[(E(Y_{ijgk}\,|\, L_g=0,\, X_i=1) \,-\, E(Y_{ijgk}\,|\, L_g=0,\, X_i=0)\big] \,.
    \end{split}
\end{align}
From (\ref{eq_interpretations}), we can interpret the fixed effect parameters in our model as follows: $\beta_0$ represents the effect of the control intervention on Subgroup 2; $\tau$ stands for the differential effect of the control intervention between Subgroup 1 and Subgroup 2;\, $\xi$ denotes the treatment effect of the experimental intervention (simply treatment effect, henceforth) on Subgroup 2; and finally $\delta$ represents the differential treatment effects between the two subgroups. Thus from (\ref{eq_mean}) and (\ref{eq_interpretations}), the treatment effects for the two subgroups can be expressed as:
\begin{align}\label{def_subgrp_effects}
    \begin{split}
        \xi \,+\, \delta \;&=:\; \delta_1 \;=\; E(Y_{ijgk}\,|\, L_g=1,\, X_i=1) \,-\, E(Y_{ijgk}\,|\, L_g=1,\, X_i=0) \qquad \text{for Subgroup 1} \,,\\
        \xi \;&=:\; \delta_2 \;=\; E(Y_{ijgk}\,|\, L_g=0,\, X_i=1) \,-\,  E(Y_{ijgk}\,|\, L_g=0,\, X_i=0) \qquad \text{for Subgroup 2}\,,
    \end{split}
\end{align}
with\, $\delta = \delta_1 - \delta_2$. Under model (\ref{Model:Subg1}), and with the specified assumptions on the random components, the covariance between the responses for varying subgroups and level one, two, and three indices can be expressed as:
\begin{equation}\label{cov_level1}
    \Cov(Y_{ijgk}, Y_{i'j'g'k'}) \;=\; \begin{cases}
        \sigma_3^2 \,+\, \sigma_2^2 \,+\, \sigma_{\text{grp}}^2 \,+\, \sigma_e^2 \;=:\; \sigma^2 & \qquad \text{for}\;\; i=i',\,j=j',\,g=g',\,k=k'\,,\\
        \sigma_3^2 \,+\, \sigma_2^2 \,+\, \sigma_{\text{grp}}^2 \;=:\; \sigma^2\,\rho_1 & \qquad \text{for}\;\; i=i',\,j=j',\,g=g',\,k\neq k'\,,\\
        \sigma_3^2 \,+\, \sigma_2^2 \;=:\; \sigma^2\,\rho_{(1)} & \qquad \text{for}\;\; i=i',\,j=j',\,g\neq g'\,,\\
        \sigma_3^2 \,+\, \sigma_{\text{grp}}^2  \;=:\; \sigma^2\,\rho_{(2)} & \qquad \text{for}\;\; i=i',\,j\neq j',\,g=g'\,,\\
        \sigma_3^2 \;=:\; \sigma^2\,\rho_2 & \qquad \text{for}\;\; i=i',\,j\neq j',\,g\neq g'\,,\\
        0   & \qquad \text{otherwise}\,,
    \end{cases}
\end{equation}
with $\sigma^2$ as the variance of the outcome $Y_{ijkg}$, where $\rho_1, \rho_{(1)}, \rho_{(2)}$ and $\rho_2$ are the intra-class correlations among different patients/students (level one units) who are treated/instructed by a common physician/teacher and fall under the same subgroup; are treated/instructed by a common physician/teacher but fall under different subgroups; are treated/instructed by different physicians/teachers but fall under the same subgroup; and are treated/instructed by different physicians/teachers and fall under different subgroups, respectively. The intra-class correlation coefficients satisfy $1 \geq \rho_1 \geq \rho_{(1)}, \rho_{(2)} \geq \rho_2 \geq 0$, along with the constraint that $\sigma^2 \rho_1 - \sigma^2 \rho_{(1)} - \sigma^2 \rho_{(2)} + \sigma^2 \rho_2 =0$, implying\, $\rho_1 + \rho_2 = \rho_{(1)} + \rho_{(2)}$.

In Section \ref{subsec:subg1_theory}, we develop a methodology to test for a differential treatment effect between the two subgroups at level one, considering a balanced design. We assume $N_{11} = N_{12} = n$ and $N_1 = 2n$ in the model presented in (\ref{Model:Subg1}), or in other words, both Subgroup 1 and Subgroup 2 have identical number of patients/students (level one units). For balanced data, using mixed models notations, model (\ref{Model:Subg1}) can be written as
\begin{align}\label{Model:Subg1_LM}
  Y \;=\; X\beta\,+\, Z_1\, U_1 +\, Z_2\, U_2 +\, Z_3\, U_3 +\, Z_0\, \epsilon\,,   
\end{align}
where $Y$ is the $2N_3 N_2\, 2n \times 1$ vector of responses; $X$ is the $2N_3 N_2\, 2n \times 4$ design matrix for the fixed effects given by
\begin{align*}
    X \;=\; \begin{pmatrix}
        1_{2N_3 N_2 2n} \;&\; 
        1_{2N_3 N_2} \otimes \begin{pmatrix}
            1_n\\
            0_n
        \end{pmatrix}  \;&\; \;
        \begin{matrix}
            1_{N_3 N_2 2n}\\
            0_{N_3 N_2 2n}
        \end{matrix} \;&\; 
        \begin{matrix}
            1_{N_3 N_2} \otimes \begin{pmatrix}
            1_n\\
            0_n
        \end{pmatrix} \\
        0_{N_3 N_2 2n}
        \end{matrix}      
    \end{pmatrix}\,;
\end{align*}
$\beta=(\beta_0, \tau, \xi, \delta)^T$ is the $4\times 1$ vector of fixed effects coefficients; $U_1=(\!( u_i )\!),\, U_2=(\!( u_{j(i)} )\!)$ and $U_3=(\!( u_{g(ij)} )\!)$ are the $2N_3 \times 1$, $2N_3 N_2 \times 1$ and $2N_3 N_2\,2 \times 1$ vectors of random effects, respectively; $\epsilon$ is the error vector; and $Z_0 = I_{2N_3} \otimes I_{N_2} \otimes I_2 \otimes I_n,\, Z_1 = I_{2N_3} \otimes 1_{N_2} \otimes 1_2 \otimes 1_n,\, Z_2 = I_{2N_3} \otimes I_{N_2} \otimes 1_2 \otimes 1_n$ and $Z_3 = I_{2N_3} \otimes I_{N_2} \otimes I_2 \otimes 1_n$ are the design matrices for the random effects.

\subsection{Subgrouping at level two}\label{subsec:subg2}
In this case, a three-level hierarchical mixed effects linear model (with the subgroup information) can be considered as follows:
\begin{align}\label{Model:Subg2}
\begin{split}
    Y_{igjk} \;&=\; \beta_0 \,+\, \tau\,L_g \,+\, \xi\,X_{i} \,+\, \delta\, L_g\,X_{i} \,+\, u_i \,+\, u_{g(i)} \,+\, u_{j(ig)} \,+\, \epsilon_{igjk}\,,
    \end{split}
\end{align}
Next, we develop a test for testing a differential treatment effect between the two subgroups at level two in Section \ref{subsec:subg2_theory}, considering a balanced design. We assume $N_{21} = N_{22} = n$ and $N_2 = 2n$ in the model presented in (\ref{Model:Subg2}), or in other words, both Subgroup 1 and Subgroup 2 have identical number of physicians/teachers (level two units). For balanced data, using mixed models notations, model (\ref{Model:Subg2}) can be written as
\begin{align}\label{Model:Subg2_LM}
  Y \;=\; X\beta\,+\, Z_1\, U_1 +\, Z_2\, U_2 +\, Z_3\, U_3 +\, Z_0\, \epsilon\,,   
\end{align}
where $Y$ is the $2N_3 2n N_1 \times 1$ vector of responses; $X$ is the $2N_3 2n N_1 \times 4$ design matrix for the fixed effects given by
\begin{align*}
    X \;=\; \begin{pmatrix}
        1_{2N_3 2n N_1} \;&\; 
        1_{2N_3} \otimes \begin{pmatrix}
            1_n\\
            0_n
        \end{pmatrix} \otimes 1_{N_1} \; &\;\;
        \begin{matrix}
            1_{N_3 2n N_1}\\
            0_{N_3 2n N_1}
        \end{matrix} \;&\; 
        \begin{matrix}
            1_{N_3} \otimes \begin{pmatrix}
            1_n\\
            0_n
        \end{pmatrix} \otimes 1_{N_1}\\
        0_{N_3 2n N_1}
        \end{matrix}      
    \end{pmatrix}\,;
\end{align*}
$\beta=(\beta_0, \tau, \xi, \delta)^T$ is the $4\times 1$ vector of fixed effects coefficients; $U_1=(\!( u_i )\!),\, U_2=(\!( u_{g(i)} )\!)$ and $U_3=(\!( u_{j(ig)} )\!)$ are the $2N_3 \times 1$, $2N_3 2 \times 1$ and $2N_3 2n \times 1$ vectors of random effects, respectively; $\epsilon$ is the error vector; and $Z_0 = I_{2N_3} \otimes I_2 \otimes I_n \otimes I_{N_1},\, Z_1 = I_{2N_3} \otimes 1_2 \otimes 1_n \otimes 1_{N_1},\, Z_2 = I_{2N_3} \otimes I_2 \otimes 1_n \otimes 1_{N_1}$ and $Z_3 = I_{2N_3} \otimes I_2 \otimes I_n \otimes 1_{N_1}$ are the design matrices for the random effects.

Under model (\ref{Model:Subg2}), and with the specified assumptions on the random components, the covariance between the responses for varying subgroups and level one, two and three indices can be expressed as follows: 
\begin{equation}\label{cov_level2}
    \Cov(Y_{igjk}, Y_{i'g'j'k'}) \;=\; \begin{cases}
        \sigma_3^2 \,+\, \sigma_2^2 \,+\, \sigma_{1}^2 \,+\, \sigma_e^2 \;=\; \sigma^2 & \qquad \text{for}\;\; i=i',\,g=g',\,j=j',\,k=k'\,,\\
        \sigma_3^2 \,+\, \sigma_2^2 \,+\, \sigma_{1}^2 \;=:\; \sigma^2\,\rho_1 & \qquad \text{for}\;\; i=i',\,g=g',\,j=j',\,k\neq k'\,,\\
        \sigma_3^2 \,+\, \sigma_{2}^2  \;=:\; \sigma^2\,\rho_{(2)} & \qquad \text{for}\;\; i=i',\,g=g',\,j\neq j'\,,\\
        \sigma_3^2 \;=:\; \sigma^2\,\rho_2 & \qquad \text{for}\;\; i=i',\,g\neq g'\,,\\
        0   & \qquad \text{otherwise}\,,
    \end{cases}
\end{equation}
where $\rho_1, \rho_{(2)}$ and $\rho_2$ are the intra-class correlations among different patients/students (level one units) who are treated/instructed by a common physician/teacher; are treated/instructed by different physicians/teachers who fall under the same subgroup; and are treated/instructed by different physicians/teachers who fall under different subgroups, respectively. The intra-class correlation coefficients satisfy\, $1 \geq \rho_1 \geq \rho_{(2)} \geq \rho_2 \geq 0$.

To summarize, our work focuses on testing for whether the treatment effects are significantly different across the two subgroups in a three-level cluster randomized trial, when the randomization is done at the highest level of the hierarchy. In other words, we aim to test for the significance of the differential treatment effect, $\delta$, between the two subgroups, that is, $H_{0} : \delta = 0$ \,vs\, $H_{a} : \delta \neq 0$. For both the cases where the subgroups are defined by categorizations of level one and level two units, we first derive the expressions of the maximum likelihood estimators (MLE) of $\delta$; then obtain explicit expressions of the variances of the MLEs of $\delta$, followed by deriving explicit, closed form expressions of the MLEs of the variance components and establishing their consistency properties. Finally, we propose a test statistic to test for $H_{0}$ \,vs\, $H_{a}$ and establish its asymptotic normality. In addition, we also derive explicit formulae for sample size determination for detecting the differential treatment effect, to achieve a given statistical power. The methodological developments are presented in details in the next section.

\section{Methodology}\label{sec:methods}

In this section, we develop a test for a differential treatment effect between the two subgroups (at both level two and level one), and to determine the sample size required to achieve a certain pre-specified power for detecting the differential subgroup effect. In view of the long and technical nature of the derivations, the proofs of all the propositions and theorems in subsections \ref{subsec:subg1_theory} and \ref{subsec:subg2_theory} are relegated to the supplementary materials.

\subsection{Subgrouping at level one}\label{subsec:subg1_theory}

In this subsection, we consider the case when the subgroups are defined by categorizations of the level one units (patients/students), that is, we focus on model (\ref{Model:Subg1}). We begin with the following proposition that presents an explicit expression of the maximum likelihood estimator (MLE) of $\delta$. 

\begin{proposition}\label{Prop_MLE_delta1}
{\rm The MLE of the differential treatment effect $\delta$ is given by}
\begin{align*}
        \hat{\delta} \;&=\; \left( \frac{1}{N_3 N_2 \,n} \dis \sum_{i=1}^{N_3} \sum_{j=1}^{N_2} \sum_{k=1}^{n} Y_{ij1k} \;-\; \frac{1}{N_3 N_2 \,n} \dis \sum_{i=N_3+1}^{2N_3} \sum_{j=1}^{N_2} \sum_{k=1}^{n} Y_{ij1k} \right)\\
        & \qquad - \; \left( \frac{1}{N_3 N_2\, n} \dis \sum_{i=1}^{N_3} \sum_{j=1}^{N_2} \sum_{k=1}^{n} Y_{ij2k} \;-\; \frac{1}{N_3 N_2 \,n} \dis \sum_{i=N_3+1}^{2N_3} \sum_{j=1}^{N_2} \sum_{k=1}^{n} Y_{ij2k} \right)\;=: \; \hat{\delta}_1 \,-\, \hat{\delta}_2\,.
\end{align*}
\end{proposition}
\noindent From (\ref{eq_interpretations}) and Proposition \ref{Prop_MLE_delta1}, it follows that the MLE\, $\hat{\delta}$ is an unbiased estimator of $\delta$. Next, we provide an explicit expression of the variance of $\hat{\delta}$ in the following proposition, which is instrumental in the subsequent discussions on sample size determination for detecting the differential treatment effect and constructing a statistical test for its significance. 

\begin{proposition}
\label{Prop_var_delta1}
{\rm The variance of $\hat{\delta}$ is given by}
\begin{align*} 
    \begin{split}
        \Var(\hat{\delta}) \;&=\; \frac{4}{N_3 N_2\, n}\, \left( \sigma_e^2 \,+\, n\,N_2\,\sigma_{\text{grp}}^2 \right)\,.
    \end{split}
\end{align*}    
\end{proposition}
We now introduce some appropriate sum of squares quantities below which we will need in the estimation of the variance components in our subsequent discussions. 

Define,
\begin{align}\label{SS_level1}
\begin{split}
    SS_0 \;&=\; \dis \sum_{i=1}^{2N_3} \sum_{j=1}^{N_2} \sum_{g=1}^2 \sum_{k=1}^n \left( Y_{ijgk} - \bar{Y}_{ijg\cdot} \right)^2\,,\\
    SS_1 \;&=\; n\, \Big[ \dis \sum_{i=1}^{N_3} \sum_{j=1}^{N_2} \left( (\bar{Y}_{ij1\cdot} - \bar{Y}_{ij2\cdot}) \,-\, (\bar{Y}^T_{\cdot \cdot 1 \cdot} - \bar{Y}^T_{\cdot \cdot 2 \cdot}) \right)^2 \;+\; \sum_{i=N_3+1}^{2N_3} \sum_{j=1}^{N_2} \left( (\bar{Y}_{ij1\cdot} - \bar{Y}_{ij2\cdot}) \,-\, (\bar{Y}^C_{\cdot \cdot 1 \cdot} - \bar{Y}^C_{\cdot \cdot 2 \cdot}) \right)^2 \Big]\,,\\
    SS_2 \;&=\; 2n\,\dis \sum_{i=1}^{2N_3} \sum_{j=1}^{N_2} \left( \bar{Y}_{ij\cdot \cdot} - \bar{Y}_{i\cdot \cdot \cdot} \right)^2 \;\;\;
    \textrm{and} \;\;\;\; SS_3 \;=\; N_2\,2n\,\Big[ \dis \sum_{i=1}^{N_3} \left(\bar{Y}_{i \cdot \cdot \cdot} - \bar{Y}^T_{\cdot \cdot \cdot \cdot}\right)^2 \,+\, \sum_{i=N_3+1}^{2N_3} \left(\bar{Y}_{i \cdot \cdot \cdot} - \bar{Y}^C_{\cdot \cdot \cdot \cdot}\right)^2 \Big]\,,
    \end{split}
\end{align}
where 
\begin{align*}
  &\bar{Y}_{ijg\cdot} = \frac{1}{n} \dis \sum_{k=1}^n Y_{ijgk} \,, \qquad \bar{Y}_{ij\cdot \cdot} = \frac{1}{2n} \dis \sum_{g=1}^2 \sum_{k=1}^n Y_{ijgk}\,, \qquad \bar{Y}^T_{\cdot \cdot g \cdot} = \frac{1}{N_3 N_2\,n} \sum_{i=1}^{N_3} \sum_{j=1}^{N_2} \sum_{k=1}^n Y_{ijgk}\,, \\
  &\bar{Y}^C_{\cdot \cdot g \cdot} = \frac{1}{N_3 N_2\,n} \sum_{i=N_3+1}^{2N_3} \sum_{j=1}^{N_2} \sum_{k=1}^n Y_{ijgk}\,, \qquad \bar{Y}_{i\cdot \cdot \cdot} = \frac{1}{N_2\,2n} \sum_{j=1}^{N_2} \sum_{g=1}^2 \sum_{k=1}^n Y_{ijgk}\,,\\ &\bar{Y}^T_{\cdot \cdot \cdot \cdot} = \frac{1}{N_3} \sum_{i=1}^{N_3}\bar{Y}_{i\cdot \cdot \cdot}\;\;\; \textrm{and} \;\;\; \bar{Y}^C_{\cdot \cdot \cdot \cdot} = \frac{1}{N_3} \sum_{i=N_3+1}^{2N_3}\bar{Y}_{i\cdot \cdot \cdot}\,,
\end{align*}
for $i=1,\dots,2N_3$,\, $j=1,\dots,N_2$,\, and\, $g=1,2$.
Note that in order to estimate $\Var(\hat{\delta})$ given in Proposition \ref{Prop_var_delta1}, we need to estimate all the variance components $\sigma_3^2,\, \sigma_2^2, \,\sigma_{\text{grp}}^2$ and $\sigma_e^2$ in (\ref{cov_level1}). It is worth mentioning that the presence of explicit, closed form solutions to the maximum likelihood (ML) equations for the variance components of balanced mixed models is not always guaranteed (see, for example, Section 4.7.e in \cite{searle2009}); in fact \cite{szatrowski1980} proposed some necessary and sufficient conditions to check for determining whether or not explicit maximum likelihood estimates exist. The following theorem ensures the existence of explicit, closed form maximum likelihood estimates of the variance components in the balanced mixed effects model presented in (\ref{Model:Subg1}), with $N_{11} = N_{12} = n$. The exact conditions to be checked are presented in full technical details in the proof of the theorem. We also rigorously derive the explicit forms of the MLEs of the variance components, that are crucial for constructing a test statistic to test the null hypothesis $H_{0} : \delta = 0$ \,vs\, $H_{a} : \delta \neq 0$. 

\begin{theorem}\label{Thm_MLE_var_comps1}
Explicit, closed form maximum likelihood estimates of the variance components exist for the balanced case of the mixed effects model presented in (\ref{Model:Subg1}), and are given by
\begin{align*}
    &\hat{\sigma}_e^2 \;=\; \frac{1}{2N_3 N_2 \,2(n-1)} \; SS_0\,,\\
    &\hat{\sigma}_{\text{grp}}^2 \;=\; \frac{1}{n}\,\left( \frac{1}{2N_3 N_2 \,2} \; SS_1 \;-\; \frac{1}{2N_3 N_2 \,2(n-1)} \; SS_0\right)\,,\\
    &\hat{\sigma}_2^2 \;=\; \frac{1}{n}\,\left( \frac{1}{2N_3 (N_2 -1)} \; SS_2 \;-\; \frac{1}{2N_3 N_2 \,2} \; SS_1\right)\,,\\
    \textrm{and} \qquad &\hat{\sigma}_3^2 \;=\; \frac{1}{N_2\,2n}\,\left( \frac{1}{2N_3} \; SS_3 \;-\; \frac{1}{2N_3 (N_2-1)} \; SS_2\right)\,,
\end{align*} 
where the sum of squares $SS_0, SS_1, SS_2$ and $SS_3$ are given in (\ref{SS_level1}).
\end{theorem}

The next theorem establishes the consistency of the MLEs of the variance components. 


\begin{theorem}\label{Thm_consistency_level1}
 The maximum likelihood estimators $\hat{\sigma}_e^2$, $\hat{\sigma}_{\text{grp}}^2$, $\hat{\sigma}_2^2$ and $\hat{\sigma}_3^2$ are unbiased estimators of the respective true parameters $\sigma_e^2$, $\sigma_{\text{grp}}^2$, $\sigma_2^2$ and $\sigma_3^2$. Moreover, as\, $N_3, N_2, n \to \infty$, $\Var(\hat{\sigma}_e^2)$, $\Var(\hat{\sigma}_{\text{grp}}^2)$, $\Var(\hat{\sigma}_2^2)$ and $\Var(\hat{\sigma}_3^2)$ are $o(1)$ quantities. In other words, as\, $N_3, N_2, n \to \infty$,
 \begin{align*}
   \hat{\sigma}_e^2 \;\overset{P}{\to}\; \sigma_e^2\,, \;\; \hat{\sigma}_{\text{grp}}^2  \;\overset{P}{\to}\;  \sigma_{\text{grp}}^2\,,\;\; \hat{\sigma}_2^2 \;\overset{P}{\to}\; \sigma_2^2\, \;\; \textrm{and} \;\;\; \hat{\sigma}_3^2 \;\overset{P}{\to}\; \sigma_3^2\,.
 \end{align*}
\end{theorem}
Thus, we can estimate $\Var(\hat{\delta})$ given in Proposition \ref{Prop_var_delta1} by plugging in the MLEs of the variance components, and the resulting variance estimator, denoted by $\widehat{\Var(\hat{\delta})}$, can be shown to be consistent. We define a statistic $T$  by standardizing $\hat{\delta}$ as follows:
\begin{align}\label{test_stat_1}
    T \;:=\; \frac{\hat{\delta} - \delta}{\sqrt{\widehat{\Var(\hat{\delta})}}}\,.
\end{align}
The consistency properties of the variance component estimators turn out to be instrumental in the following theorem, where we study the asymptotic behavior of the statistic $T$.
\begin{theorem}\label{Thm_asymp_T_level1}
   As $N_3, N_2, n \to \infty$, we have \,$T \overset{d}{\to} N(0,1)$. 
\end{theorem}
Under the null hypothesis $H_{0} : \delta = 0$, the above statistic boils down to\, $T_0 := \frac{\hat{\delta}}{\sqrt{\widehat{\Var(\hat{\delta})}}}$. We consider $T_0$ as the test statistic to test for $H_{0}$ against the alternative hypothesis $H_{a} : \delta \neq 0$. As a consequence of Theorem \ref{Thm_asymp_T_level1}, we have that under $H_0$, the test statistic $T_0$ asymptotically follows a standard normal distribution. Thus, to test for $H_{0}$ against $H_{a} : \delta \neq 0$ at a given level $\alpha \in (0,1)$, the rejection rule can be proposed as follows: reject $H_0$ if\, $\vert T_0 \vert > z_{1-\alpha/2}$, where $z_{1-\alpha/2}$ is the upper $\alpha/2$ quantile of the standard normal distribution, that is, $z_{1-\alpha/2} = \Phi^{-1}(1-\alpha/2)$, and $\Phi$ denotes the cumulative distribution function of the standard normal distribution.

The discussions so far have been aimed at constructing an appropriate test statistic to test for whether there is a significant difference in the treatment effect between the two subgroups. We now turn to the design aspect and determining sample size requirements for cluster randomized trials with three level hierarchical data to detect a differential subgroup effect.  

Consider the statistic $\tilde{T}=\frac{\hat{\delta} - \delta}{\sqrt{\Var(\hat{\delta})}}$, where the parameters $\sigma^2$, $\rho_1$, $\rho_{(1)}$, $\rho_{(2)}$ and $\rho_2$, or equivalently, the variance components $\sigma_3^2,\, \sigma_2^2, \,\sigma_{\text{grp}}^2$ and $\sigma_e^2$ are assumed to be known. Then, we have $\tilde{T} \sim N(0,1)$. Note that this is true even for finite samples, when the variance components are known. When the variance components are unknown, we need to replace $\Var(\hat{\delta})$ in the denominator of $\tilde{T}$ by $\widehat{\Var(\hat{\delta})}$ to get $T$, which is shown to asymptotically follow a standard normal distribution in Theorem \ref{Thm_asymp_T_level1}.

At a fixed level $\alpha \in (0,1)$, to test for the null hypothesis $H_{0} : \delta = 0$ against $H_{a} : \delta \neq 0$ based on $\tilde{T}$, that is, assuming the variance components are known, we reject $H_0$ if\, $\vert \tilde{T}_0 \vert := \left\vert \frac{\hat{\delta}}{\sqrt{\Var(\hat{\delta})}} \right\vert > z_{1-\alpha/2}$. The power of the test, denoted by $\phi$, can therefore be written as follows:
\begin{align}\label{power_ind_diff}
    \begin{split}
        \phi \;=\; 1-\beta \;&\leq\; P_{H_{a}}\left( \vert \tilde{T}_0 \vert > z_{1-\alpha/2} \right) \;\leq \; P_{H_{a}}\left( \vert \tilde{T} \vert \; >\; z_{1-\alpha/2} \;-\;  \frac{\vert\delta \vert}{\sqrt{\Var(\hat{\delta})}}\right)\\
        &=\; \Phi \left( \frac{\vert \delta \vert}{\sqrt{\Var(\hat{\delta})}} \;-\; z_{1-\alpha/2} \right)\,,
    \end{split}
\end{align}
where\, the second inequality follows from a simple use of triangle inequality and $\beta$\, is the probability of the type II error. 
From (\ref{power_ind_diff}), some simple algebraic manipulations yield
\begin{align}\label{power_ind4}
    \frac{\delta^2}{\Var(\hat{\delta})} \;\geq\; \left(z_{1-\alpha/2} \,+\, z_{1-\beta}  \right)^2\,.
\end{align}
Thus, in order to achieve the power $\phi=1-\beta$, the sample sizes need to satisfy (\ref{power_ind4}). Assuming $N_2$ and $N_3$ (that is, the number of level two and level three units) are known, for a fixed set of parameter values and for a given power $\phi=1-\beta$, the problem thus reduces to finding the smallest $n=N_{11}=N_{12}$ that satisfies (\ref{power_ind4}).

\begin{proposition}\label{Prop_samp_size1}
{\rm
    Assuming $N_2, N_3$ and the parameters $\delta$, $\sigma^2$, $\rho_1$, $\rho_{(1)}$, $\rho_{(2)}$ and $\rho_2$ are known, the number of level one units needed within each subgroup to achieve the power $\phi = 1-\beta$\, is the smallest integer greater than
    \begin{align*}
        n \;=\; A\,\left( \frac{\delta^2}{\big(z_{1-\alpha/2} \,+\, z_{1-\beta} \big)^2} \,-\, B\, \right)^{-1}, \qquad \textrm{where} \qquad A \;=\; \frac{4\sigma_e^2}{N_3 N_2} \;\;\; \textrm{and} \;\; B \;=\; \frac{4\sigma_{\text{grp}}^2}{N_3}  \,.
    \end{align*}
}
\end{proposition}

\subsection{Subgrouping at level two}\label{subsec:subg2_theory}
In this subsection, we present similar developments as presented in the previous subsection, but this time in the context of subgrouping of the level two units and with model (\ref{Model:Subg2}) in view. Not to add another fresh set of notations for the sample sizes, the different parameters and their estimates, that is, to avoid as much as possible complexities arising from introducing too many notations, we choose to use the same set of notations that we have used in the previous subsection. It only needs to be understood that the notations in this subsection are now being used in the context of model (\ref{Model:Subg2}) and subgrouping at level two. We begin with the following proposition that presents an explicit expression of the MLE of $\delta$. 

\begin{proposition}\label{Prop_MLE_delta2}
{\rm The MLE of the differential treatment effect $\delta$ is given by}
\begin{align*}
        \hat{\delta} \;&=\; \left( \frac{1}{N_3\,n\,N_1} \dis \sum_{i=1}^{N_3} \sum_{j=1}^{n} \sum_{k=1}^{N_1} Y_{i1jk} \,-\, \frac{1}{N_3\,n\,N_1} \dis \sum_{i=N_3+1}^{2N_3} \sum_{j=1}^{n} \sum_{k=1}^{N_1} Y_{i1jk} \right)\\
        & \qquad - \; \left( \frac{1}{N_3\,n\,N_1} \dis \sum_{i=1}^{N_3} \sum_{j=1}^{n} \sum_{k=1}^{N_1} Y_{i2jk} \,-\, \frac{1}{N_3\,n\,N_1} \dis \sum_{i=N_3+1}^{2N_3} \sum_{j=1}^{n} \sum_{k=1}^{N_1} Y_{i2jk} \right)\\
        &=: \; \hat{\delta}_1 \,-\, \hat{\delta}_2\,.
\end{align*}
\end{proposition}
\noindent From (\ref{eq_interpretations}) and Proposition \ref{Prop_MLE_delta2}, it can be seen that the MLE \,$\hat{\delta}$ is an unbiased estimator of $\delta$. In the proposition below, we provide an explicit expression of the variance of $\hat{\delta}$, which is crucial in the subsequent discussions on sample size determination for detecting the differential treatment effect and constructing a statistical test of significance. 

\begin{proposition}\label{Prop_var_delta2}
{\rm The variance of $\hat{\delta}$ is given by}
\begin{align*} 
    \begin{split}
        \Var(\hat{\delta}) \;&=\; \frac{4}{N_3 \,n\, N_1}\,\left(\, \sigma_e^2 \,+\, N_1\,\sigma_1^2 \,+\, n\,N_1\,\sigma_2^2 \,\right)\,.
    \end{split}
\end{align*}    
\end{proposition}
\noindent 
Similar to Section \ref{subsec:subg1_theory}, we would now introduce some sum of squares quantities that will be used in the estimation of the variance components in our subsequent discussions. Define,
\begin{align}\label{SS_level2}
\begin{split}
    SS_0 \;&=\; N_1 \, \dis \sum_{i=1}^{2N_3} \sum_{g=1}^2 \sum_{j=1}^n \left( \bar{Y}_{igj\cdot} - \bar{Y}_{\cdot \cdot \cdot \cdot} \right)^2\,, \qquad 
    SS_1 \;=\; N_1 \, \dis \sum_{i=1}^{2N_3} \sum_{g=1}^2 \sum_{j=1}^n \left( \bar{Y}_{igj\cdot} - \bar{Y}_{ig\cdot \cdot} \right)^2  \,,\\
    SS_2 \;&=\; n N_1\, \Big[ \dis \sum_{i=1}^{N_3} \left( (\bar{Y}_{i1\cdot \cdot} - \bar{Y}_{i2\cdot\cdot}) \,-\, (\bar{Y}^T_{\cdot 1 \cdot \cdot} - \bar{Y}^T_{\cdot 2 \cdot \cdot}) \right)^2 \;+\; \sum_{i=N_3+1}^{2N_3} \left( (\bar{Y}_{i1\cdot\cdot} - \bar{Y}_{i2\cdot\cdot}) \,-\, (\bar{Y}^C_{\cdot 1\cdot \cdot} - \bar{Y}^C_{\cdot 2 \cdot \cdot}) \right)^2 \Big]\,,\\
    \textrm{and} \;\;\;\; SS_3 \;&=\; 2n\,N_1 \, \Big[ \dis \sum_{i=1}^{N_3} \left(\bar{Y}_{i \cdot \cdot \cdot} - \bar{Y}^T_{\cdot \cdot \cdot \cdot}\right)^2 \,+\, \sum_{i=N_3+1}^{2N_3} \left(\bar{Y}_{i \cdot \cdot \cdot} - \bar{Y}^C_{\cdot \cdot \cdot \cdot}\right)^2 \Big]\,,
    \end{split}
\end{align}
where 
\begin{align*}
  &\bar{Y}_{igj\cdot} = \frac{1}{N_1} \dis \sum_{k=1}^{N_1} Y_{igjk} \,, \qquad \bar{Y}_{ig\cdot \cdot} := \frac{1}{nN_1} \dis \sum_{j=1}^n \sum_{k=1}^{N_1} Y_{igjk}\,, \qquad \bar{Y}^T_{\cdot g \cdot \cdot} := \frac{1}{N_3\,n\, N_1} \sum_{i=1}^{N_3} \sum_{j=1}^{n} \sum_{k=1}^{N_1} Y_{igjk}\,, \\
  &\bar{Y}^C_{\cdot g \cdot \cdot} = \frac{1}{N_3\,n\, N_1} \sum_{i=N_3+1}^{2N_3} \sum_{j=1}^{n} \sum_{k=1}^{N_1} Y_{igjk}\,, \qquad \bar{Y}_{i\cdot \cdot \cdot} = \frac{1}{2n\,N_1} \sum_{g=1}^2 \sum_{j=1}^{n} \sum_{k=1}^{N_3} Y_{igjk}\,,\\ &\bar{Y}^T_{\cdot \cdot \cdot \cdot} = \frac{1}{N_3} \sum_{i=1}^{N_3}\bar{Y}_{i\cdot \cdot \cdot}\;\;\; \textrm{and} \;\;\; \bar{Y}^C_{\cdot \cdot \cdot \cdot} = \frac{1}{N_3} \sum_{i=N_3+1}^{2N_3}\bar{Y}_{i\cdot \cdot \cdot}\,,
\end{align*}
for\, $i=1,\dots,2N_3$,\, $j=1,\dots,n$,\, and\, $g=1,2$.
Note that in order to estimate $\Var(\hat{\delta})$, we need to estimate all the variance components $\sigma_3^2,\, \sigma_2^2, \,\sigma_1^2$ and $\sigma_e^2$ in (\ref{cov_level2}). 
The following theorem ensures the existence of explicit, closed form, maximum likelihood estimates of the variance components in the balanced mixed effects model presented in (\ref{Model:Subg2}), with $N_{21} = N_{22} = n$. We also present the explicit forms of the MLEs of the variance components, that are instrumental towards constructing a test statistic to test for $H_{0} : \delta = 0$ \,vs\, $H_{a} : \delta \neq 0$. 

\begin{theorem}\label{Thm_MLE_var_comps2}
Explicit, closed form, maximum likelihood estimates of the variance components exist for the balanced case of the mixed effects model presented in (\ref{Model:Subg2}), and are given by
\begin{align*}
    &\hat{\sigma}_e^2 \;=\; \frac{1}{2N_3 \,2n\,(N_1-1)} \; SS_0\,,\\
    &\hat{\sigma}_1^2 \;=\; \frac{1}{N_1}\,\left( \frac{1}{2N_3\,2\,(n-1)} \; SS_1 \;-\; \frac{1}{2N_3 \,2n\,(N_1-1)} \; SS_0\right)\,,\\
    &\hat{\sigma}_2^2 \;=\; \frac{1}{n N_1}\,\left( \frac{1}{2N_3\,2} \; SS_2 \;-\; \frac{1}{2N_3\,2\,(n-1)} \; SS_1\right)\,,\\
    \textrm{and} \qquad &\hat{\sigma}_3^2 \;=\; \frac{1}{2n\,N_1}\,\left( \frac{1}{2N_3} \; SS_3 \;-\; \frac{1}{2N_3\,2} \; SS_2\right)\,,
\end{align*} 
where the sum of squares $SS_0, SS_1, SS_2$ and $SS_3$ are given in (\ref{SS_level2}).
\end{theorem}

In the next theorem, we establish the consistency of the MLEs of the variance components.

\begin{theorem}\label{Thm_consistency_level2}
 The maximum likelihood estimators $\hat{\sigma}_e^2$, $\hat{\sigma}_1^2$, $\hat{\sigma}_2^2$ and $\hat{\sigma}_3^2$ are unbiased estimators of the respective true parameters $\sigma_e^2$, $\sigma_1^2$, $\sigma_2^2$ and $\sigma_3^2$. Moreover, as\, $N_3, n, N_1 \to \infty$, $\Var(\hat{\sigma}_e^2)$, $\Var(\hat{\sigma}_1^2)$, $\Var(\hat{\sigma}_2^2)$ and $\Var(\hat{\sigma}_3^2)$ are $o(1)$ quantities. Together these imply that, as\, $N_3, n, N_1 \to \infty$,
 \begin{align*}
   \hat{\sigma}_e^2 \;\overset{P}{\to}\; \sigma_e^2\,, \;\; \hat{\sigma}_1^2  \;\overset{P}{\to}\;  \sigma_1^2\,,\;\; \hat{\sigma}_2^2 \;\overset{P}{\to}\; \sigma_2^2\, \;\; \textrm{and} \;\;\; \hat{\sigma}_3^2 \;\overset{P}{\to}\; \sigma_3^2\,.
 \end{align*}
\end{theorem}
Thus, we can estimate $\Var(\hat{\delta})$ given in Proposition \ref{Prop_var_delta2} by plugging in the MLEs of the variance components, and the resulting variance estimator, denoted by $\widehat{\Var(\hat{\delta})}$, can be shown to be consistent. With all the above developments, we define the statistic $T$ by standardizing $\hat{\delta}$ as follows:
\begin{align}\label{test_stat_2}
    T \;:=\; \frac{\hat{\delta} - \delta}{\sqrt{\widehat{\Var(\hat{\delta})}}}\,.
\end{align}
The consistency properties of the variance component estimators turn out to be instrumental in the following theorem, where we study the asymptotic behavior of the statistic $T$.
\begin{theorem}\label{Thm_asymp_T_level2}
   As $N_3, n, N_1 \to \infty$, we have \,$T \overset{d}{\to} N(0,1)$. 
\end{theorem}
Clearly, under the null, the statistic $T$ boils down to $T_0 := \frac{\hat{\delta}}{\sqrt{\widehat{\Var(\hat{\delta})}}}$, which we consider as the test statistic to test for $H_{0} : \delta = 0$ against the alternative $H_{a} : \delta \neq 0$. As a consequence of Theorem \ref{Thm_asymp_T_level2}, we have that under $H_{0}$, the test statistic $T$ asymptotically follows a standard normal distribution. Thus, to test for $H_{0}$ \,against\, $H_{a} : \delta \neq 0$ at level $\alpha \in (0,1)$, the rejection rule can be proposed as follows: reject $H_0$ if\, $\vert T_0 \vert > z_{1-\alpha/2}$. 

We now turn to the design aspect and determining sample size requirements for cluster randomized trials with three level hierarchical data to detect a differential subgroup effect.  Consider the statistic $\tilde{T}=\frac{\hat{\delta}-\delta}{\sqrt{\Var(\hat{\delta})}}$, where the variance components $\sigma_3^2,\, \sigma_2^2, \,\sigma_1^2$ and $\sigma_e^2$ are assumed to be known. Then, we have $\tilde{T} \sim N(0,1)$. Note that this is true even for finite samples, when the variance components are known. When the variance components are unknown, we need to replace $\Var(\hat{\delta})$ in the denominator of $\tilde{T}$ by $\widehat{\Var(\hat{\delta})}$ to get $T$, which is shown to asymptotically follow a standard normal distribution in Theorem \ref{Thm_asymp_T_level2}.

At a fixed level $\alpha \in (0,1)$, to test for the null hypothesis $H_{0} : \delta = 0$ against $H_{a} : \delta \neq 0$ based on $\tilde{T}$, that is, assuming the variance components are known, we reject $H_0$ if\, $\vert \tilde{T}_0 \vert := \left\vert \frac{\hat{\delta}}{\sqrt{\Var(\hat{\delta})}} \right\vert > z_{1-\alpha/2}$. 
Similar to the discussions presented in Section \ref{subsec:subg1_theory}, in order to achieve the power $\phi=1-\beta$, the sample sizes need to satisfy 
\begin{align}\label{power_ind42}
    \frac{\delta^2}{\Var(\hat{\delta})} \;\geq\; \left(z_{1-\alpha/2} \,+\, z_{1-\beta}  \right)^2\,.
\end{align}
Assuming $n$ and $N_3$ (that is, the number of level two and level three units) are known, for a fixed set of parameter values and for a given power $\phi=1-\beta$, the problem reduces to finding the smallest $N_1$ that satisfies (\ref{power_ind42}).

\begin{proposition}\label{Prop_samp_size2}
{\rm
    Assuming $n, N_3$ and the parameters $\delta$, $\sigma_3^2,\, \sigma_2^2, \,\sigma_1^2$ and $\sigma_e^2$ are known, the number of level one units needed to achieve the power $\phi = 1-\beta$\, is the smallest integer greater than
    \begin{align*}
        N_1 \;=\; A'\,\left( \frac{\delta^2}{\big(z_{1-\alpha/2} \,+\, z_{1-\beta} \big)^2} \,-\, B'\, \right)^{-1}, \qquad \textrm{where} \;\;\; A' &= \frac{4\,\sigma_e^2}{N_3\,n}  \;\;\;  \textrm{and} \;\; B'= \frac{4\,\left(\sigma_1^2 + n\,\sigma_2^2\right)}{N_3\,n}\,.
    \end{align*}
}
\end{proposition}

\begin{remark}\label{remark_next_1}
    {\rm
    If we happen to reject the null hypothesis $H_{0} : \delta = 0$ against $H_{a} : \delta \neq 0$, that is, if the differential treatment effect turns out to be statistically significant, the subsequent problem of interest might be to test for the individual treatment effects on the two subgroups. That is, the problem of interest might then be to test for $H_{0g} : \delta_g = 0$ against $H_{ag} : \delta_g \neq 0$, for\, $g=1,2$. This is not the primary goal of our work; however, for the sake of completeness, we present in Section \ref{appendix1} in the supplementary materials explicit expressions for sample size requirements to detect the individual treatment effects, when subgrouping is done at level one and level two.
    }
\end{remark}

\begin{remark}\label{remark_next_2}
    {\rm
    If, however, we fail to reject $H_{0} : \delta = 0$, that is, if there is no significant difference between the treatment effects on the two subgroups, we may disregard the subgroup information from our model and turn to testing for the treatment effect on the overall population. The models we presented in (\ref{Model:Subg1}) and (\ref{Model:Subg2}) would then reduce to 
\begin{equation}\label{Model_no_subgrp}
    Y_{ijk} \;=\; \beta_0 \,+\, \xi\,X_i \,+\, u_i \,+\, u_{j(i)} \,+\, \epsilon_{ijk}\,,
\end{equation}
and \, $\xi = E(Y_{ijk}\,|\, X_i=1) \,-\,  E(Y_{ijk}\,|\, X_i=0)$ would become the sole parameter of interest. Thus the model in (\ref{Model_no_subgrp}) is a special case of our more general models when there are no subgroups. This model has been studied in details by \cite{heo_leon_2008}, where the authors have derived closed form expressions for the sample size requirements to detect the treatment effect on the responses or outcomes. We, therefore, skip discussions on testing for the treatment effect on the overall population in our paper, and simply refer the reader to \cite{heo_leon_2008}.
    }
\end{remark}

\begin{remark}\label{remark_3}
 {\rm
We have derived in Theorems \ref{Thm_MLE_var_comps1} and \ref{Thm_MLE_var_comps2} explicit closed form expressions for the MLEs of the variance components of the mixed effects models presented in (\ref{Model:Subg1}) and (\ref{Model:Subg2}), respectively. However, the estimates of the different variance components, other than the error variance, are not guaranteed to be non-negative. In case one or more of the variance component estimators turn out to be negative in practice, we replace any negative value with a zero, which is tantamount to dropping the corresponding random effect from the model.
 }
\end{remark}

\section{Simulation Studies}\label{sec:simulation}
In this section, we conduct extensive simulation studies to examine the empirical size and power of our proposed test for significance of the differential treatment effect over synthetic data. The parameter settings we considered are illustrated below. The significance level is chosen to be $\alpha=0.05$. We fix $\beta_0 = 0$, $\tau=0$, $\xi = 0.5$, $\sigma=1$, 
and consider two different settings for the variance components:
\begin{itemize}
    \item[I.] $\rho_1=0.2$, $\rho_{(1)} = \rho_{(2)} = 0.15$ \,and\, $\rho_2 = 0.1$, and
    \item[II.] $\rho_1=0.1$, $\rho_{(1)} = \rho_{(2)} = 0.075$\, and\, $\rho_2 = 0.05$.
\end{itemize}
We vary $\delta$ over $\{0, 0.5, 1\}$, and present in the following two subsections our simulation results corresponding to the cases when the subgrouping is done at level one and level two, respectively. 

\subsection{Subgrouping at level one}\label{sec:sim1}

In this case, we calculate the variance components $\sigma_3^2, \sigma_2^2, \sigma_{\text{grp}}^2$ and $\sigma_e^2$ using (\ref{cov_level1}) as $\sigma_3^2 = \sigma^2 \rho_2$, $\sigma_{\text{grp}}^2 = \sigma^2 \rho_{(2)} - \sigma_3^2$, $\sigma_2^2 = \sigma^2 \rho_{(1)} - \sigma_3^2$\, and \,$\sigma_e^2 = \sigma^2 - \sigma^2 \rho_1$. The formal steps of the simulation procedure are presented below:

\begin{enumerate}
    \item With the given parameter settings as above and specified values of $N_3, N_2$ and $n$, we generate the random intercepts as follows:
    \begin{enumerate}[]
        \item $u_i$ is generated independently from $N(0, \sigma_3^2)$,\, for \,$i=1, \dots, 2N_3$.
        \item For each $i$, generate $u_{j(i)}$ independently from $N(0, \sigma_2^2)$,\, for \,$j=1, \dots, N_2$.
        \item For each $i$ and $j$, generate $u_{g(ij)}$ independently from $N(0, \sigma_{\text{grp}}^2)$,\, for \,$g=1,2$.
        \item Finally, for each combination of $i, j$ and $g$, generate $\epsilon_{ijgk}$ independently from $N(0, \sigma_e^2)$,\, for \,$k=1, \dots, n$.
    \end{enumerate}
    
    \item Generate the response variable $Y_{ijgk}$ from model (\ref{Model:Subg1}) using the above specifications of the fixed effect parameters and the random effects generated in steps 1(a--d).
    \item Estimate the parameter of interest, $\delta$, 
    using Proposition \ref{Prop_MLE_delta1}, and compute the test statistic $T_0$.
    \item For each combination of $(N_3, N_2, n)$, repeat the steps 1--3, 1000 times to get 1000 values of the test statistic, denoted by $\{ T_{0a} \}_{a=1}^{1000}$.
    \item Calculate the empirical power ($\tilde{\phi}$) as the proportion of times (out of 1000) we reject the null hypothesis $H_{0} : \delta = 0$ against $H_{a} : \delta \neq 0$, that is,
    \begin{align*}
        \tilde{\phi} \;:=\; \frac{1}{1000} \dis \sum_{a=1}^{1000} \mathbbm{1}\left(\vert T_{0a} \vert > z_{1-\alpha/2} \right)\,.
    \end{align*}
\end{enumerate}

\noindent We present in Table \ref{tab1} below the empirical power out of 1000 simulation runs, obtained by following steps (1--5) above, for different combinations of $N_3, N_2$ and $n$, and, for the two settings of the variance components (I and II). Note that the entries corresponding to $\delta=0$ are essentially the Type-I error rates. Overall from the results presented in Table \ref{tab1}, we observe that:
\begin{itemize}
    \item The Type-I error rates are all below the nominal 5\% level.
    \item For the same $(N_3, N_2, n)$ combination, the empirical power increases with increase in $\delta$.
    \item For the same $(N_3, N_2)$ combination and a fixed $\delta$, increase in $n$ increases the empirical power and reduces the Type-I error rate.
    \item For a fixed $n$ and $\delta$, increase in $N_3$ and/or $N_2$ increases the empirical power and reduces the Type-I error rate.
\end{itemize}

\begin{table}[!ht]
	\centering
	\caption{Empirical power of the proposed test for different combinations of $(N_3, N_2, n)$ over 1000 simulation runs for subgrouping at level one.}
	\label{tab1}
	\begin{tabular}{ccc  c | c  c }
\toprule
&&&& \multicolumn{2}{c}{Empirical Power ($\tilde{\phi}$) }
\\ \cmidrule(r){5-6}
$N_3$ & $N_2$ & $n$ & $\delta$ & I & II\\
\hline
5 & 6 & 15 & 0 & 0.002 & 0.015  \\
& & & 0.5 & 0.756 & 0.941 \\
& & & 1 &  1 & 1 \\
\hline
5 & 6 & 20 & 0 & 0.002 & 0.009  \\
& & & 0.5 & 0.789 & 0.956 \\
& & & 1 &  1 & 1 \\
\hline
5 & 6 & 25 & 0 & 0.001 & 0.004  \\
& & & 0.5 & 0.807 & 0.976 \\
& & & 1 &  1 & 1 \\
\hline
5 & 7 & 15 & 0 & 0.001 & 0.011 \\
& & & 0.5 &  0.787 & 0.954 \\
& & & 1 &  1  & 1\\
\hline
5 & 7 & 20 & 0 & 0.001 & 0.001 \\
& & & 0.5 &  0.799 & 0.974 \\
& & & 1 &  1  & 1\\
\hline
5 & 7 & 25 & 0 & 0.001 & 0.006 \\
& & & 0.5 &  0.824 & 0.982 \\
& & & 1 &  1  & 1\\
\hline
6 & 7 & 15 & 0 & 0.003 & 0.006 \\
& & & 0.5 &  0.877 & 0.978 \\
& & & 1 &  1 & 1 \\
\hline
6 & 7 & 20 & 0 & 0 & 0.003 \\
& & & 0.5 &  0.901 & 0.991 \\
& & & 1 &  1 & 1 \\
\hline
6 & 7 & 25 & 0 & 0 & 0.002 \\
& & & 0.5 &  0.931 & 0.995 \\
& & & 1 &  1 & 1 \\
\hline
6 & 8 & 15 & 0 & 0.002 & 0.004 \\
& & & 0.5 &  0.87 & 0.99 \\
& & & 1 & 1 & 1 \\
\hline
6 & 8 & 20 & 0 & 0 & 0 \\
& & & 0.5 &  0.912 & 0.994 \\
& & & 1 & 1 & 1 \\
\hline
6 & 8 & 25 & 0 & 0 & 0.001 \\
& & & 0.5 &  0.933 & 0.995 \\
& & & 1 & 1 & 1 \\
		\toprule
	\end{tabular}
	\\
\end{table}

\subsection{Subgrouping at level two}\label{sec:sim2}

In this case, we calculate the variance components $\sigma_3^2, \sigma_2^2, \sigma_1^2$ and $\sigma_e^2$ using (\ref{cov_level2}) as $\sigma_3^2 = \sigma^2 \rho_2$, $\sigma_2^2 = \sigma^2 \rho_{(2)} - \sigma_3^2$, $\sigma_1^2 = \sigma^2 \rho_1 - \sigma_3^2 - \sigma_2^2$\, and \,$\sigma_e^2 = \sigma^2 - \sigma^2 \rho_1$. The formal steps of the simulation procedure are presented below, which are essentially similar to steps (1--5) presented in the previous subsection:

\begin{enumerate}
    \item[1*.] With the given parameter settings and specified values of $N_3, n$ and $N_1$, we generate the random intercepts as follows:
    \begin{enumerate}[]
        \item $u_i$ is generated independently from $N(0, \sigma_3^2)$,\, for \,$i=1, \dots, 2N_3$.
        \item For each $i$, generate $u_{g(i)}$ independently from $N(0, \sigma_2^2)$,\, for \,$g=1, 2$.
        \item For each $i$ and $g$, generate $u_{j(ig)}$ independently from $N(0, \sigma_1^2)$,\, for \,$j=1, \dots, n$.
        \item Finally, for each combination of $i, g$ and $j$, generate $\epsilon_{igjk}$ independently from $N(0, \sigma_e^2)$,\, for \,$k=1, \dots, N_1$.
    \end{enumerate}
    
    \item[2*.] Generate the response variable $Y_{igjk}$ from model (\ref{Model:Subg2}) using the above specifications of the fixed effect parameters and the random effects generated in steps 1*(a--d).
    \item[3*.] Estimate the parameter of interest, $\delta$, 
    using Proposition \ref{Prop_MLE_delta2}, and compute the test statistic $T_0$.
    \item[4*.] For each combination of $(N_3, n, N_1)$, repeat the steps 1*--3* 1000 times to get 1000 values of the test statistic, denoted by $\{ T_{0a} \}_{a=1}^{1000}$.
    \item[5*.] Calculate the empirical power ($\tilde{\phi}$) similar to Step 5 in Section \ref{sec:sim1}.
\end{enumerate}

\noindent We present in Table \ref{tab2} below the empirical power out of 1000 simulation runs, obtained by following steps (1*--5*) above, for different combinations of $N_3, n$ and $N_1$, and, for the two different settings for the variance components.
 It is to be noted that we need higher sample sizes to control the Type-I error in case of subgrouping at level two. A similar observation is also made in \cite{wang2023} when reporting power. Overall from the results presented in Table \ref{tab2}, we observe that:
\begin{itemize}
    \item The Type-I error rates are controlled at the nominal 5\% level.
    \item For the same $(N_3, n, N_1)$ combination, the empirical power increases with increase in $\delta$.
    \item For the same $(N_3, n)$ combination and a fixed $\delta$, increase in $N_1$ increases the empirical power and reduces the Type-I error rate.
    \item For a fixed $N_1$ and $\delta$, increase in $N_3$ and/or $n$ increases the empirical power and reduces the Type-I error rate.
\end{itemize}

\begin{table}[!ht]
	\centering
	\caption{Empirical power of the proposed test for different combinations of $(N_3, n, N_1)$ over 1000 simulation runs subgrouping at level two.}
	\label{tab2}
	\begin{tabular}{ccc  c | c  c}
\toprule
&&&& \multicolumn{2}{c}{Empirical Power ($\tilde{\phi}$) }
\\ \cmidrule(r){5-6}
$N_3$ & $n$ & $N_1$ & $\delta$ & I & II\\
\hline
10 & 15 & 20 & 0 & 0.076 & 0.079 \\
& & & 0.5 & 0.925 & 0.997 \\
& & & 1 &  1 & 1 \\
\hline
10 & 15 & 30 & 0 & 0.072 & 0.069 \\
& & & 0.5 & 0.923 & 0.996  \\
& & & 1 &  1 & 1 \\
\hline
10 & 15 & 40 & 0 & 0.078 & 0.075 \\
& & & 0.5 & 0.928 & 0.998 \\
& & & 1 &  1 & 1 \\
\hline
20 & 15 & 20 & 0 & 0.069 & 0.068 \\
& & & 0.5 & 0.997 & 1 \\
& & & 1 &  1 & 1 \\
\hline
20 & 15 & 30 & 0 & 0.067 & 0.065 \\
& & & 0.5 &  0.996 & 1 \\
& & & 1 &  1 & 1 \\
\hline
20 & 15 & 40 & 0 & 0.062 & 0.062 \\
& & & 0.5 & 0.998 & 1 \\
& & & 1 &  1 & 1 \\
\hline
40 & 20 & 20 & 0 & 0.054 & 0.049 \\
& & & 0.5 & 1 & 1 \\
& & & 1 &  1 & 1 \\
\hline
40 & 20 & 30 & 0 & 0.064 & 0.064 \\
& & & 0.5 &  1 & 1 \\
& & & 1 &  1 & 1 \\
\hline
40 & 20 & 40 & 0 & 0.058 & 0.058 \\
& & & 0.5 & 1 & 1 \\
& & & 1 &  1 & 1 \\
\hline
60 & 20 & 20 & 0 & 0.042 & 0.045 \\
& & & 0.5 & 1 & 1 \\
& & & 1 &  1 & 1 \\
\hline
60 & 20 & 30 & 0 & 0.042 & 0.036 \\
& & & 0.5 & 1 & 1  \\
& & & 1 &  1 & 1 \\
\hline
60 & 20 & 40 & 0 & 0.051 & 0.051 \\
& & & 0.5 & 1 & 1 \\
& & & 1 &  1 & 1 \\
		\toprule
	\end{tabular}
	\\
\end{table}

\section{Motivating Data Application}\label{sec:real}
{We now return to the motivating real-world application introduced in Section \ref{intro:motiv_eg}.} To reverse the escalating rates of HIV in the Bahamas, the government embarked on an inter-agency approach, targeting children and adolescents. An evidence-based program was used through the government schools \citep{Deveaux}, and the effect of the intervention (FOYC + CImPACT) from the Centers for Disease Control Diffusion of Effective Behavioral Interventions portfolio, was studied on the reduction of risk-taking behaviours related to HIV/STI transmission and teen pregnancy. Woven throughout, FOYC is a decision-making model that provides guidance and practice in problem-solving with a focus on how to obtain factual information on sexual health. CImPACT is a single session including a 24-minute educational video filmed in the Bahamas, teaching effective communication and listening strategies related to difficult topics and safe-sex, followed by two role-plays for the parent and youth, a discussion, and a condom demonstration. This study has been ongoing for more than 15 years, funded by three separate National Institute of Health grants, and has yielded a rich dataset resulting from multiple cluster randomized trial implementations \citep{wang2020, Stanton2, Stanton1}. For this paper, we have selected a cross-sectional subset of the data from a three-level CRT, with students at level one, teachers at level two, and schools at level three, and the randomization is done at the school level. The treatments are coded as ‘1’ for the intervention and ‘0’ for the control. The data we used consisted of 10 middle schools, with 5 schools randomized in the control arm and the remaining in the treatment arm. Each school has three to five classes, with four schools having five classes and two schools having three classes. The number of students per school varies from 72 to 118 and the total number of students is 933. Since we have used only cross-sectional observations for this analysis, no missing data is reported. Condom use self-efficacy knowledge is the primary student level outcome variable in the present analysis.  It was assessed using a six-item scale. Agreement was measured through a five-point Likert scale (1=strongly disagree to 5=strongly agree). The internal consistency estimate of the scale was 0.88. A composite score was calculated as a mean score across the six items (range 1 to 5). 

Using the methods developed in this paper, we performed two separate subgroup analysis. For level 1 subgrouping, we assessed whether gender (Male vs. Female) has any differential subgrouping effect. Similarly for level 2 we created two types of class size using the mode of the class size ($25$) as a cut to create Large vs. Small class. The goal is to assess whether class size has any differential effect on the outcome. Baseline observation for each student is also used to adjust for the initial knowledge. For the level 1 subgrouping, Table \ref{Data table1} gives the estimates of the fixed effects along with the $p$-values of the tests. From the table we observe that at the significance level of $0.05$, the intervention effect and gender-specific main effect are significant, however, gender by intervention interaction effects is non-significant, indicating the possible absence of any subgrouping by gender. 
For level 2 subgrouping, Table \ref{Data table2} gives the estimates of the fixed effects along with the $p$-values of the tests. At the significance level of $0.05$, the intervention effect is significant; however, the class-size specific main effect is not significant. The class-size by intervention interaction effects is significant, indicating possible subgrouping by class size. The effect size is negative ($-0.3405$), indicating the Condom use self-efficacy knowledge permeates better in a smaller class room than a larger one, which is consistent with the general expectation from an education based HIV prevention program.

\begin{table}[!ht]
	\centering
	\caption{Fixed Effects Estimates and $p$-values of the tests} \label{Data table1}
	\begin{tabular}{ccccc}
		\hline 
		Effect & Estimate & Standard Error & $t$-value & $p$-value\tabularnewline
		\hline
		Intercept & 2.82 & 0.0928 & 30.35 & $<$.0001\tabularnewline
		Intervention & 0.3139 & 0.1045 & 3.00 & $<$.0001\tabularnewline
		Gender & 0.4839 & 0.09468 & 5.11 & $<$.0001\tabularnewline
		Intervention{*}Gender & 0.0454 & 0.1375 & 0.33 & 0.7411\tabularnewline
		\hline 
	\end{tabular}
\end{table}

\begin{table}[!ht]
	\centering
	\caption{Fixed Effects Estimates and $p$-values of the tests} \label{Data table2}
	\begin{tabular}{ccccc}
		\hline 
		Effect & Estimate & Standard Error & $t$-value & $p$-value\tabularnewline
		\hline
		Intercept & 2.8975 & 0.09006 & 32.17 & $<$.0001\tabularnewline
		Intervention & 0.5070 & 0.09071 & 5.59 & $<$.0001 \tabularnewline
		Class Size & 0.1711 & 0.0918 & 1.86 & 0.1213 \tabularnewline
		Intervention{*}Class Size & -0.3405 & 0.1414 & -2.41 & 0.0162\tabularnewline
		\hline 
	\end{tabular}
\end{table}
\section{Conclusions}\label{concl}
The heterogeneity of the population is inherent and may not be completely avoidable despite the idealistic and selective (e.g. strict inclusion-exclusion criteria) nature of clinical trial design. With the rise of pragmatic trials and the inherent hierarchical nature of population heterogeneity, patients are better modeled when a subgroup can be identified that is more similar to them than a large, fairly general population. Hence, understanding the treatment response for a diverse and stratified population is important. However, assessing such an analysis of subgroup effects is mostly treated with skepticism and is mostly exploratory in nature. Furthermore, almost all subgroup analysis development is in the context of simple RCTs. Yet, with the emergence of more and more CRTs, subgroups can emerge at various cluster levels. In this paper, we consider subgroup analysis for CRTs when randomization is carried out at the highest level and subgroups can exist at various levels \citep{fazzari2014sample}. We develop a consistent test for the differential intervention effect between the two subgroups at different levels of the hierarchy, which is the key methodological contribution. We also developed sample size and power formulae in the case of a planned confirmatory subgroup analysis. {Finally, we applied the methodology on the motivating real-world dataset arising from a three-level CRT related to HIV prevention in the Bahamas.} 
{Using the methodology we developed}, we performed two subgroup analyses with mixed findings. While the subgroup at the lowest level, i.e. at the student level, was non-significant, the subgroup at the class level was indeed significant. Since some of the subgroups are modifiable, these findings point to future directions of preventative education, albeit with cost and feasibility in mind.

Since this work represents an initial exploration of subgroup analysis for CRTs, significant future research opportunities lie ahead. For instance, in multi-layered CRTs, the randomization unit may vary, and the trial's cross-sectional versus longitudinal nature can also impact subgroup modeling and underlying assumptions \citep{heo2009sample, heo2013sample}. Additionally, there is an important direction for generalizing our results beyond two subgroups as well as concurrent existence of subgroups at more than one level and possible ramification of interactions (if any) between them. Furthermore, there is considerable potential for further development of hierarchical Bayes methods for subgroup analysis, primarily applied to RCTs \citep{ jones2011bayesian,hsu2019hierarchical}, but adaptable for CRTs. We are currently working on extending our research to address some of these future challenges, particularly focusing on clusters of longitudinal CRTs and stepped-wedge CRTs.

\bibliographystyle{imsart-nameyear} 
\bibliography{paper-ref}

\begin{thebibliography}{32}

\bibitem[\protect\citeauthoryear{Ahn, Heo and Zhang}{2014}]{ahn2014sample}
\begin{bbook}[author]
\bauthor{\bsnm{Ahn},~\bfnm{Chul}\binits{C.}}, \bauthor{\bsnm{Heo},~\bfnm{Moonseoung}\binits{M.}} \AND \bauthor{\bsnm{Zhang},~\bfnm{Song}\binits{S.}}
(\byear{2014}).
\btitle{Sample size calculations for clustered and longitudinal outcomes in clinical research}.
\bpublisher{CRC Press}.
\end{bbook}
\endbibitem

\bibitem[\protect\citeauthoryear{Burns and Santos}{1995}]{burns1995}
\begin{barticle}[author]
\bauthor{\bsnm{Burns},~\bfnm{Barbara~J}\binits{B.~J.}} \AND \bauthor{\bsnm{Santos},~\bfnm{Alberto~B}\binits{A.~B.}}
(\byear{1995}).
\btitle{Assertive community treatment: an update of randomized trials.}
\bjournal{Psychiatric Services (Washington, DC)}
\bvolume{46}
\bpages{669--675}.
\end{barticle}
\endbibitem

\bibitem[\protect\citeauthoryear{Cui et~al.}{2002}]{cui2002}
\begin{barticle}[author]
\bauthor{\bsnm{Cui},~\bfnm{Lu}\binits{L.}}, \bauthor{\bsnm{James~Hung},~\bfnm{HM}\binits{H.}}, \bauthor{\bsnm{Wang},~\bfnm{Sue~Jane}\binits{S.~J.}} \AND \bauthor{\bsnm{Tsong},~\bfnm{Yi}\binits{Y.}}
(\byear{2002}).
\btitle{Issues related to subgroup analysis in clinical trials}.
\bjournal{Journal of biopharmaceutical statistics}
\bvolume{12}
\bpages{347--358}.
\end{barticle}
\endbibitem

\bibitem[\protect\citeauthoryear{Cunningham and Johnson}{2016}]{cunningham2016}
\begin{barticle}[author]
\bauthor{\bsnm{Cunningham},~\bfnm{Tina~D}\binits{T.~D.}} \AND \bauthor{\bsnm{Johnson},~\bfnm{Robert~E}\binits{R.~E.}}
(\byear{2016}).
\btitle{Design effects for sample size computation in three-level designs}.
\bjournal{Statistical Methods in Medical Research}
\bvolume{25}
\bpages{505--519}.
\end{barticle}
\endbibitem

\bibitem[\protect\citeauthoryear{Deveaux et~al.}{2011}]{Deveaux}
\begin{barticle}[author]
\bauthor{\bsnm{Deveaux},~\bfnm{L.}\binits{L.}}, \bauthor{\bsnm{Lunn},~\bfnm{S.}\binits{S.}}, \bauthor{\bsnm{Mae},~\bfnm{B.~R.}\binits{B.~R.}} \AND \bauthor{\bparticle{et} \bsnm{al}}
(\byear{2011}).
\btitle{Focus on Youth in the Caribbean: Beyond the Numbers}.
\bjournal{Journal of the International Association of Physicians in AIDS Care}
\bvolume{10}
\bpages{316-325}.
\end{barticle}
\endbibitem

\bibitem[\protect\citeauthoryear{Donner}{1998}]{Donner1998}
\begin{barticle}[author]
\bauthor{\bsnm{Donner},~\bfnm{A.}\binits{A.}}
(\byear{1998}).
\btitle{Some aspects of the design and analysis of cluster randomization trials}.
\bjournal{Applied Statistics}
\bvolume{47}
\bpages{95-113}.
\end{barticle}
\endbibitem

\bibitem[\protect\citeauthoryear{Fazzari, Kim and Heo}{2014}]{fazzari2014sample}
\begin{barticle}[author]
\bauthor{\bsnm{Fazzari},~\bfnm{Melissa~J}\binits{M.~J.}}, \bauthor{\bsnm{Kim},~\bfnm{Mimi~Y}\binits{M.~Y.}} \AND \bauthor{\bsnm{Heo},~\bfnm{Moonseong}\binits{M.}}
(\byear{2014}).
\btitle{Sample size determination for three-level randomized clinical trials with randomization at the first or second level}.
\bjournal{Journal of Biopharmaceutical Statistics}
\bvolume{24}
\bpages{579--599}.
\end{barticle}
\endbibitem

\bibitem[\protect\citeauthoryear{Figueroa}{2014}]{figueroa2014}
\begin{barticle}[author]
\bauthor{\bsnm{Figueroa},~\bfnm{J~Peter}\binits{J.~P.}}
(\byear{2014}).
\btitle{Review of HIV in the Caribbean: significant progress and outstanding challenges}.
\bjournal{Current HIV/AIDS Reports}
\bvolume{11}
\bpages{158--167}.
\end{barticle}
\endbibitem

\bibitem[\protect\citeauthoryear{Ghosh et~al.}{2022}]{Ghosh2022}
\begin{barticle}[author]
\bauthor{\bsnm{Ghosh},~\bfnm{S.}\binits{S.}}, \bauthor{\bsnm{Mukhopadhyay},~\bfnm{S.}\binits{S.}}, \bauthor{\bsnm{Majumder},~\bfnm{P.}\binits{P.}} \AND \bauthor{\bsnm{Wang},~\bfnm{B.}\binits{B.}}
(\byear{2022}).
\btitle{Statistical power and sample size requirements to detect an intervention by time interaction in four level longitudinal cluster randomized trials}.
\bjournal{Statistics in Medicine}
\bvolume{41}
\bpages{2542-2556}.
\end{barticle}
\endbibitem

\bibitem[\protect\citeauthoryear{Hedeker and Gibbons}{2006}]{hedeker}
\begin{bbook}[author]
\bauthor{\bsnm{Hedeker},~\bfnm{Donald}\binits{D.}} \AND \bauthor{\bsnm{Gibbons},~\bfnm{Robert~D}\binits{R.~D.}}
(\byear{2006}).
\btitle{Longitudinal data analysis}
\bvolume{451}.
\bpublisher{John Wiley \& Sons}.
\end{bbook}
\endbibitem

\bibitem[\protect\citeauthoryear{Hedges and Borenstein}{2014}]{Hedges2014}
\begin{barticle}[author]
\bauthor{\bsnm{Hedges},~\bfnm{L.}\binits{L.}} \AND \bauthor{\bsnm{Borenstein},~\bfnm{M.}\binits{M.}}
(\byear{2014}).
\btitle{Conditional optimal design in three- and four-level experiments}.
\bjournal{Journal of Educational and Behavioral Statistics}
\bvolume{39}
\bpages{257-282}.
\end{barticle}
\endbibitem

\bibitem[\protect\citeauthoryear{Heo and Leon}{2008}]{heo_leon_2008}
\begin{barticle}[author]
\bauthor{\bsnm{Heo},~\bfnm{Moonseong}\binits{M.}} \AND \bauthor{\bsnm{Leon},~\bfnm{Andrew~C}\binits{A.~C.}}
(\byear{2008}).
\btitle{Statistical power and sample size requirements for three level hierarchical cluster randomized trials}.
\bjournal{Biometrics}
\bvolume{64}
\bpages{1256--1262}.
\end{barticle}
\endbibitem

\bibitem[\protect\citeauthoryear{Heo and Leon}{2009}]{heo2009sample}
\begin{barticle}[author]
\bauthor{\bsnm{Heo},~\bfnm{Moonseong}\binits{M.}} \AND \bauthor{\bsnm{Leon},~\bfnm{Andrew~C}\binits{A.~C.}}
(\byear{2009}).
\btitle{Sample size requirements to detect an intervention by time interaction in longitudinal cluster randomized clinical trials}.
\bjournal{Statistics in medicine}
\bvolume{28}
\bpages{1017--1027}.
\end{barticle}
\endbibitem

\bibitem[\protect\citeauthoryear{Heo, Xue and Kim}{2013}]{heo2013sample}
\begin{barticle}[author]
\bauthor{\bsnm{Heo},~\bfnm{Moonseong}\binits{M.}}, \bauthor{\bsnm{Xue},~\bfnm{Xiaonan}\binits{X.}} \AND \bauthor{\bsnm{Kim},~\bfnm{Mimi~Y}\binits{M.~Y.}}
(\byear{2013}).
\btitle{Sample size requirements to detect an intervention by time interaction in longitudinal cluster randomized clinical trials with random slopes}.
\bjournal{Computational statistics \& data analysis}
\bvolume{60}
\bpages{169--178}.
\end{barticle}
\endbibitem

\bibitem[\protect\citeauthoryear{Hsu, Zalkikar and Tiwari}{2019}]{hsu2019hierarchical}
\begin{barticle}[author]
\bauthor{\bsnm{Hsu},~\bfnm{Yu-Yi}\binits{Y.-Y.}}, \bauthor{\bsnm{Zalkikar},~\bfnm{Jyoti}\binits{J.}} \AND \bauthor{\bsnm{Tiwari},~\bfnm{Ram~C}\binits{R.~C.}}
(\byear{2019}).
\btitle{Hierarchical Bayes approach for subgroup analysis}.
\bjournal{Statistical Methods in Medical Research}
\bvolume{28}
\bpages{275--288}.
\end{barticle}
\endbibitem

\bibitem[\protect\citeauthoryear{Jones et~al.}{2011}]{jones2011bayesian}
\begin{barticle}[author]
\bauthor{\bsnm{Jones},~\bfnm{Hayley~E}\binits{H.~E.}}, \bauthor{\bsnm{Ohlssen},~\bfnm{David~I}\binits{D.~I.}}, \bauthor{\bsnm{Neuenschwander},~\bfnm{Beat}\binits{B.}}, \bauthor{\bsnm{Racine},~\bfnm{Amy}\binits{A.}} \AND \bauthor{\bsnm{Branson},~\bfnm{Michael}\binits{M.}}
(\byear{2011}).
\btitle{Bayesian models for subgroup analysis in clinical trials}.
\bjournal{Clinical Trials}
\bvolume{8}
\bpages{129--143}.
\end{barticle}
\endbibitem

\bibitem[\protect\citeauthoryear{Lagakos et~al.}{2006}]{lagakos2006}
\begin{barticle}[author]
\bauthor{\bsnm{Lagakos},~\bfnm{Stephen~W}\binits{S.~W.}} \betal{et~al.}
(\byear{2006}).
\btitle{The challenge of subgroup analyses-reporting without distorting}.
\bjournal{New England Journal of Medicine}
\bvolume{354}
\bpages{1667}.
\end{barticle}
\endbibitem

\bibitem[\protect\citeauthoryear{Laird and Ware}{1982}]{laird1982random}
\begin{barticle}[author]
\bauthor{\bsnm{Laird},~\bfnm{Nan~M}\binits{N.~M.}} \AND \bauthor{\bsnm{Ware},~\bfnm{James~H}\binits{J.~H.}}
(\byear{1982}).
\btitle{Random-effects models for longitudinal data}.
\bjournal{Biometrics}
\bpages{963--974}.
\end{barticle}
\endbibitem

\bibitem[\protect\citeauthoryear{Li et~al.}{2023}]{li2023}
\begin{barticle}[author]
\bauthor{\bsnm{Li},~\bfnm{Fan}\binits{F.}}, \bauthor{\bsnm{Chen},~\bfnm{Xinyuan}\binits{X.}}, \bauthor{\bsnm{Tian},~\bfnm{Zizhong}\binits{Z.}}, \bauthor{\bsnm{Esserman},~\bfnm{Denise}\binits{D.}}, \bauthor{\bsnm{Heagerty},~\bfnm{Patrick~J}\binits{P.~J.}} \AND \bauthor{\bsnm{Wang},~\bfnm{Rui}\binits{R.}}
(\byear{2023}).
\btitle{Designing three-level cluster randomized trials to assess treatment effect heterogeneity}.
\bjournal{Biostatistics}
\bvolume{24}
\bpages{833--849}.
\end{barticle}
\endbibitem

\bibitem[\protect\citeauthoryear{Lipkovich et~al.}{2024}]{lipkovich2024}
\begin{barticle}[author]
\bauthor{\bsnm{Lipkovich},~\bfnm{Ilya}\binits{I.}}, \bauthor{\bsnm{Svensson},~\bfnm{David}\binits{D.}}, \bauthor{\bsnm{Ratitch},~\bfnm{Bohdana}\binits{B.}} \AND \bauthor{\bsnm{Dmitrienko},~\bfnm{Alex}\binits{A.}}
(\byear{2024}).
\btitle{Modern approaches for evaluating treatment effect heterogeneity from clinical trials and observational data}.
\bjournal{Statistics in Medicine}
\bvolume{43}
\bpages{4388--4436}.
\end{barticle}
\endbibitem

\bibitem[\protect\citeauthoryear{Moerbeek, van Breukelen and Berger}{2000}]{Moerbeek}
\begin{barticle}[author]
\bauthor{\bsnm{Moerbeek},~\bfnm{Mirjam}\binits{M.}}, \bauthor{\bparticle{van} \bsnm{Breukelen},~\bfnm{Gerard J.~P.}\binits{G.~J.~P.}} \AND \bauthor{\bsnm{Berger},~\bfnm{Martijn P.~F.}\binits{M.~P.~F.}}
(\byear{2000}).
\btitle{Design Issues for Experiments in Multilevel Populations}.
\bjournal{Journal of Educational and Behavioral Statistics}
\bvolume{25}
\bpages{271-284}.
\end{barticle}
\endbibitem

\bibitem[\protect\citeauthoryear{Platt et~al.}{2010}]{platt2010}
\begin{barticle}[author]
\bauthor{\bsnm{Platt},~\bfnm{Richard}\binits{R.}}, \bauthor{\bsnm{Takvorian},~\bfnm{Samuel~U}\binits{S.~U.}}, \bauthor{\bsnm{Septimus},~\bfnm{Edward}\binits{E.}}, \bauthor{\bsnm{Hickok},~\bfnm{Jason}\binits{J.}}, \bauthor{\bsnm{Moody},~\bfnm{Julia}\binits{J.}}, \bauthor{\bsnm{Perlin},~\bfnm{Jonathan}\binits{J.}}, \bauthor{\bsnm{Jernigan},~\bfnm{John~A}\binits{J.~A.}}, \bauthor{\bsnm{Kleinman},~\bfnm{Ken}\binits{K.}} \AND \bauthor{\bsnm{Huang},~\bfnm{Susan~S}\binits{S.~S.}}
(\byear{2010}).
\btitle{Cluster randomized trials in comparative effectiveness research: randomizing hospitals to test methods for prevention of healthcare-associated infections}.
\bjournal{Medical Care}
\bpages{S52--S57}.
\end{barticle}
\endbibitem

\bibitem[\protect\citeauthoryear{Pocock et~al.}{2002}]{pocock2002}
\begin{barticle}[author]
\bauthor{\bsnm{Pocock},~\bfnm{Stuart~J}\binits{S.~J.}}, \bauthor{\bsnm{Assmann},~\bfnm{Susan~E}\binits{S.~E.}}, \bauthor{\bsnm{Enos},~\bfnm{Laura~E}\binits{L.~E.}} \AND \bauthor{\bsnm{Kasten},~\bfnm{Linda~E}\binits{L.~E.}}
(\byear{2002}).
\btitle{Subgroup analysis, covariate adjustment and baseline comparisons in clinical trial reporting: current practiceand problems}.
\bjournal{Statistics in medicine}
\bvolume{21}
\bpages{2917--2930}.
\end{barticle}
\endbibitem

\bibitem[\protect\citeauthoryear{Searle, Casella and McCulloch}{2009}]{searle2009}
\begin{bbook}[author]
\bauthor{\bsnm{Searle},~\bfnm{Shayle~R}\binits{S.~R.}}, \bauthor{\bsnm{Casella},~\bfnm{George}\binits{G.}} \AND \bauthor{\bsnm{McCulloch},~\bfnm{Charles~E}\binits{C.~E.}}
(\byear{2009}).
\btitle{Variance components}.
\bpublisher{John Wiley \& Sons}.
\end{bbook}
\endbibitem

\bibitem[\protect\citeauthoryear{Searle and Henderson}{1979}]{Searle1979}
\begin{barticle}[author]
\bauthor{\bsnm{Searle},~\bfnm{Shayle~R}\binits{S.~R.}} \AND \bauthor{\bsnm{Henderson},~\bfnm{Harold~Y}\binits{H.~Y.}}
(\byear{1979}).
\btitle{Dispersion matrices for variance components models}.
\bjournal{Journal of the American Statistical Association}
\bvolume{74}
\bpages{465--470}.
\end{barticle}
\endbibitem

\bibitem[\protect\citeauthoryear{Stanton and Li}{2014}]{Stanton1}
\begin{barticle}[author]
\bauthor{\bsnm{Stanton},~\bfnm{B.}\binits{B.}} \AND \bauthor{\bsnm{Li},~\bfnm{X.}\binits{X.}}
(\byear{2014}).
\btitle{A quarter-century of HIV prevention intervention efforts among children and adolescents across the globe}.
\bjournal{Health Psychology and Behavioral Medicine}
\bvolume{2}
\bpages{231-251}.
\end{barticle}
\endbibitem

\bibitem[\protect\citeauthoryear{Stanton et~al.}{1995}]{Stanton2}
\begin{barticle}[author]
\bauthor{\bsnm{Stanton},~\bfnm{B.}\binits{B.}}, \bauthor{\bsnm{Black},~\bfnm{M.}\binits{M.}}, \bauthor{\bsnm{Feigelman},~\bfnm{S.}\binits{S.}} \AND \bauthor{\bparticle{et} \bsnm{al.}}
(\byear{1995}).
\btitle{Development of a culturally, theoretically and developmentally based survey instrument for assessing risk behaviors among African-American early adolescents living in urban low-income neighborhoods}.
\bjournal{AIDS Education and Prevention}
\bvolume{7}
\bpages{160–177}.
\end{barticle}
\endbibitem

\bibitem[\protect\citeauthoryear{Szatrowski and Miller}{1980}]{szatrowski1980}
\begin{barticle}[author]
\bauthor{\bsnm{Szatrowski},~\bfnm{Ted~H}\binits{T.~H.}} \AND \bauthor{\bsnm{Miller},~\bfnm{John~J}\binits{J.~J.}}
(\byear{1980}).
\btitle{Explicit maximum likelihood estimates from balanced data in the mixed model of the analysis of variance}.
\bjournal{The Annals of Statistics}
\bpages{811--819}.
\end{barticle}
\endbibitem

\bibitem[\protect\citeauthoryear{Varadhan and Seeger}{2013}]{varadhan2013}
\begin{bincollection}[author]
\bauthor{\bsnm{Varadhan},~\bfnm{Ravi}\binits{R.}} \AND \bauthor{\bsnm{Seeger},~\bfnm{John~D}\binits{J.~D.}}
(\byear{2013}).
\btitle{Estimation and reporting of heterogeneity of treatment effects}.
In \bbooktitle{Developing a protocol for observational comparative effectiveness research: A user's guide}
\bpublisher{Agency for Healthcare Research and Quality (US)}.
\end{bincollection}
\endbibitem

\bibitem[\protect\citeauthoryear{Vicens et~al.}{2011}]{vicens2011}
\begin{barticle}[author]
\bauthor{\bsnm{Vicens},~\bfnm{Caterina}\binits{C.}}, \bauthor{\bsnm{Socias},~\bfnm{Isabel}\binits{I.}}, \bauthor{\bsnm{Mateu},~\bfnm{Catalina}\binits{C.}}, \bauthor{\bsnm{Leiva},~\bfnm{Alfonso}\binits{A.}}, \bauthor{\bsnm{Bejarano},~\bfnm{Ferran}\binits{F.}}, \bauthor{\bsnm{Sempere},~\bfnm{Ermengol}\binits{E.}}, \bauthor{\bsnm{Basora},~\bfnm{Josep}\binits{J.}}, \bauthor{\bsnm{Palop},~\bfnm{Vicente}\binits{V.}}, \bauthor{\bsnm{Mengual},~\bfnm{Marta}\binits{M.}}, \bauthor{\bsnm{Beltran},~\bfnm{Jose~Luis}\binits{J.~L.}} \betal{et~al.}
(\byear{2011}).
\btitle{Comparative efficacy of two primary care interventions to assist withdrawal from long term benzodiazepine use: a protocol for a clustered, randomized clinical trial}.
\bjournal{BMC Family Practice}
\bvolume{12}
\bpages{1--7}.
\end{barticle}
\endbibitem

\bibitem[\protect\citeauthoryear{Wang et~al.}{2020}]{wang2020}
\begin{barticle}[author]
\bauthor{\bsnm{Wang},~\bfnm{B.}\binits{B.}}, \bauthor{\bsnm{Deveaux},~\bfnm{L.}\binits{L.}}, \bauthor{\bsnm{Lunn},~\bfnm{S.}\binits{S.}}, \bauthor{\bsnm{Dinaj-Koci},~\bfnm{V.}\binits{V.}}, \bauthor{\bsnm{Ghosh},~\bfnm{S.}\binits{S.}}, \bauthor{\bsnm{Li},~\bfnm{X.}\binits{X.}}, \bauthor{\bsnm{Marshall},~\bfnm{S.}\binits{S.}}, \bauthor{\bsnm{Rolle},~\bfnm{G.}\binits{G.}}, \bauthor{\bsnm{Forbes},~\bfnm{N.}\binits{N.}} \AND \bauthor{\bsnm{Stanton},~\bfnm{B.}\binits{B.}}
(\byear{2020}).
\btitle{Bahamas national implementation project: Proposal for sustainability of an evidence-based HIV prevention intervention in a school setting}.
\bjournal{JMIR research protocols,}
\bvolume{9}
\bpages{e14816}.
\end{barticle}
\endbibitem

\bibitem[\protect\citeauthoryear{Wang et~al.}{2023}]{wang2023}
\begin{barticle}[author]
\bauthor{\bsnm{Wang},~\bfnm{Xueqi}\binits{X.}}, \bauthor{\bsnm{Goldfeld},~\bfnm{Keith~S}\binits{K.~S.}}, \bauthor{\bsnm{Taljaard},~\bfnm{Monica}\binits{M.}} \AND \bauthor{\bsnm{Li},~\bfnm{Fan}\binits{F.}}
(\byear{2023}).
\btitle{Sample size requirements to test subgroup-specific treatment effects in cluster-randomized trials}.
\bjournal{Prevention Science}
\bpages{1--15}.
\end{barticle}
\endbibitem

\end{thebibliography}


\begin{thebibliography}{4}
\bibitem[\protect\citeauthoryear{Searle1979}{1979}]{r1}\textsc{Searle, S.R.} and \textsc{Henderson, H.Y.} (1979). \textit{Dispersion matrices for variance components models}, \textit{Journal of the American Statistical Association}, \textit{74}(366), 465--470.

\bibitem[\protect\citeauthoryear{searle2009}{2009}]{r2}\textsc{Searle, S.R., Casella, G.} and \textsc{McCulloch, C.E.} (2009). \textit{Variance components}, \textit{John Wiley \& Sons}.

\bibitem[\protect\citeauthoryear{szatrowski1980}{1980}]{r3}\textsc{Szatrowski, T.H.} and \textsc{Miller, J.J.} (1980). \textit{Explicit maximum likelihood estimates from balanced data in the mixed model of the analysis of variance}. \textit{The Annals of Statistics}, \textit{8}(4), 811--819.

\end{thebibliography}

\setcounter{section}{0}\renewcommand{\thesection}{S.\arabic{section}}

\setcounter{subsection}{0}\renewcommand{\thesubsection}{S.\arabic{section}.\arabic{subsection}}

\newpage



\hspace{1.8in} \Large \textsc{Supplementary material}   
\normalsize

\vspace{0.2in}

\section{Sample size determination for detecting  treatment effects on the individual subgroups}\label{appendix1}

Suppose the null hypothesis $H_{0} : \delta = 0$ gets rejected when tested against the alternative $H_{a} : \delta \neq 0$, that is, the differential treatment effect turns out to be statistically significant. Subsequently, one may be interested to test for the significance of the individual treatment effects on the two subgroups. In other words, the problem of interest might then be to test for $H_{0g} : \delta_g = 0$ against $H_{ag} : \delta_g \neq 0$, for\, $g=1,2$. We present in the following two subsections explicit expressions for sample size requirements to detect the individual treatment effects, when subgrouping is done at level one and level two, respectively.

\subsection{Subgrouping at level one}\label{appendix11}
In equations (\ref{var_MLE_1}) and (\ref{var_MLE_2}) in the technical appendix, we essentially show that
\begin{align*}
    \Var(\hat{\delta}_g) \;=\; \frac{2\,\sigma^2\, f_g}{N_{1g} N_2 N_3}
\end{align*}
{\rm for} \,$g=1,2$,\, where\, $f_g = 1 + \rho_1\,(N_{1g}-1) + \rho_{(2)}\,N_{1g} (N_2 -1)$. Therefore, to test for the significance of the treatment effect on the individual subgroups, that is, to test for the null hypothesis $H_{0g} : \delta_g = 0$ against $H_{ag} : \delta_g \neq 0$ individually for $g=1,2$, we propose the following test statistic by standardizing the estimator $\hat{\delta_g}$:
\begin{align*}
    T_g \;=\; \frac{\sqrt{N_{1g} N_2 N_3}\,\,\hat{\delta}_g}{\sigma\,\sqrt{2 f_g}}\,.
\end{align*}
Now, if the parameters $\sigma^2$, $\rho_1$, $\rho_{(1)}$, $\rho_{(2)}$ (and hence $\rho_2$), or equivalently, the variance components $\sigma_3^2$, $\sigma_2^2$, $\sigma_{\text{grp}}^2$ and $\sigma_e^2$, are assumed to be known, it is easy to see that\, $T_g \sim N(0,1)$ under the null $H_{0g}$. Following similar lines of (\ref{power_ind_diff}), the power of the test, denoted by $\phi$, can therefore be written as:
\begin{align}\label{power_ind}
    \begin{split}
        \phi \;=\; 1-\beta \;\leq\; P_{H_{ag}}\left( \vert T_g \vert > z_{1-\alpha_g/2} \right) \;\leq\; \Phi \left( \frac{\sqrt{N_{1g} N_2 N_3}\,\,\vert\delta_g \vert}{\sigma\,\sqrt{2 f_g}} \;-\; z_{1-\alpha_g/2} \right)\,,
    \end{split}
\end{align}
where\, $\alpha_g \in (0,1)$ is the significance level for testing $H_{0g}$ against $H_{ag}$ for $g=1,2$\, and\, $\beta$\, is the probability of the type II error. From (\ref{power_ind}), some simple algebraic manipulations yield
\begin{align}\label{power_ind2}
    \frac{\sqrt{N_{1g} N_2 N_3}\,\,\vert \delta_g \vert}{\sigma\,\sqrt{2 f_g}} \;\geq\; z_{1-\alpha_g} \,+\, z_{1-\beta}\,.
\end{align}
Consider $\Delta_g = \vert\delta_g\vert/\sigma$\, as a signal to noise ratio. Then we get from (\ref{power_ind2}) 
\begin{align*}
    \Delta_g^2\,N_{1g} N_2 N_3 \;\geq\; 2\,\left(z_{1-\alpha_g} \,+\, z_{1-\beta}\right)^2\,\left(1 + \rho_1\,(N_{1g}-1) + \rho_{(2)}\,N_{1g} (N_2 -1)\right)\,,
\end{align*}
which further yields
\begin{align}\label{power_ind3}
    N_{1g} \;\geq\; \frac{2\,\left(z_{1-\alpha_g} \,+\, z_{1-\beta}\right)^2\,(1-\rho_1)}{\Delta_g^2\, N_2 N_3 \,-\,2\,\left(z_{1-\alpha_g} \,+\, z_{1-\beta}\right)^2\,\left(\rho_1 \,+\,(N_2-1)\,\rho_{(2)}\right)}\,,
\end{align}
for\, $g=1,2$. Assuming the number of level two and level three units in the trial to be pre-specified, we can thus obtain from (\ref{power_ind3}) the sample sizes for the two individual subgroups necessary to achieve the power\, $\phi=1-\beta$. Note that setting $\delta_1 = \delta_2$ and $\alpha_1 = \alpha_2$ leads to $N_{11}=N_{12}$.

\subsection{Subgrouping at level two}\label{appendix12}
In equation (\ref{var_1and2}) in the technical appendix, we essentially show that  
\begin{align}\label{var_MLE2}
    \Var(\hat{\delta}_g) \;=\; \frac{2\,\sigma^2\, f_g^*}{N_1 N_{2g} N_3}
\end{align}
for\, $g=1,2$,\, where\, $f_g^* = 1 + \rho_1\,(N_1-1) + \rho_{(2)}\,N_1 (N_{2g} -1)$. Therefore to test for the null hypotheses $H_{0g} : \delta_g = 0$ against $H_{ag} : \delta_g \neq 0$, we propose the following test statistic by standardizing the estimator $\hat{\delta}_g$:
\begin{align*}
    T_g \;=\; \frac{\sqrt{N_1 N_{2g} N_3}\,\,\hat{\delta}_g}{\sigma\,\sqrt{2 f_g^*}}\,.
\end{align*}
Once again, note that $T_g \sim N(0,1)$ under the null $H_{0g}$ when the parameters $\sigma^2$, $\rho_1$, $\rho_{(2)}$ and $\rho_2$ are known, or in other words, when the variance components $\sigma_3^2$, $\sigma_2^2$, $\sigma_{\text{grp}}^2$ and $\sigma_e^2$ are known. Along similar lines of (\ref{power_ind}) and (\ref{power_ind2}), it can be shown that in order to achieve the power $\phi=1-\beta$, the sample sizes need to satisfy
\begin{align*}
    \Delta_g^2\,N_1 N_{2g} N_3 \;\geq\; 2\,\left(z_{1-\alpha_g} \,+\, z_{1-\beta}\right)^2\,\left(1 + \rho_1\,(N_1-1) + \rho_{(2)}\,N_1 (N_{2g} -1)\right)\,,
\end{align*}
which yields
\begin{align}\label{power_ind5}
    N_{2g} \;\geq\; \frac{2\,\left(z_{1-\alpha_g} \,+\, z_{1-\beta}\right)^2\,\left(1-\rho_1 \,+\, (\rho_1-\rho_{(2)})N_1 \right)}{\Delta_g^2\, N_1 N_3 \,-\,2\,\left(z_{1-\alpha_g} \,+\, z_{1-\beta}\right)^2\,\rho_{(2)} N_1}\,,
\end{align}
for\, $g=1,2$. Thus, assuming the numbers of level one and level three units in the trial are pre-specified, we can obtain from (\ref{power_ind5}) the sample sizes for the two subgroups necessary to achieve the power\, $\phi=1-\beta$. Once again, setting $\delta_1 = \delta_2$ and $\alpha_1 = \alpha_2$ leads to $N_{21}=N_{22}$.\\

\section{Technical appendix}\label{tech_appendix}

In this section, we present the proofs of the key theoretical results of this paper, in case of subgrouping at level two (Section \ref{subsec:subg2_theory}). 
We omit the proofs of the theoretical results in case of subgrouping at level one (Section \ref{subsec:subg1_theory}), which would essentially follow similar lines, only with changed design matrices $X, Z_1, Z_2$ and $Z_3$, and the vectors of random effects $U_1, U_2$ and $U_3$; the forms and  expressions of which are illustrated in full details in Section \ref{subsec:subg1}.\\

\emph{Notations and identities.} We first introduce some additional notations and mention some linear algebraic identities which turn out to be very helpful throughout the subsequent technical derivations. We denote by $A^T$ the transpose of a matrix $A$. For a square matrix $A$, let $\textrm{tr}(A)$ denote the trace of $A$. For a generic positive integer $a$, denote by $I_a$ the $a \times a$ identity matrix. Further denote $J_a = 1_a 1_a^T$, $\bar{J}_a = \frac{1}{a} J_a$, and $C_a = I_a - \bar{J}_a$. Then the following identities hold: $\textrm{tr}(I_a) = a$, $\textrm{tr}(J_a) = a$, $\textrm{tr}(\bar{J}_a) = 1$, $\textrm{tr}(C_a) = a-1$, $J_a 1_a = a 1_a$, $\bar{J}_a 1_a = 1_a$, $C_a 1_a = 0$, $C_a J_a = 0$ and $\bar{J}_a J_a = J_a$. Regarding properties of kronecker product of matrices, if $A, B, C$ and $D$ are matrices such that one can form the matrix products $AC$ and $BD$, then $(A \otimes B) (C \otimes D) = (AC \otimes BD)$. $A \otimes B$ is invertible if and only if both $A$ and $B$ are invertible, in which case the inverse is given by $(A \otimes B)^{-1} = A^{-1} \otimes B^{-1}$. And finally, $(A \otimes B)^T = A^T \otimes B^T$ and $\textrm{tr}(A \otimes B) = \textrm{tr}(A) \textrm{tr}(B)$.\\

Now, note that in (\ref{Model:Subg2_LM}), we presented the model (\ref{Model:Subg2}) using mixed model notations as 
\begin{align*}
    Y \;=\; X\beta + Z_1\,U_1 + Z_2\, U_2 + Z_3\,U_3 + Z_0\,\epsilon\,,
\end{align*}
where $Z_0 = I_{2N_3} \otimes I_2 \otimes I_n \otimes I_{N_1},\, Z_1 = I_{2N_3} \otimes 1_2 \otimes 1_n \otimes 1_{N_1},\, Z_2 = I_{2N_3} \otimes I_2 \otimes 1_n \otimes 1_{N_1}$ and $Z_3 = I_{2N_3} \otimes I_2 \otimes I_n \otimes 1_{N_1}$. Denote $V := \Cov(Y)$. We begin with the following lemma where we rigorously establish an explicit expression of the inverse of the covariance matrix of $Y$. This serves an instrumental role in the proofs of the main theoretical results presented in Section \ref{subsec:subg2_theory}.

\begin{lemma}\label{Lemma_V}
The inverse of the covariance matrix of $Y$ is given by
\begin{align*}
    V^{-1} \;&=\; \frac{1}{\lambda_0}\,I_{2N_3} \otimes I_2 \otimes I_n \otimes C_{N_1} \;+\; \frac{1}{\lambda_1}\,I_{2N_3} \otimes I_2 \otimes C_n \otimes \bar{J}_{N_1}\\
   &\;\;\;\; + \frac{1}{\lambda_2}\,I_{2N_3} \otimes C_2 \otimes \bar{J}_n \otimes \bar{J}_{N_1} \;+\; \frac{1}{\lambda_3}\,I_{2N_3} \otimes \bar{J}_2 \otimes \bar{J}_n \otimes \bar{J}_{N_1}\,.
\end{align*}
\end{lemma}

\begin{proof}[Proof of Lemma \ref{Lemma_V}]
Note that the covariance matrix of $Y$ can be written as
\begin{align*}
    V \;&=\; \Cov(Y) \;=\; \sigma_e^2\, Z_0 Z_0^T + \sigma_3^2\, Z_1 Z_1^T + \sigma_2^2\, Z_2 Z_2^T + \sigma_1^2\, Z_3 Z_3^T\\
    &=\; \sigma_e^2\; I_{2N_3} \otimes I_2 \otimes I_n \otimes I_{N_1} \;+\; \sigma_3^2\;I_{2N_3} \otimes J_2 \otimes J_n \otimes J_{N_1}\\
    & \;\;\;\;\; +\; \sigma_2^2\;I_{2N_3} \otimes I_2 \otimes J_n \otimes J_{N_1} \;+\; \sigma_1^2\; I_{2N_3} \otimes I_2 \otimes I_n \otimes J_{N_1}\\
    &=\; \dis \sum_{a_4,a_3,a_2,a_1=0}^1 \theta_{a_4 a_3 a_2 a_1}\,\left( J_{2N_3}^{a_4} \otimes J_2^{a_3} \otimes J_n^{a_2} \otimes J_{N_1}^{a_1} \right)\,,
\end{align*}
where $\theta_{0000}=\sigma_e^2$, $\theta_{0111}=\sigma_3^2$, $\theta_{0011}=\sigma_2^2$, $\theta_{0001}=\sigma_1^2$, and $\theta_{a_4 a_3 a_2 a_1} = 0$ for other combinations of $(a_4, a_3, a_2, a_1)$. We define
\begin{align*}
    \theta \;&:=\; \big( \theta_{0000}, \theta_{0001}, \theta_{0010}, \theta_{0011}, \theta_{0100}, \theta_{0101}, \theta_{0110}, \theta_{0111}, \\
    & \qquad\;\; \theta_{1000}, \theta_{1001}, \theta_{1010}, \theta_{1011}, \theta_{1100}, \theta_{1101}, \theta_{1110}, \theta_{1111} \big)\\
    &=\; \big(\sigma_e^2, \sigma_1^2, 0, \sigma_2^2, 0, 0, 0, \sigma_3^2,\, 0_8^T \big)\,.
\end{align*}
Further define
{
\allowdisplaybreaks
\begin{align*}
    T \;&:=\; \begin{pmatrix}
        1 & 0\\
        1 & 2N_3
    \end{pmatrix} \otimes 
    \begin{pmatrix}
        1 & 0\\
        1 & 2
    \end{pmatrix} \otimes 
    \begin{pmatrix}
        1 & 0\\
        1 & n
    \end{pmatrix} \otimes 
    \begin{pmatrix}
        1 & 0\\
        1 & N_1
    \end{pmatrix} \\
    &=\; \begin{pmatrix}
        1 & 0\\
        1 & 2N_3
    \end{pmatrix} \otimes 
    \begin{pmatrix}
        1 & 0\\
        1 & 2
    \end{pmatrix} \otimes 
    \begin{pmatrix}
        1 & 0 & 0 & 0\\
        1 & N_1 & 0 & 0 \\
        1 & 0 & n & 0\\
        1 & N_1 & n & nN_1
    \end{pmatrix}  \\
    &=\; \begin{pmatrix}
        1 & 0\\
        1 & 2N_3
    \end{pmatrix} \otimes 
    \begin{pmatrix}
        1 & 0 & 0 & 0 & 0 & 0 & 0 & 0\\
        1 & N_1 & 0 & 0 & 0 & 0 & 0 & 0 \\
        1 & 0 & n & 0 & 0 & 0 & 0 & 0\\
        1 & N_1 & n & nN_1 & 0 & 0 & 0 & 0 \\
        1 & 0 & 0 & 0 & 2 & 0 & 0 & 0\\
        1 & N_1 & 0 & 0 & 2 & 2N_1 & 0 & 0 \\
        1 & 0 & n & 0 & 2 & 0 & 2n & 0\\
        1 & N_1 & n & nN_1 & 2 & 2N_1 & 2n & 2nN_1
    \end{pmatrix}\; =:\; \begin{pmatrix}
        1 & 0\\
        1 & 2N_3
    \end{pmatrix} \otimes A\\
  &=\; \begin{pmatrix}
        A & O_{8,8}\\
        A & 2N_3\, A
    \end{pmatrix}\,.
\end{align*}
}
Next we define
{
\allowdisplaybreaks
\begin{align*}
\lambda \;&:=\; T \theta \;=\; \begin{pmatrix}
        A & O_{8,8}\\
        A & 2N_3\, A
    \end{pmatrix} \, \left( \sigma_e^2, \sigma_1^2, 0, \sigma_2^2, 0, 0, 0, \sigma_3^2,\, 0_8^T \right)^T\\
    &=\; \big(\sigma_e^2,\, \sigma_e^2 + N_1 \sigma_1^2,\, \sigma_e^2,\, \sigma_e^2 + N_1 \sigma_1^2 + nN_1 \sigma_2^2,\, \sigma_e^2,\,  \sigma_e^2 + N_1 \sigma_1^2,\, \sigma_e^2, \,\sigma_e^2 + N_1 \sigma_1^2 + nN_1 \sigma_2^2 + 2nN_1 \sigma_3^2,\\
    & \qquad \sigma_e^2,\, \sigma_e^2 + N_1 \sigma_1^2,\, \sigma_e^2,\, \sigma_e^2 + N_1 \sigma_1^2 + nN_1 \sigma_2^2,\, \sigma_e^2,\,  \sigma_e^2 + N_1 \sigma_1^2,\, \sigma_e^2, \,\sigma_e^2 + N_1 \sigma_1^2 + nN_1 \sigma_2^2 + 2nN_1 \sigma_3^2 \big)^T\\
    &=:\; \big(\lambda_0, \, \lambda_1, \, \lambda_0, \, \lambda_2, \, \lambda_0, \, \lambda_1, \, \lambda_0, \, \lambda_3, \, \lambda_0, \, \lambda_1, \, \lambda_0, \, \lambda_2, \, \lambda_0, \, \lambda_1, \, \lambda_0, \, \lambda_3 \big)
\end{align*}
}
Finally define $\Delta := T^{-1} \lambda^{-1}$, where
\begin{align*}
    T^{-1} \;=\; \frac{1}{2N_3}\,\begin{pmatrix}
        2N_3 & 0\\
        -1 & 1
    \end{pmatrix} \otimes
    \frac{1}{2}\,\begin{pmatrix}
        2 & 0\\
        -1 & 1
    \end{pmatrix} \otimes
    \frac{1}{n}\,\begin{pmatrix}
        n & 0\\
        -1 & 1
    \end{pmatrix} \otimes
    \frac{1}{N_1}\,\begin{pmatrix}
        N_1 & 0\\
        -1 & 1
    \end{pmatrix} \,.
\end{align*}
Some straightforward algebraic calculations yield
{
\allowdisplaybreaks
\begin{align*}
\Delta \;=\; \begin{pmatrix}
    1/\lambda_0 =: \Delta_{0000}\\
    \frac{1}{N_1} \left( 1/\lambda_1 - 1/\lambda_0 \right) =: \Delta_{0001}\\
    0 \\
     \frac{1}{nN_1} \left( 1/\lambda_2 - 1/\lambda_1 \right) =: \Delta_{0011}\\
     0\\
     0\\
     0\\
      \frac{1}{2nN_1} \left( 1/\lambda_3 - 1/\lambda_2 \right) =: \Delta_{0111}\\
      0_8
\end{pmatrix}
\end{align*}
}
\noindent With the above, and using the result from Section 2.2 in \cite{Searle1979}, we obtain the inverse of the covariance matrix $V$ as
{
\allowdisplaybreaks
\begin{align*}
V^{-1} \;&=\; \Delta_{0000} \, I_{2N_3} \otimes I_2 \otimes I_n \otimes I_{N_1} \;+\; \Delta_{0001} \, I_{2N_3} \otimes I_2 \otimes I_n \otimes J_{N_1}\\
&\qquad + \Delta_{0011} \, I_{2N_3} \otimes I_2 \otimes J_n \otimes J_{N_1} \;+\; \Delta_{0111} \, I_{2N_3} \otimes J_2 \otimes J_n \otimes J_{N_1}\\
&=\; \frac{1}{\lambda_0}\,I_{2N_3} \otimes I_2 \otimes I_n \otimes I_{N_1} \;+\; \left(\frac{1}{\lambda_1} - \frac{1}{\lambda_0}\right)\,I_{2N_3} \otimes I_2 \otimes I_n \otimes \bar{J}_{N_1}\\
& \qquad + \left(\frac{1}{\lambda_2} - \frac{1}{\lambda_1}\right)\,I_{2N_3} \otimes I_2 \otimes \bar{J}_n \otimes \bar{J}_{N_1} \;+\; \left(\frac{1}{\lambda_3} - \frac{1}{\lambda_2}\right)\,I_{2N_3} \otimes \bar{J}_2 \otimes \bar{J}_n \otimes \bar{J}_{N_1}\\
&=\; \frac{1}{\lambda_0}\,I_{2N_3} \otimes I_2 \otimes I_n \otimes C_{N_1} \;+\; \frac{1}{\lambda_1}\,I_{2N_3} \otimes I_2 \otimes C_n \otimes \bar{J}_{N_1}\\
   &\qquad + \frac{1}{\lambda_2}\,I_{2N_3} \otimes C_2 \otimes \bar{J}_n \otimes \bar{J}_{N_1} \;+\; \frac{1}{\lambda_3}\,I_{2N_3} \otimes \bar{J}_2 \otimes \bar{J}_n \otimes \bar{J}_{N_1}\,,
\end{align*}
}
which completes the proof of the lemma.\\
\end{proof}

\begin{proof}[Proof of Proposition \ref{Prop_MLE_delta2}]
Recall from the discussions in Section \ref{subsec:subg2} that model (\ref{Model:Subg2}) can be written as 
\begin{align*}
  Y \;=\; X\beta\,+\, Z_1\, U_1 +\, Z_2\, U_2 +\, Z_3\, U_3 +\, Z_0\, \epsilon\,,   
\end{align*}
where 
\begin{align*}
    X^T \;=\; \begin{pmatrix}
        1_{N_3 2n N_1}^T & 1_{N_3 2n N_1}^T \\
        \underbrace{1_{nN_1}^T\; 0_{nN_1}^T\; \dots \;1_{nN_1}^T\; 0_{nN_1}^T}_{N_3\, \textrm{times}} & \underbrace{1_{nN_1}^T\; 0_{nN_1}^T\; \dots \;1_{nN_1}^T\; 0_{nN_1}^T}_{N_3\, \textrm{times}}\\ 
        1_{N_3 2n N_1}^T & 0_{N_3 2n N_1}^T \\
        \underbrace{1_{nN_1}^T\; 0_{nN_1}^T\; \dots \;1_{nN_1}^T\; 0_{nN_1}^T}_{N_3\, \textrm{times}} & 0_{N_3 2n N_1}^T
    \end{pmatrix}\,.
\end{align*}
The maximum likelihood (ML) estimates of the fixed effects components are obtained by equating the following (see, for example, Section 6.2 in \cite{searle2009}):
\begin{align}\label{prop3.2.1_eq1}
   X^T V^{-1} X\beta \;=\; X^T V^{-1} Y\,, 
\end{align}
where the explicit form of $V^{-1}$ is obtained in Lemma \ref{Lemma_V} as
\begin{align*}
    V^{-1} \;&=\;  \frac{1}{\lambda_0}\,I_{2N_3} \otimes I_2 \otimes I_n \otimes C_{N_1} \;+\; \frac{1}{\lambda_1}\,I_{2N_3} \otimes I_2 \otimes C_n \otimes \bar{J}_{N_1}\\
   &\;\;\;\; + \frac{1}{\lambda_2}\,I_{2N_3} \otimes C_2 \otimes \bar{J}_n \otimes \bar{J}_{N_1} \;+\; \frac{1}{\lambda_3}\,I_{2N_3} \otimes \bar{J}_2 \otimes \bar{J}_n \otimes \bar{J}_{N_1}\,,\\
   &=:\; I + II + III + IV\,.
\end{align*}
Thus we can write
\begin{align}\label{prop3.2.1_eq2}
    X^T V^{-1} \;=\; X^T I \,+\, X^T II \,+\, X^T III \,+\, X^T IV\,.
\end{align}
Note that
{
\allowdisplaybreaks
\begin{align*}
I \;&=\; \frac{1}{\lambda_0}\,I_{2N_3 2n} \otimes C_{N_1} \;=\; \frac{1}{\lambda_0}\; diag\left(\underbrace{I_{N_1} - \frac{1}{N_1}J_{N_1}\,,\, \dots,\, I_{N_1} - \frac{1}{N_1}J_{N_1}}_{2N_3 2n\;\; \textrm{times}} \right)\,,
\end{align*}
}
and therefore $X^T I = 0$. Likewise, we have
{
\allowdisplaybreaks
\begin{align*}
II \;&=\; \frac{1}{\lambda_1}\,\left(I_{2N_3} \otimes I_2 \otimes I_n \otimes \bar{J}_{N_1}\;-\; I_{2N_3} \otimes I_2 \otimes \bar{J}_n \otimes \bar{J}_{N_1} \right) \\
&=\; \frac{1}{N_1\,\lambda_1}\, \left( I_{2N_3 2n} \otimes J_{N_1} \;-\; \frac{1}{n}\,I_{2N_3 2} \otimes J_{nN_1} \right)\\
&=\; \frac{1}{N_1\,\lambda_1}\, \left\{ diag\left(\underbrace{ J_{N_1}, \, \dots,\,  J_{N_1}}_{2N_3 2n \; \textrm{times}}\right) \;-\; \frac{1}{n}\; diag\left(\underbrace{ J_{nN_1}, \, \dots,\,  J_{nN_1}}_{2N_3 2 \; \textrm{times}}\right) \right\}\\
&=\; \frac{1}{N_1\,\lambda_1}\; diag\left( \underbrace{ diag\left(\underbrace{ J_{N_1}, \, \dots,\,  J_{N_1}}_{n \; \textrm{times}}\right) - \frac{1}{n}\,J_{nN_1},\, \dots\,,\, diag\left(\underbrace{ J_{N_1}, \, \dots,\,  J_{N_1}}_{n \; \textrm{times}}\right) - \frac{1}{n}\,J_{nN_1} }_{2N_3 2 \; \textrm{times}}\right)\,,
\end{align*}
}
and therefore $X^T II = 0$ as well. Thus equation (\ref{prop3.2.1_eq2}) boils down to
\begin{align}\label{prop3.2.1_eq3}
    X^T V^{-1} \;=\;  X^T III \,+\, X^T IV\,.
\end{align}
Now observe that
{
\allowdisplaybreaks
\begin{align*}
III \;&=\; \frac{1}{2nN_1\,\lambda_2}\, I_{2N_3} \otimes \begin{pmatrix}
    1 & -1 \\
    -1 & 1
\end{pmatrix} \otimes J_{nN_1} \;=\; \frac{1}{2nN_1\,\lambda_2}\, I_{2N_3} \otimes \begin{pmatrix}
    J_{nN_1} & -J_{nN_1} \\
    -J_{nN_1} & J_{nN_1}
\end{pmatrix}\\
&=\; \frac{1}{2nN_1\,\lambda_2}\; diag\left(\underbrace{\begin{pmatrix}
    J_{nN_1} & -J_{nN_1} \\
    -J_{nN_1} & J_{nN_1}
\end{pmatrix}\,,\, \dots\,,\, \begin{pmatrix}
    J_{nN_1} & -J_{nN_1} \\
    -J_{nN_1} & J_{nN_1}
\end{pmatrix}}_{2N_3\; \textrm{times}} \right)\,.
\end{align*}
}
Therefore we have
{
\allowdisplaybreaks
\begin{align*}
X^T III \;&=\; \frac{1}{2nN_1\,\lambda_2}\;\begin{pmatrix}
        0_{2N_3 2n N_1}^T \\
        \underbrace{nN_1\,1_{nN_1}^T\; -nN_1\,1_{nN_1}^T\; \dots \;nN_1\,1_{nN_1}^T\; -nN_1\,1_{nN_1}^T\;}_{2N_3\, \textrm{times}} \\ 
        0_{2N_3 2n N_1}^T \\
        \underbrace{nN_1\,1_{nN_1}^T\; -nN_1\,1_{nN_1}^T\; \dots \;nN_1\,1_{nN_1}^T\; -nN_1\,1_{nN_1}^T\;}_{N_3\, \textrm{times}} & 0_{N_3 2n N_1}^T
    \end{pmatrix}\\
    &=\;\frac{1}{2\,\lambda_2}\;\begin{pmatrix}
        0_{2N_3 2n N_1}^T \\
        \underbrace{1_{nN_1}^T\; -1_{nN_1}^T\; \dots \;\,1_{nN_1}^T\; -1_{nN_1}^T\;}_{2N_3\, \textrm{times}} \\ 
        0_{2N_3 2n N_1}^T \\
        \underbrace{1_{nN_1}^T\; -1_{nN_1}^T\; \dots \;\,1_{nN_1}^T\; -1_{nN_1}^T\;}_{N_3\, \textrm{times}} & 0_{N_3 2n N_1}^T
    \end{pmatrix}\,.
\end{align*}
}
Some straightforward algebraic calculations yield
{
\allowdisplaybreaks
\begin{align*}
X^T III\, X \;&=\; \frac{1}{2\,\lambda_2}\; \begin{pmatrix}
    0 & 0 & 0 & 0\\
    0 & 2N_3 nN_1 & 0 & N_3 nN_1\\
    0 & 0 & 0 & 0\\
    0 & N_3 nN_1 & 0 & N_3 nN_1
\end{pmatrix} \;=\; \frac{N_3\,nN_1}{2\,\lambda_2}\; \begin{pmatrix}
    0 & 0 & 0 & 0\\
    0 & 2 & 0 & 1\\
    0 & 0 & 0 & 0\\
    0 & 1 & 0 & 1
\end{pmatrix}\\
\textrm{and} \qquad X^T III\; Y \;&=\; \frac{N_3\,nN_1}{2\,\lambda_2}\; \begin{pmatrix}
    0\\
   2\left( \bar{Y}_{\cdot 1 \cdot \cdot} - \bar{Y}_{\cdot 2 \cdot \cdot}\right)\\
   0\\
   \bar{Y}_{\cdot 1 \cdot \cdot} - \bar{Y}_{\cdot 2 \cdot \cdot}
\end{pmatrix}\,.
\end{align*}
}
Finally, note that
{
\allowdisplaybreaks
\begin{align*}
IV \;&=\; \frac{1}{2nN_1\,\lambda_3}\; I_{2N_3} \otimes J_{2nN_1} \;=\;  \frac{1}{2nN_1\,\lambda_3}\; diag\left( \underbrace{J_{2nN_1}\,, \,\dots\,,\, J_{2nN_1}}_{2N_3\; \textrm{times}} \right)\,,
\end{align*}
}
and therefore
\begin{align*}
    X^T IV \;=\; \frac{1}{2nN_1\,\lambda_3}\; \begin{pmatrix}
        2nN_1\, 1_{2N_3 2n N_1}^T\\
        nN_1\, 1_{2N_3 2n N_1}^T \\
        2nN_1\, 1_{N_3 2n N_1}^T & 0_{N_3 2n N_1}^T\\
        nN_1\, 1_{N_3 2n N_1}^T & 0_{N_3 2n N_1}^T
    \end{pmatrix} \;=\; \frac{1}{\lambda_3}\;\begin{pmatrix}
        1_{2N_3 2n N_1}^T\\
        \frac{1}{2}\, 1_{2N_3 2n N_1}^T \\
        1_{N_3 2n N_1}^T & 0_{N_3 2n N_1}^T\\
        \frac{1}{2}\, 1_{N_3 2n N_1}^T & 0_{N_3 2n N_1}^T
    \end{pmatrix}\,.
\end{align*}
Once again, some straightforward calculations yield
{
\allowdisplaybreaks
\begin{align*}
X^T IV X \;&=\; \frac{1}{\lambda_3}\; \begin{pmatrix}
    2N_3 2n N_1 & 2N_3 n N_1 & N_3 2n N_1 & N_3 n N_1\\
    2N_3 n N_1 & N_3 n N_1 & N_3 n N_1 & \frac{1}{2}\,N_3 n N_1 \\
    N_3 2n N_1 & N_3 n N_1 & N_3 2n N_1 & N_3 n N_1\\
    N_3 n N_1 & \frac{1}{2}\,N_3 n N_1 & N_3 n N_1 & \frac{1}{2}\,N_3 n N_1
\end{pmatrix} \;=\; \frac{N_3 n N_1}{2\lambda_3}\;\begin{pmatrix}
    8 & 4 & 4 & 2\\
    4 & 2 & 2 & 1\\
    4 & 2 & 4 & 2\\
    2 & 1 & 2 & 1
\end{pmatrix}\,,\\
\textrm{and} \qquad X^T IV\, Y \;&=\; \frac{1}{\lambda_3}\; \begin{pmatrix}
    2N_3 2n N_1\, \bar{Y}_{\cdot \cdot \cdot \cdot}\\
    N_3 2n N_1\, \bar{Y}_{\cdot \cdot \cdot \cdot}\\
    N_3 2n N_1\, \bar{Y}_{\cdot \cdot \cdot \cdot}\\
    N_3 n N_1\, \bar{Y}_{\cdot \cdot \cdot \cdot}
\end{pmatrix} \;=\; \frac{N_3 2n N_1}{2\lambda_3}\; \begin{pmatrix}
    4 \bar{Y}_{\cdot \cdot \cdot \cdot}\\
    2 \bar{Y}_{\cdot \cdot \cdot \cdot}\\
    2 \bar{Y}_{\cdot \cdot \cdot \cdot}\\
    \bar{Y}_{\cdot \cdot \cdot \cdot}
\end{pmatrix}\,.
\end{align*}
}
Combining all the above, the left hand side of (\ref{prop3.2.1_eq1}) is given by
\begin{align}\label{prop3.2.1_eq4}
\begin{split}
 X^T V^{-1} X\beta \;&=\; X^T III\,X\beta\,  +\, X^T IV\, X\beta   \\
 &=\; \left\{\frac{N_3\,nN_1}{2\,\lambda_2}\; \begin{pmatrix}
    0 & 0 & 0 & 0\\
    0 & 2 & 0 & 1\\
    0 & 0 & 0 & 0\\
    0 & 1 & 0 & 1
\end{pmatrix} \;+\; \frac{N_3 n N_1}{2\lambda_3}\;\begin{pmatrix}
    8 & 4 & 4 & 2\\
    4 & 2 & 2 & 1\\
    4 & 2 & 4 & 2\\
    2 & 1 & 2 & 1
\end{pmatrix} \right\} \begin{pmatrix}
    \beta_0 \\ \tau \\ \xi \\ \delta
\end{pmatrix}\\
&=\; \frac{N_3\,nN_1}{2\,\lambda_2}\;\begin{pmatrix}
    0 \\ 2\tau + \delta \\ 0 \\ \tau + \delta
\end{pmatrix} \;+\; \frac{N_3\,n N_1}{2\lambda_3}\; \begin{pmatrix}
    2\,(4\beta_0 + 2\tau + 2\xi + \delta)\\
    4\beta_0 + 2\tau + 2\xi + \delta\\
    2\,(2\beta_0 + \tau + 2\xi + \delta)\\
    2\beta_0 + \tau + 2\xi + \delta
\end{pmatrix}\,.
\end{split}
\end{align}
And the right hand side of (\ref{prop3.2.1_eq1}) is given by
\begin{align}\label{prop3.2.1_eq5}
\begin{split}
X^T V^{-1} Y \;&=\; X^T III\,Y\,  +\, X^T IV\, Y   \\
 &=\; \frac{N_3\,nN_1}{2\,\lambda_2}\;
  \begin{pmatrix}
    0\\
   2\left( \bar{Y}_{\cdot 1 \cdot \cdot} - \bar{Y}_{\cdot 2 \cdot \cdot}\right)\\
   0\\
   \bar{Y}_{\cdot 1 \cdot \cdot} - \bar{Y}_{\cdot 2 \cdot \cdot}
\end{pmatrix} \;+\; \frac{N_3\,n N_1}{2\lambda_3}\; \begin{pmatrix}
    8 \bar{Y}_{\cdot \cdot \cdot \cdot}\\
    4 \bar{Y}_{\cdot \cdot \cdot \cdot}\\
    4 \bar{Y}_{\cdot \cdot \cdot \cdot}\\
    2 \bar{Y}_{\cdot \cdot \cdot \cdot}
\end{pmatrix}\,.  
\end{split}
\end{align}
From (\ref{prop3.2.1_eq1}), (\ref{prop3.2.1_eq4}) and (\ref{prop3.2.1_eq5}), we get
{
\allowdisplaybreaks
\begin{align}
& 2\tau + \delta \;=\; 2\left( \bar{Y}_{\cdot 1 \cdot \cdot} - \bar{Y}_{\cdot 2 \cdot \cdot}\right) \;=\; \bar{Y}_{\cdot 1 \cdot \cdot}^T + \bar{Y}_{\cdot 1 \cdot \cdot}^C - \bar{Y}_{\cdot 2 \cdot \cdot}^T - \bar{Y}_{\cdot 2 \cdot \cdot}^C\,, \label{prop3.2.1_eq6a}\\
& \tau + \delta \;=\; \bar{Y}_{\cdot 1 \cdot \cdot}^T - \bar{Y}_{\cdot 2 \cdot \cdot}^T\,, \label{prop3.2.1_eq6b}\\
& 4\beta_0 + 2\tau + 2\xi + \delta \;=\; 4 \bar{Y}_{\cdot \cdot \cdot \cdot}\,, \label{prop3.2.1_eq6c}\\
\textrm{and} \qquad & 2\beta_0 + \tau + 2\xi + \delta \;=\; 2 \bar{Y}_{\cdot \cdot \cdot \cdot}^T\,. \label{prop3.2.1_eq6d}
\end{align}
}
Subtracting (\ref{prop3.2.1_eq6a}) from $2\times$(\ref{prop3.2.1_eq6b}), we get
\begin{align}\label{prop3.2.1_eq7}
\begin{split}
  \delta \;&=\;  2\bar{Y}_{\cdot 1 \cdot \cdot}^T - 2\bar{Y}_{\cdot 2 \cdot \cdot}^T - \bar{Y}_{\cdot 1 \cdot \cdot}^T - \bar{Y}_{\cdot 1 \cdot \cdot}^C + \bar{Y}_{\cdot 2 \cdot \cdot}^T + \bar{Y}_{\cdot 2 \cdot \cdot}^C \\
  &=\; (\bar{Y}_{\cdot 1 \cdot \cdot}^T - \bar{Y}_{\cdot 1 \cdot \cdot}^C) \,-\, (\bar{Y}_{\cdot 2 \cdot \cdot}^T - \bar{Y}_{\cdot 2 \cdot \cdot}^C)\,,
  \end{split}
\end{align}
which completes the proof of the proposition. We further note the following, which turn out to be useful in the proofs of the subsequent theorems. Subtracting (\ref{prop3.2.1_eq6d}) from (\ref{prop3.2.1_eq6c}), we get
\begin{align}\label{prop3.2.1_eq8}
  2\beta_0 + \tau \;&=\; 4\bar{Y}_{\cdot \cdot \cdot \cdot}  - 2\bar{Y}_{\cdot \cdot \cdot \cdot}^T \;=\; 2(\bar{Y}_{\cdot \cdot \cdot \cdot}^T + \bar{Y}_{\cdot \cdot \cdot \cdot}^C) - \bar{Y}_{\cdot \cdot \cdot \cdot}^T \;=\; 2\bar{Y}_{\cdot \cdot \cdot \cdot}^C\,.
\end{align}
And finally, subtracting (\ref{prop3.2.1_eq6b}) from (\ref{prop3.2.1_eq6a}), we get
\begin{align}\label{prop3.2.1_eq9}
\tau \;=\; \bar{Y}_{\cdot 1 \cdot \cdot}^C - \bar{Y}_{\cdot 2 \cdot \cdot}^C\,.
\end{align}

\end{proof}

\begin{proof}[Proof of Proposition \ref{Prop_var_delta2}]
Recall that, from Proposition \ref{Prop_MLE_delta2}, we have
\begin{align*}
    \begin{split}
        \hat{\delta} \;&=\; \left( \frac{1}{N_3\,n\,N_1} \dis \sum_{i=1}^{N_3} \sum_{j=1}^{n} \sum_{k=1}^{N_1} Y_{i1jk}\,\mathbbm{1}(X_i=1) \,-\, \frac{1}{N_3\,n\,N_1} \dis \sum_{i=1}^{N_3} \sum_{j=1}^{n} \sum_{k=1}^{N_1} Y_{i1jk}\,\mathbbm{1}(X_i=0) \right)\\
        & \qquad - \; \left( \frac{1}{N_3\,n\,N_1} \dis \sum_{i=1}^{N_3} \sum_{j=1}^{n} \sum_{k=1}^{N_1} Y_{i2jk}\,\mathbbm{1}(X_i=1) \,-\, \frac{1}{N_3\,n\,N_1} \dis \sum_{i=1}^{N_3} \sum_{j=1}^{n} \sum_{k=1}^{N_1} Y_{i2jk}\,\mathbbm{1}(X_i=0) \right)\\
        &=\; \left(A_2 - B_2\right) \;-\; \left(C_2 - D_2\right)\,.
    \end{split}
\end{align*}
On similar lines of (\ref{var1}) and (\ref{var_MLE_1}), 
\begin{align}\label{var_1and2}
    \Var(A_2-B_2) \;=\; \Var(C_2-D_2) \; =\;\frac{2\,\sigma^2\, f^*}{N_3 \,n\, N_1}\,,
\end{align}
where\, $f^* = 1 \,+\, \rho_1\,(N_1-1) \,+\, \rho_{(2)}\,N_1\, (n -1)$. Now observe that
\begin{align}
    \begin{split}
        \Cov(A_2,C_2) \;&=\; \frac{1}{N_3^2\,n^2\, N_1^2}\, \sum_{i=1}^{N_3} \sum_{j,j'=1}^{n} \sum_{k,k'=1}^{N_1} \Cov(Y_{i1jk},Y_{i2j'k'})\\
        &=\; \frac{\sigma^2\,\rho_2}{N_3}\,,
    \end{split}
\end{align}
and\, $\Cov(B_2,D_2) = \Cov(A_2,C_2)$. 
Thus, combining all the above, we have 
\begin{align}\label{var_delta_hat2}
    \begin{split}
        \Var(\hat{\delta}) \;&=\; \Var(A_2-B_2) \,+\, \Var(C_2-D_2) \,-\,\, 2\,\Cov(A_2-B_2,\, C_2-D_2)\\
        &=\; \frac{2\,\sigma^2\, f^*}{N_3 \,n\, N_1}\,+\, \frac{2\,\sigma^2\, f^*}{N_3 \,n\, N_1} \,\,-\,\, 2\,\Cov(A_2,C_2) \,-\, 2\,\Cov(B_2,D_2)\\
        &=\; \frac{4\,\sigma^2}{N_3 \,n\, N_1}\,\left( 1 + \rho_1\,(N_1 -1) + \rho_{(2)}\,N_1 (n -1) \right) \, - \, \frac{4\,\sigma^2\,\rho_2}{N_3}\\
        &=\; \frac{4}{N_3 \,n\, N_1}\,\left( \sigma^2\,(1-\rho_1)\,+\, N_1\,(\sigma^2\,\rho_1 - \sigma^2\,\rho_{(2)}) \,+\, n\,N_1\,(\sigma^2\,\rho_{(2)} - \sigma^2\,\rho_2) \right)\\
        &=\; \frac{4}{N_3 \,n\, N_1}\,\left( \sigma_e^2 \,+\, N_1\,\sigma_1^2 \,+\, n\,N_1\,\sigma_2^2 \right)\,.
    \end{split}
\end{align}

\end{proof}

\begin{proof}[Proof of Theorem \ref{Thm_MLE_var_comps2}]
We first establish that the model (\ref{Model:Subg2})
\begin{align*}
    Y_{igjk} \;&=\; \beta_0 \,+\, \tau\,L_g \,+\, \xi\,X_{ijk} \,+\, \delta\, L_g\,X_{ijk} \,+\, u_i \,+\, u_{g(i)} \,+\, u_{j(ig)} \,+\, \epsilon_{igjk}\,,
\end{align*}
admits explicit, closed form maximum likelihood (ML) estimates of the variance components. In view of Theorem 3 in \cite{szatrowski1980}, let $r$ be the number of random factors in the above mixed model (excluding the error term), and $s$ be the number of different symbols used as subscripts in the model equation. Thus, here we have $r=3$ and $s=4$. 

Let $t_p^T$ be a row vector of order $s$ that is null except for unity at its $q^{th}$ element, when the $p^{th}$ random factor of the model has the $q^{th}$ subscript; and define $t_0 = 1_s$. Let $T$ be a $(r+1)\times s$ matrix defined as $T= (t_0, \dots , t_r)^T$. Let $w_1, \dots, w_s$ be the columns of $T$ and $w_0 =1_{r+1}$. Define $\mathcal{W}$ as the smallest set containing $w_0, w_1, \dots, w_s$ closed under Hadamard multiplication of vectors. Observe that here we have 
\begin{align*}
    T \;=\; \begin{pmatrix} 
         1 & 1 & 1 & 1\\
         1 & 0 & 0 & 0 \\
         1 & 1 & 0 & 0 \\
         1 & 1 & 1 & 0
    \end{pmatrix}
   \qquad \textrm{and} \qquad
   \mathcal{W} \;=\; \begin{pmatrix} 
         1 & 1 & 1 & 1\\
         1 & 0 & 0 & 0 \\
         1 & 1 & 0 & 0 \\
         1 & 1 & 1 & 0
    \end{pmatrix}\,.
\end{align*}
Let $n(\mathcal{W})$ denote the number of distinct columns in $\mathcal{W}$. Clearly $n(\mathcal{W}) = 4 = r+1$. Therefore, by Theorem 3 in \cite{szatrowski1980}, the model admits explicit ML estimates of the variance components. 

Now, we turn to proving the second part of the theorem. For the estimation of the variance components, we need to satisfy
\begin{align}\label{eq3.2.1.1}
    \textrm{tr}\left( V^{-1} Z_i Z_i^T \right) \;=\; \Vert  Z_i^T V^{-1} (Y-X\beta) \Vert^2\,, 
\end{align}
for $i=0, 1, 2, 3$. Recall that we have from Lemma \ref{Lemma_V} and the discussions in Section \ref{subsec:subg2} of the main paper, that
\begin{align*}
    V^{-1} \;&=\;  \frac{1}{\lambda_0}\,I_{2N_3} \otimes I_2 \otimes I_n \otimes C_{N_1} \;+\; \frac{1}{\lambda_1}\,I_{2N_3} \otimes I_2 \otimes C_n \otimes \bar{J}_{N_1}\\
   &\;\;\;\; + \frac{1}{\lambda_2}\,I_{2N_3} \otimes C_2 \otimes \bar{J}_n \otimes \bar{J}_{N_1} \;+\; \frac{1}{\lambda_3}\,I_{2N_3} \otimes \bar{J}_2 \otimes \bar{J}_n \otimes \bar{J}_{N_1}\,,
\end{align*}
and $Z_0 = I_{2N_3} \otimes I_2 \otimes I_n \otimes I_{N_1},\, Z_1 = I_{2N_3} \otimes 1_2 \otimes 1_n \otimes 1_{N_1},\, Z_2 = I_{2N_3} \otimes I_2 \otimes 1_n \otimes 1_{N_1}$ and $Z_3 = I_{2N_3} \otimes I_2 \otimes I_n \otimes 1_{N_1}$. Hence we have 
\begin{align*}
    Z_0 Z_0^T \;&=\; I_{2N_3} \otimes I_2 \otimes I_n \otimes I_{N_1}\,, \qquad
    Z_1 Z_1^T \;=\; I_{2N_3} \otimes J_2 \otimes J_n \otimes J_{N_1}\,, \\
    Z_2 Z_2^T \;&=\; I_{2N_3} \otimes I_2 \otimes J_n \otimes J_{N_1}\,, \qquad
    Z_3 Z_3^T \;=\; I_{2N_3} \otimes I_2 \otimes I_n \otimes J_{N_1}\,.
\end{align*}
For the left hand side of (\ref{eq3.2.1.1}), note that
{
\allowdisplaybreaks
\begin{align}\label{eq3.2.1.0.5}
\begin{split}
\textrm{tr}(V^{-1} Z_0 Z_0^T) \;&=\; \frac{1}{\lambda_0}\, 2N_3\, 2n\, (N_1 -1) \;+\; \frac{1}{\lambda_1}\, 2N_3\, 2\,(n-1) \;+\; \frac{1}{\lambda_2}\, 2N_3 \;+\; \frac{1}{\lambda_3}\, 2N_3\,,\\
\textrm{tr}(V^{-1} Z_1 Z_1^T) \;&=\; \frac{1}{\lambda_3}\,2N_3\,2n N_1\,,\\
\textrm{tr}(V^{-1} Z_2 Z_2^T) \;&=\;  2N_3\,n N_1\, \left(\frac{1}{\lambda_2} + \frac{1}{\lambda_3}\right)\,,\\
\textrm{tr}(V^{-1} Z_3 Z_3^T) \;&=\; \frac{1}{\lambda_1}\, 2N_3\,2(n-1)\,N_1 \;+\; \frac{1}{\lambda_2}\, 2N_3\,N_1 \;+\; \frac{1}{\lambda_3}\, 2N_3 \,N_1\,.
\end{split}
\end{align}
}
And for the right hand side of (\ref{eq3.2.1.1}), note that
{
\allowdisplaybreaks
\begin{align*}
Z_0 V^{-1} \;&=\; V^{-1}\,,\\
Z_1 V^{-1} \;&=\; \frac{1}{\lambda_3}\,I_{2N_3} \otimes 1_2^T \otimes 1_n^T \otimes 1_{N_1}^T \,,\\
Z_2 V^{-1} \;&=\; \frac{1}{\lambda_2}\,I_{2N_3} \otimes C_2 \otimes 1_n^T \otimes 1_{N_1}^T \;+\; \frac{1}{\lambda_3}\,I_{2N_3} \otimes \bar{J}_2 \otimes 1_n^T \otimes 1_{N_1}^T \,,\\
Z_3 V^{-1} \;&=\; \frac{1}{\lambda_1}\,I_{2N_3} \otimes I_2 \otimes C_n \otimes 1_{N_1}^T \;+\; \frac{1}{\lambda_2}\,I_{2N_3} \otimes C_2 \otimes \bar{J}_n \otimes 1_{N_1}^T \;+\; \frac{1}{\lambda_3}\,I_{2N_3} \otimes \bar{J}_2 \otimes \bar{J}_n \otimes 1_{N_1}^T \,.
\end{align*}
}
With all the above, we have
{
\allowdisplaybreaks
\begin{align*}
Z_1 V^{-1} (Y-X\beta) \;&=\;  \frac{1}{\lambda_3}\,I_{2N_3} \otimes 1_2^T \otimes 1_n^T \otimes 1_{N_1}^T \, (Y-X\beta)\;=\; \frac{1}{\lambda_3}\; diag\Big(\underbrace{1_{2nN_1}^T, \;\dots\;, 1_{2nN_1}^T }_{2N_3\, \textrm{times}}\Big)\, (Y-X\beta)\\
&=\; Z_1^T V^{-1} Y \;-\; Z_1^T V^{-1} X\beta \\
&=\; \frac{2nN_1}{\lambda_3}\begin{pmatrix}
    \bar{Y}_{1\cdot\cdot\cdot} \\
    \vdots \\
    \bar{Y}_{2N_3\cdot\cdot\cdot}
\end{pmatrix}  \;-\; \frac{nN_1}{\lambda_3} \begin{pmatrix}
    (2\beta_0 + \tau + 2\xi +\delta)\,1_{N_3} \\
    (2\beta_0 + \tau)\,1_{N_3}
\end{pmatrix}\\
&=\; \frac{2nN_1}{\lambda_3}\begin{pmatrix}
    \bar{Y}_{1\cdot\cdot\cdot} \\
    \vdots \\
    \bar{Y}_{2N_3\cdot\cdot\cdot}
\end{pmatrix}  \;-\; \frac{2nN_1}{\lambda_3} \begin{pmatrix}
    \bar{Y}_{\cdot\cdot\cdot\cdot}^T\,1_{N_3} \\
    \bar{Y}_{\cdot\cdot\cdot\cdot}^C\,1_{N_3}
\end{pmatrix} \;\;=\;\; \frac{2nN_1}{\lambda_3}\begin{pmatrix}
    \bar{Y}_{1\cdot\cdot\cdot} - \bar{Y}_{\cdot\cdot\cdot\cdot}^T \\
    \vdots \\
    \bar{Y}_{N_3\cdot\cdot\cdot} - \bar{Y}_{\cdot\cdot\cdot\cdot}^T \\
    \bar{Y}_{N_3+1\cdot\cdot\cdot} - \bar{Y}_{\cdot\cdot\cdot\cdot}^C \\
    \vdots \\
    \bar{Y}_{2N_3\cdot\cdot\cdot} - \bar{Y}_{\cdot\cdot\cdot\cdot}^C 
\end{pmatrix}\,.
\end{align*}
}
Therefore, to satisfy (\ref{eq3.2.1.1}), we need to have
{
\allowdisplaybreaks
\begin{align}\label{eq3.2.1.0.6}
\begin{split}
\textrm{tr}\left( V^{-1} Z_1 Z_1^T \right) \;&=\; \Vert  Z_1^T V^{-1} (Y-X\beta) \Vert^2\,, \\
\textrm{i.e.,} \qquad \frac{1}{\lambda_3} 2N_3 2n N_1 \;&=\; \frac{(2nN_1)^2}{\lambda_3^2} \left[ \dis \sum_{i=1}^{N_3} (\bar{Y}_{i\cdot\cdot\cdot} - \bar{Y}_{\cdot\cdot\cdot\cdot}^T)^2 \,+\, \dis \sum_{i=N_3+1}^{2N_3} (\bar{Y}_{i\cdot\cdot\cdot} - \bar{Y}_{\cdot\cdot\cdot\cdot}^C)^2  \right] \;=:\; \frac{2nN_1}{\lambda_3^2}\, SS_3\,,
\end{split}
\end{align}
}
where 
\begin{align}\label{eq3.2.1.2}
    SS_3 &:=\; 2n\,N_1\, \left[ \dis \sum_{i=1}^{N_3} (\bar{Y}_{i\cdot\cdot\cdot} - \bar{Y}_{\cdot\cdot\cdot\cdot}^T)^2 \,+\, \dis \sum_{i=N_3+1}^{2N_3} (\bar{Y}_{i\cdot\cdot\cdot} - \bar{Y}_{\cdot\cdot\cdot\cdot}^C)^2  \right]\,.
\end{align}
This yields
\begin{align}\label{eq3.2.1.3}
    \lambda_3 \;=\; \frac{1}{2N_3}\, SS_3\,.
\end{align}
Next, observe that
{
\allowdisplaybreaks
\begin{align*}
Z_2 V^{-1} (Y-X\beta) \;&=\; \frac{1}{\lambda_2}\,I_{2N_3} \otimes C_2 \otimes 1_n^T \otimes 1_{N_1}^T \; (Y-X\beta) \;+\; \frac{1}{\lambda_3}\,I_{2N_3} \otimes \bar{J}_2 \otimes 1_n^T \otimes 1_{N_1}^T \; (Y-X\beta)\\
&=:\; C_1 -D_1 \,,
\end{align*}
}
where 
{
\allowdisplaybreaks
\begin{align*}
C_1 \;&=\; \frac{1}{2\lambda_2} \,I_{2N_3} \otimes \begin{pmatrix}
   1_{nN_1}^T & -1_{nN_1}^T \\
   -1_{nN_1}^T & 1_{nN_1}^T
\end{pmatrix}\, Y \;+\; \frac{1}{2\lambda_3} \,I_{2N_3} \otimes \begin{pmatrix}
   1_{nN_1}^T & 1_{nN_1}^T \\
   1_{nN_1}^T & 1_{nN_1}^T
\end{pmatrix}\, Y\\
&=\; \frac{1}{2\lambda_2} \; \,diag\left(\underbrace{\begin{pmatrix}
   1_{nN_1}^T & -1_{nN_1}^T \\
   -1_{nN_1}^T & 1_{nN_1}^T
\end{pmatrix}, \, \dots,\, \begin{pmatrix}
   1_{nN_1}^T & -1_{nN_1}^T \\
   -1_{nN_1}^T & 1_{nN_1}^T
\end{pmatrix}}_{2N_3 \, \textrm{times}}\right)\, Y\\
& \qquad + \; \frac{1}{2\lambda_3} \; \,diag\left(\underbrace{\begin{pmatrix}
   1_{2nN_1}^T  \\
   1_{2nN_1}^T
\end{pmatrix}, \, \dots,\, \begin{pmatrix}
   1_{2nN_1}^T  \\
   1_{2nN_1}^T
\end{pmatrix}}_{2N_3 \, \textrm{times}}\right)\,Y\\
&=\; \frac{1}{2\lambda_2}\,nN_1\, \begin{pmatrix}
    \bar{Y}_{11\cdot\cdot} - \bar{Y}_{12\cdot\cdot}\\
    - (\bar{Y}_{11\cdot\cdot} - \bar{Y}_{12\cdot\cdot})\\
    \vdots \\
    \bar{Y}_{2N_3\,1\cdot\cdot} - \bar{Y}_{2N_3\,2\cdot\cdot}\\
    - (\bar{Y}_{2N_3\,1\cdot\cdot} - \bar{Y}_{2N_3\,2\cdot\cdot})
\end{pmatrix} \;+\; \frac{1}{2\lambda_3}\, 2nN_1\, \begin{pmatrix}
    \bar{Y}_{1\cdot\cdot\cdot}\,1_2\\
    \vdots\\
    \bar{Y}_{2N_3\cdot\cdot\cdot}\,1_2
\end{pmatrix}\,,
\end{align*}
}
and 
{
\allowdisplaybreaks
\begin{align*}
D_1 \;&=\; \frac{nN_1}{2\lambda_2}\, \begin{pmatrix}
    (\tau + \delta)\, 1_{N_3} \otimes \begin{pmatrix}
        1\\
        -1
    \end{pmatrix}\\
    \vdots \\
     \tau\, 1_{N_3} \otimes \begin{pmatrix}
        1\\
        -1
    \end{pmatrix} 
\end{pmatrix} \;+\; \frac{nN_1}{2\lambda_3}\, \begin{pmatrix}
    (2\beta_0 + \tau + 2\xi + \delta) \,1_{N_3} \otimes 1_2 \\
    \vdots \\
    (2\beta_0 + \tau) \,1_{N_3} \otimes 1_2
\end{pmatrix}\\
&=\; \frac{nN_1}{2\lambda_2} \begin{pmatrix}
    (\bar{Y}_{\cdot 1 \cdot \cdot}^T - \bar{Y}_{\cdot 2 \cdot \cdot}^T)\, 1_{N_3} \otimes \begin{pmatrix}
        1\\
        -1
    \end{pmatrix}\\
   (\bar{Y}_{\cdot 1 \cdot \cdot}^C - \bar{Y}_{\cdot 2 \cdot \cdot}^C)\, 1_{N_3} \otimes \begin{pmatrix}
        1\\
        -1
    \end{pmatrix}
\end{pmatrix} \;+\; \frac{2nN_1}{2\lambda_3}\, \begin{pmatrix}
    \bar{Y}_{\cdot \cdot\cdot\cdot}^T\; 1_{2N_3}\\
    \bar{Y}_{\cdot \cdot\cdot\cdot}^C\; 1_{2N_3}
\end{pmatrix}\,.
\end{align*}
}
Thus, to satisfy (\ref{eq3.2.1.1}), we need to have
{
\allowdisplaybreaks
\begin{align*}
\textrm{tr}\left( V^{-1} Z_2 Z_2^T \right) \;&=\; \Vert  Z_2^T V^{-1} (Y-X\beta) \Vert^2\,, \\
\textrm{i.e.,}\qquad \frac{1}{\lambda_2}\, 2N_3 n N_1 \;+\; \frac{1}{\lambda_3}\, 2N_3 n N_1 \;&=\; \frac{(nN_1)^2}{(2\lambda_2)^2}\,2\Big[ \dis \sum_{i=1}^{N_3} \left\{ (\bar{Y}_{i1\cdot\cdot} - \bar{Y}_{i2\cdot\cdot}) - (\bar{Y}_{\cdot 1 \cdot\cdot}^T - \bar{Y}_{\cdot 2 \cdot\cdot})^T \right\}^2 \\
& \qquad \qquad \qquad \;+\; \sum_{i=N_3+1}^{2N_3} \left\{ (\bar{Y}_{i1\cdot\cdot} - \bar{Y}_{i2\cdot\cdot}) - (\bar{Y}_{\cdot 1 \cdot\cdot}^C - \bar{Y}_{\cdot 2 \cdot\cdot}^C) \right\}^2 \Big]\\
& \qquad +\; \frac{(2nN_1)^2}{(2\lambda_3)^2}\; 2\left[ \dis \sum_{i=1}^{N_3} (\bar{Y}_{i\cdot\cdot\cdot} - \bar{Y}_{\cdot\cdot\cdot\cdot}^T)^2 \,+\, \dis \sum_{i=N_3+1}^{2N_3} (\bar{Y}_{i\cdot\cdot\cdot} - \bar{Y}_{\cdot\cdot\cdot\cdot}^C)^2  \right]\,,\\
\textrm{i.e.,}\qquad \frac{2N_3}{\lambda_2} \,+\, \frac{2N_3}{\lambda_3}  \;&=:\; \frac{1}{2\lambda_2^2}\,SS_2 \;+\; \frac{1}{\lambda_3^2}\,SS_3\,,
\end{align*}
}
where 
\begin{align}\label{eq3.2.1.4}
   SS_2 \;&:=\; nN_1\,  \Big[ \dis \sum_{i=1}^{N_3} \left\{ (\bar{Y}_{i1\cdot\cdot} - \bar{Y}_{i2\cdot\cdot}) - (\bar{Y}_{\cdot 1 \cdot\cdot}^T - \bar{Y}_{\cdot 2 \cdot\cdot}^T) \right\}^2 \;+\; \sum_{i=N_3+1}^{2N_3} \left\{ (\bar{Y}_{i1\cdot\cdot} - \bar{Y}_{i2\cdot\cdot}) - (\bar{Y}_{\cdot 1 \cdot\cdot}^C - \bar{Y}_{\cdot 2 \cdot\cdot}^C) \right\}^2 \Big]\,.
\end{align}
This, together with (\ref{eq3.2.1.3}), yields
\begin{align}\label{eq3.2.1.5}
    \frac{2N_3}{\lambda_2} \;=\; \frac{1}{2\lambda_2^2}\, SS_2\,, \qquad i.e., \qquad \lambda_2 \;=\; \frac{1}{2N_3\,2}\, SS_2\,.
\end{align}

\noindent Next observe that 
\begin{align*}
    Z_3^T V^{-1} (Y-X\beta) \;=\; Z_3^T V^{-1} Y \,-\, Z_3^T V^{-1} X\beta \;=:\; C_3 - D_3\,,
\end{align*}
where
{
\allowdisplaybreaks
\begin{align*}
C_3 \;&:=\; \left( \frac{1}{\lambda_1}\,I_{2N_3} \otimes I_2 \otimes C_n \otimes 1_{N_1}^T + \frac{1}{\lambda_2}\,I_{2N_3} \otimes C_2 \otimes \bar{J}_n \otimes 1_{N_1}^T + \frac{1}{\lambda_3}\,I_{2N_3} \otimes \bar{J}_2 \otimes \bar{J}_n \otimes 1_{N_1}^T \right)\,Y\\
&=\; \frac{1}{\lambda_1} \left( I_{2N_3} \otimes I_2 \otimes I_n \otimes 1_{N_1}^T \;-\; \frac{1}{n}\,I_{2N_3} \otimes I_2 \otimes J_n \otimes 1_{N_1}^T \right)\,Y\\
& \qquad +\; \frac{1}{2n\lambda_2}\, I_{2N_3} \otimes \begin{pmatrix}
   1 & -1\\
   -1 & 1
\end{pmatrix} \otimes \left.\begin{pmatrix}
    1_{nN_1}^T\\
    \vdots\\
    1_{nN_1}^T
\end{pmatrix} \right\}{\small{n\,\textrm{times}}}\;\;\;\;\ Y \\
& \qquad +\; \frac{1}{2n\lambda_3}\, I_{2N_3} \otimes \left.\begin{pmatrix}
    1_{2nN_1}^T\\
    \vdots\\
    1_{2nN_1}^T
\end{pmatrix} \right\}{\small{2n\,\textrm{times}}}\;\;\;\; Y \\
&=\; \frac{1}{\lambda_1} \left( I_{2N_3 2n} \otimes 1_{N_1}^T \;-\; \frac{1}{n}\,I_{2N_3 2} \otimes  \left.\begin{pmatrix}
    1_{nN_1}^T\\
    \vdots\\
    1_{nN_1}^T
\end{pmatrix} \right\}{\small{n\,\textrm{times}}} \right)\,Y\\
& \qquad +\; \frac{1}{2n\lambda_2}\, I_{2N_3} \otimes 
\begin{pmatrix}
    \left.\begin{matrix}
        1_{nN_1}^T & -1_{nN_1}^T\\
        \vdots\\
        1_{nN_1}^T & -1_{nN_1}^T
    \end{matrix}\right\}{\small{n\,\textrm{times}}}\\ \\
    \left.\begin{matrix}
        -1_{nN_1}^T & 1_{nN_1}^T\\
        \vdots\\
        -1_{nN_1}^T & 1_{nN_1}^T
    \end{matrix}\right\}{\small{n\,\textrm{times}}}
\end{pmatrix}\,Y \;\;+\;\; \frac{1}{2n\lambda_3}\, I_{2N_3} \otimes \left.\begin{pmatrix}
    1_{2nN_1}^T\\
    \vdots\\
    1_{2nN_1}^T
\end{pmatrix} \right\}{\small{2n\,\textrm{times}}}\;\;\;\; Y\\
&=\; \frac{1}{\lambda_1} \; diag\Big(\underbrace{1_{N_1}^T, \;\dots\;, 1_{N_1}^T }_{2N_3 2n\; \textrm{times}}\Big)\,Y \;\;-\;\; \frac{1}{n\lambda_1} \; diag\left(\underbrace{ \begin{pmatrix}
    1_{nN_1}^T\\
    \vdots\\
    1_{nN_1}^T
\end{pmatrix}\,,\, \dots \,, \left.\begin{pmatrix}
    1_{nN_1}^T\\
    \vdots\\
    1_{nN_1}^T
\end{pmatrix} \right\}{\small{n\,\textrm{times}}}}_{2N_3 2\;\textrm{times}} \right)\,Y \\
& \qquad +\; \frac{1}{2n\lambda_2}\; diag \left(\underbrace{ \begin{pmatrix}
    \left.\begin{matrix}
        1_{nN_1}^T & -1_{nN_1}^T\\
        \vdots\\
        1_{nN_1}^T & -1_{nN_1}^T
    \end{matrix}\right\}{\small{n\,\textrm{times}}}\\ \\
    \left.\begin{matrix}
        -1_{nN_1}^T & 1_{nN_1}^T\\
        \vdots\\
        -1_{nN_1}^T & 1_{nN_1}^T
    \end{matrix}\right\}{\small{n\,\textrm{times}}}
\end{pmatrix}\,, \;\dots\;, \begin{pmatrix}
    \left.\begin{matrix}
        1_{nN_1}^T & -1_{nN_1}^T\\
        \vdots\\
        1_{nN_1}^T & -1_{nN_1}^T
    \end{matrix}\right\}{\small{n\,\textrm{times}}}\\ \\
    \left.\begin{matrix}
        -1_{nN_1}^T & 1_{nN_1}^T\\
        \vdots\\
        -1_{nN_1}^T & 1_{nN_1}^T
    \end{matrix}\right\}{\small{n\,\textrm{times}}}
\end{pmatrix}}_{2N_3\; \textrm{times}} \right)\,Y \\
& \qquad +\; \frac{1}{2n\lambda_3}\; diag \left(\underbrace{ \left.\begin{pmatrix}
    1_{2nN_1}^T\\
    \vdots\\
    1_{2nN_1}^T
\end{pmatrix} \right\}{\small{2n\,\textrm{times}}}\;\;, \; \dots\;, \left.\begin{pmatrix}
    1_{2nN_1}^T\\
    \vdots\\
    1_{2nN_1}^T
\end{pmatrix} \right\}{\small{2n\,\textrm{times}}} }_{2N_3\; \textrm{times}}\right)\,Y \\
&=\; \frac{N_1}{\lambda_1} \begin{pmatrix}
    \bar{Y}_{111\cdot}\\
    \bar{Y}_{112\cdot}\\
    \vdots \\
    \bar{Y}_{2N_3\,2\,n\,\cdot}
\end{pmatrix} \;\;-\;\; \frac{N_1}{\lambda_1} \begin{pmatrix}
\bar{Y}_{11\cdot\cdot}\, 1_n\\
\bar{Y}_{12\cdot\cdot}\, 1_n\\
\vdots\\
\bar{Y}_{2N_3\,1\cdot\cdot}\, 1_n\\
\bar{Y}_{2N_3\,2\cdot\cdot}\, 1_n\\
\end{pmatrix} \;\;+\;\; \frac{N_1}{2\lambda_2} \begin{pmatrix}
   (\bar{Y}_{11\cdot\cdot} - \bar{Y}_{12\cdot\cdot} )\,\begin{pmatrix}
       1_n\\
       -1_n
   \end{pmatrix}\\
   \vdots\\
   (\bar{Y}_{2N_3\,1\cdot\cdot} - \bar{Y}_{2N_3\,2\cdot\cdot} )\,\begin{pmatrix}
       1_n\\
       -1_n
   \end{pmatrix}
\end{pmatrix}\\
& \qquad +\; \frac{N_1}{\lambda_3} \begin{pmatrix}
    \bar{Y}_{1\cdot\cdot\cdot}\, 1_{2n}\\
    \vdots \\
    \bar{Y}_{2N_3\cdot\cdot\cdot}\, 1_{2n}
\end{pmatrix}\,,
\end{align*}
}
and
{
\allowdisplaybreaks
\begin{align*}
D_3 \;&:=\; \frac{N_1}{2\lambda_2} \begin{pmatrix}
    1_{N_3} \otimes\, \begin{matrix}
    \left(\tau + \delta\right)\begin{pmatrix}
    1_{n}\\
    -1_{n}
    \end{pmatrix}
    \end{matrix}\\
    1_{N_3} \otimes \,\begin{matrix}
    \tau \begin{pmatrix}
    1_{n}\\
    -1_{n}
    \end{pmatrix}
    \end{matrix}
\end{pmatrix} \;\;+\;\; \frac{N_1}{2\lambda_3} \begin{pmatrix}
    (2\beta_0 + \tau + 2\xi + \delta)\, 1_{N_3} \otimes 1_{2n}\\
    (2\beta_0 + \tau)\, 1_{N_3} \otimes 1_{2n}
\end{pmatrix}\\
&=\; \frac{N_1}{2\lambda_2} \begin{pmatrix}
    1_{N_3} \otimes\, \begin{matrix}    \left(\bar{Y}_{\cdot 1 \cdot\cdot}^T - \bar{Y}_{\cdot 2 \cdot\cdot}^T\right)\begin{pmatrix}
    1_{n}\\
    -1_{n}
    \end{pmatrix}
    \end{matrix}\\
    1_{N_3} \otimes \,\begin{matrix}
    \left(\bar{Y}_{\cdot 1 \cdot\cdot}^C - \bar{Y}_{\cdot 2 \cdot\cdot}^C\right) \begin{pmatrix}
    1_{n}\\
    -1_{n}
    \end{pmatrix}
    \end{matrix}
\end{pmatrix} \;\;+\;\; \frac{N_1}{2\lambda_3} \begin{pmatrix}
    2\bar{Y}_{\cdot \cdot \cdot\cdot}^T\; 1_{N_3} \otimes 1_{2n}\\
    2\bar{Y}_{\cdot \cdot \cdot\cdot}^C\; 1_{N_3} \otimes 1_{2n}
\end{pmatrix}\,.
\end{align*}
}
To satisfy (\ref{eq3.2.1.1}), we need to have
{
\allowdisplaybreaks
\begin{align*}
\textrm{tr}\left( V^{-1} Z_3 Z_3^T \right) \;=\; \Vert  Z_3^T V^{-1} (Y-X\beta) \Vert^2\,, 
\end{align*}
}
that is,
{
\allowdisplaybreaks
\begin{align*}
&\frac{1}{\lambda_1}\, 2N_3\,2(n-1)\,N_1 \;+\; \frac{1}{\lambda_2}\, 2N_3\,N_1 \;+\; \frac{1}{\lambda_3}\, 2N_3 \,N_1 \\
&=\; \Bigg\Vert\, \frac{N_1}{\lambda_1} \begin{pmatrix}
    \bar{Y}_{111\cdot} - \bar{Y}_{11\cdot\cdot} \\
    \vdots \\
    \bar{Y}_{2N_3\,2\,n\,\cdot} - \bar{Y}_{2N_3\,2\cdot\cdot}
\end{pmatrix} \;+\; \frac{N_1}{2\lambda_2} \begin{pmatrix}
   \left((\bar{Y}_{11\cdot\cdot} - \bar{Y}_{12\cdot\cdot} ) - (\bar{Y}_{\cdot 1\cdot\cdot}^T - \bar{Y}_{\cdot 2\cdot\cdot}^T ) \right)\,\begin{pmatrix}
       1_n\\
       -1_n
   \end{pmatrix}\\
   \vdots\\
   \left((\bar{Y}_{2N_3\,1\cdot\cdot} - \bar{Y}_{2N_3\,2\cdot\cdot} ) - (\bar{Y}_{\cdot \,1\cdot\cdot} - \bar{Y}_{\cdot\,2\cdot\cdot} )\right)\,\begin{pmatrix}
       1_n\\
       -1_n
   \end{pmatrix}
\end{pmatrix}\\
& \qquad +\; \frac{N_1}{\lambda_3} \begin{pmatrix}
    \left(\bar{Y}_{1\cdot\cdot\cdot} - \bar{Y}_{\cdot\cdot\cdot\cdot}^T\right)\, 1_{2n}\\
    \vdots \\
     \left(\bar{Y}_{N_3\cdot\cdot\cdot} - \bar{Y}_{\cdot\cdot\cdot\cdot}^T\right)\, 1_{2n}\\
     \vdots\\
    \left(\bar{Y}_{N_3+1\,\cdot\cdot\cdot} - \bar{Y}_{\cdot\cdot\cdot\cdot}^C\right)\, 1_{2n}\\
    \vdots \\
     \left(\bar{Y}_{2N_3\cdot\cdot\cdot} - \bar{Y}_{\cdot\cdot\cdot\cdot}^C\right)\, 1_{2n}
\end{pmatrix}\, \Bigg\Vert^2\\
&=\; \frac{N_1^2}{\lambda_1^2} \,\dis \sum_{i=1}^{2N_3}\sum_{g=1}^2 \sum_{j=1}^n \left( \bar{Y}_{igj \cdot} - \bar{Y}_{ig\cdot\cdot}\right)^2 \;+\; \frac{N_1}{2\lambda_2^2}\, SS_2 \;+\;  \frac{N_1}{\lambda_3^2}\, SS_3\,.
\end{align*}
}
\noindent This, together with (\ref{eq3.2.1.3}) and (\ref{eq3.2.1.5}), imply that
\begin{align}\label{eq3.2.1.6}
    \frac{2N_3\,2(n-1)}{\lambda_1} \;=\; \frac{SS_1}{\lambda_1^2}\,, \qquad \textrm{where}\;\; SS_1 \;=\; N_1\,\dis \sum_{i=1}^{2N_3}\sum_{g=1}^2 \sum_{j=1}^n \left( \bar{Y}_{igj \cdot} - \bar{Y}_{ig\cdot\cdot}\right)^2\,.
\end{align}
This yields
\begin{align}\label{eq3.2.1.7}
    \lambda_1 \;=\; \frac{1}{2N_3\,2(n-1)}\, SS_1\,.
\end{align}

\noindent And finally, note that
{
\allowdisplaybreaks
\begin{align*}
Z_0 V^{-1} \;=\; V^{-1} \;&=\; \frac{1}{\lambda_0}\,I_{2N_3\,2n} \otimes \left(I_{N_1} - \frac{1}{N_1} J_{N_1} \right) \;+\; \frac{1}{N_1\,\lambda_1}\,I_{2N_3\,2} \otimes \left( I_n - \frac{1}{n} J_n \right) \otimes J_{N_1}\\
   &\qquad + \frac{1}{2nN_1\,\lambda_2}\,I_{2N_3} \otimes \begin{pmatrix}
       1 & -1\\
       -1 & 1 
   \end{pmatrix} \otimes J_{n N_1} \;+\; \frac{1}{2nN_1\,\lambda_3}\,I_{2N_3} \otimes J_{2nN_1} \\
   &=\; \frac{1}{\lambda_0}\,I_{2N_3 2n N_1} - \frac{1}{N_1 \lambda_0}\,I_{2N_3 2n} \otimes J_{N_1} + \frac{1}{N_1 \lambda_1}\,I_{2N_3 2n} \otimes J_{N_1} - \frac{1}{nN_1 \lambda_1}\,I_{2N_3 2} \otimes J_{nN_1}\\
   & \qquad + \frac{1}{2nN_1\,\lambda_2}\,I_{2N_3} \otimes \begin{pmatrix}
       J_{n N_1} & -J_{n N_1}\\
       -J_{n N_1} & J_{n N_1}
   \end{pmatrix} +  \frac{1}{2nN_1\,\lambda_3}\,I_{2N_3} \otimes J_{2nN_1}\\
   &=\; \frac{1}{\lambda_0}\,I_{2N_3 2n N_1} - \frac{1}{N_1 \lambda_0}\; \,diag\left(\underbrace{J_{N_1}, \dots, J_{N_1}}_{2N_3 2n \; \textrm{times}}\right) + \frac{1}{N_1 \lambda_1}\; \,diag\left(\underbrace{J_{N_1}, \dots, J_{N_1}}_{2N_3 2n \; \textrm{times}}\right)\\
   & \qquad - \frac{1}{nN_1 \lambda_1}\; \,diag\left(\underbrace{J_{nN_1}, \dots, J_{nN_1}}_{2N_3 2 \; \textrm{times}}\right)\\
   & \qquad + \frac{1}{2nN_1\,\lambda_2}\; \,diag\left(\underbrace{ \begin{pmatrix}
     J_{n N_1} & -J_{n N_1}\\
    -J_{n N_1} & J_{n N_1}  
   \end{pmatrix}, \dots, \begin{pmatrix}
     J_{n N_1} & -J_{n N_1}\\
    -J_{n N_1} & J_{n N_1}  
   \end{pmatrix}}_{2N_3 \; \textrm{times}}\right)\\
   &\qquad +  \frac{1}{2nN_1\,\lambda_3}\; \,diag\left(\underbrace{J_{2nN_1}, \dots, J_{2nN_1}}_{2N_3 \; \textrm{times}}\right)\,,
\end{align*}
}
which yields
{
\allowdisplaybreaks
\begin{align*}
Z_0 V^{-1} Y \;&=\; \frac{1}{\lambda_0}\, \bar{Y}_{\cdot \cdot \cdot \cdot} \;-\;  \frac{1}{\lambda_0}\, \begin{pmatrix}
    \bar{Y}_{111\cdot}\,1_{N_1}\\
    \vdots \\
    \bar{Y}_{11n\cdot}\,1_{N_1}\\
    \bar{Y}_{121\cdot}\,1_{N_1}\\
    \vdots \\
    \bar{Y}_{12n\cdot}\,1_{N_1}\\
    \vdots\\ \vdots \\
    \bar{Y}_{2N_3 11\cdot}\,1_{N_1}\\
    \vdots \\
    \bar{Y}_{2N_3 1n\cdot}\,1_{N_1}\\
    \vdots\\
    \bar{Y}_{2N_3 21\cdot}\,1_{N_1}\\
    \vdots \\
    \bar{Y}_{2N_3 2n\cdot}\,1_{N_1}
\end{pmatrix} \;+\; \frac{1}{\lambda_1}\, \begin{pmatrix}
    \bar{Y}_{111\cdot}\,1_{N_1}\\
    \vdots \\
    \bar{Y}_{11n\cdot}\,1_{N_1}\\
    \bar{Y}_{121\cdot}\,1_{N_1}\\
    \vdots \\
    \bar{Y}_{12n\cdot}\,1_{N_1}\\
    \vdots \\ \vdots \\
    \bar{Y}_{2N_3 11\cdot}\,1_{N_1}\\
    \vdots \\
    \bar{Y}_{2N_3 1n\cdot}\,1_{N_1}\\
    \bar{Y}_{2N_3 21\cdot}\,1_{N_1}\\
    \vdots \\
    \bar{Y}_{2N_3 2n\cdot}\,1_{N_1}
\end{pmatrix} \;-\; \frac{1}{\lambda_1}\, \begin{pmatrix}
    \bar{Y}_{11\cdot\cdot}\,1_{nN_1}\\
    \bar{Y}_{12\cdot\cdot}\,1_{nN_1}\\
    \vdots \\
    \bar{Y}_{2N_3\,1\cdot\cdot}\,1_{nN_1}\\
    \bar{Y}_{2N_3\,2\cdot\cdot}\,1_{nN_1}\\
    \vdots
\end{pmatrix}\\
& \;\; + \; \frac{1}{2\lambda_2}\, \begin{pmatrix}
    \left(\bar{Y}_{11\cdot\cdot} - \bar{Y}_{12\cdot\cdot}\right)\,\begin{pmatrix}
    1_{nN_1}\\
    -1_{nN_1}
    \end{pmatrix}\\
    \vdots \\
    \left(\bar{Y}_{2N_3\,1\cdot\cdot} - \bar{Y}_{2N_3\,2\cdot\cdot}\right)\,\begin{pmatrix}
    1_{nN_1}\\
    -1_{nN_1}
    \end{pmatrix}
\end{pmatrix} \;+\; \frac{1}{\lambda_3}\,\begin{pmatrix}
    \bar{Y}_{1\cdot\cdot\cdot}\,1_{2nN_1}\\
    \vdots\\
    \bar{Y}_{2N_3\cdot\cdot\cdot}\,1_{2nN_1}
\end{pmatrix}
\end{align*}
}
and 
{
\allowdisplaybreaks
\begin{align*}
Z_0 V^{-1} X\beta \;&=\; \frac{1}{\lambda_0} 
\begin{pmatrix}
1_{N_3} \otimes \begin{pmatrix}
    (\beta_0 + \tau + \xi + \delta)\, 1_{nN_1}\\
    (\beta_0 + \xi)\, 1_{nN_1}
\end{pmatrix}\\ 
    1_{N_3} \otimes \begin{pmatrix}
    (\beta_0 + \tau)\, 1_{nN_1}\\
    \beta_0 \,1_{nN_1}
    \end{pmatrix}
\end{pmatrix}
\;-\; \frac{1}{\lambda_0} \begin{pmatrix}
1_{N_3} \otimes \begin{pmatrix}
    (\beta_0 + \tau + \xi + \delta)\, 1_{nN_1}\\
    (\beta_0 + \xi)\, 1_{nN_1}
    \end{pmatrix}\\
    1_{N_3} \otimes \begin{pmatrix}
    (\beta_0 + \tau)\, 1_{nN_1}\\
    \beta_0 \,1_{nN_1}
    \end{pmatrix}
\end{pmatrix}\\
& \qquad +\; \frac{1}{\lambda_1} \begin{pmatrix}
    1_{N_3} \otimes \begin{pmatrix}
    (\beta_0 + \tau + \xi + \delta)\, 1_{nN_1}\\
    (\beta_0 + \xi)\, 1_{nN_1}
    \end{pmatrix}\\
    1_{N_3} \otimes \begin{pmatrix}
    (\beta_0 + \tau)\, 1_{nN_1}\\
    \beta_0 \,1_{nN_1}
    \end{pmatrix}
\end{pmatrix} \;-\; \frac{1}{\lambda_1} \begin{pmatrix}
     1_{N_3} \otimes \begin{pmatrix}
     (\beta_0 + \tau + \xi + \delta)\, 1_{nN_1}\\
    (\beta_0 + \xi)\, 1_{nN_1}
    \end{pmatrix}\\
     1_{N_3} \otimes \begin{pmatrix}
    (\beta_0 + \tau)\, 1_{nN_1}\\
    \beta_0 \,1_{nN_1}
    \end{pmatrix}
\end{pmatrix}\\
& \qquad +\; \frac{1}{2\lambda_2}\, \begin{pmatrix}
    1_{N_3} \otimes\, \begin{matrix}
    \left(\tau + \delta\right)\begin{pmatrix}
    1_{nN_1}\\
    -1_{nN_1}
    \end{pmatrix}
    \end{matrix}\\
    1_{N_3} \otimes \,\begin{matrix}
    \tau \begin{pmatrix}
    1_{nN_1}\\
    -1_{nN_1}
    \end{pmatrix}
    \end{matrix}
\end{pmatrix} \;+\; \frac{1}{2\lambda_3} \begin{pmatrix}
    (2\beta_0 + \tau + 2\xi + \delta)\, 1_{N_3} \otimes 1_{2nN_1}\\
    (2\beta_0 + \tau)\, 1_{N_3} \otimes 1_{2nN_1}
\end{pmatrix}\\
&=\;  \frac{1}{2\lambda_2}\, \begin{pmatrix}
    1_{N_3} \otimes\;\left(\bar{Y}_{\cdot 1 \cdot\cdot}^T - \bar{Y}_{\cdot 2 \cdot\cdot}^T \right)\,\begin{pmatrix}
    1_{nN_1}\\
    -1_{nN_1}
    \end{pmatrix}\\
    1_{N_3} \otimes\;  \left(\bar{Y}_{\cdot 1 \cdot\cdot}^C - \bar{Y}_{\cdot 2 \cdot\cdot}^C \right)\,\begin{pmatrix}
    1_{nN_1}\\
    -1_{nN_1}
    \end{pmatrix}
\end{pmatrix} \;+\; \frac{1}{\lambda_3} \begin{pmatrix}
    \bar{Y}_{\cdot \cdot \cdot\cdot}^T\; 1_{N_3} \otimes 1_{2nN_1}\\
    \bar{Y}_{\cdot \cdot \cdot\cdot}^C\; 1_{N_3} \otimes 1_{2nN_1}
\end{pmatrix}\;.
\end{align*}
}
Therefore we have
{
\allowdisplaybreaks
\begin{align*}
&\Vert Z_0 V^{-1} (Y-X\beta) \Vert^2 \;=\; \Vert -\,\frac{1}{\lambda_0}\, \begin{pmatrix}
    (\bar{Y}_{111\cdot} - \bar{Y}_{\cdot \cdot \cdot\cdot}) \,1_{N_1}\\
    \vdots \\
    (\bar{Y}_{2N_3\,2\,n\cdot} - \bar{Y}_{\cdot \cdot \cdot\cdot})\,1_{N_1}\\
\end{pmatrix} \;+\; \frac{1}{\lambda_1}\, \begin{pmatrix}
    (\bar{Y}_{111\cdot} - \bar{Y}_{11 \cdot\cdot}) \,1_{N_1}\\
    \vdots \\
    (\bar{Y}_{11n\cdot} - \bar{Y}_{11 \cdot\cdot}) \,1_{N_1}\\
    \vdots \\
    (\bar{Y}_{2N_3\,2\,1\cdot} - \bar{Y}_{2N_3\,2 \cdot\cdot})\,1_{N_1}\\
    \vdots\\
    (\bar{Y}_{2N_3\,2\,n\cdot} - \bar{Y}_{2N_3\,2 \cdot\cdot})\,1_{N_1}\\
\end{pmatrix}\\
&\qquad + \frac{1}{2\lambda_2}\, \begin{pmatrix}
    \left\{(\bar{Y}_{11 \cdot\cdot} - \bar{Y}_{12 \cdot\cdot}) - (\bar{Y}_{\cdot 1 \cdot\cdot}^T - \bar{Y}_{\cdot 2 \cdot\cdot}^T)\right\}\,\begin{pmatrix}
    1_{nN_1}\\
    -1_{nN_1}
    \end{pmatrix}\\
    \vdots \\
    \left\{(\bar{Y}_{2N_3\,1 \cdot\cdot} - \bar{Y}_{2N_3\,2 \cdot\cdot}) - (\bar{Y}_{\cdot 1 \cdot\cdot}^C - \bar{Y}_{\cdot 2 \cdot\cdot}^C)\right\}\,\begin{pmatrix}
    1_{nN_1}\\
    -1_{nN_1}
    \end{pmatrix}
\end{pmatrix} \; +\; \frac{1}{\lambda_3} \begin{pmatrix}
    (\bar{Y}_{1\cdot \cdot\cdot} - \bar{Y}_{\cdot \cdot \cdot\cdot}^T)\, 1_{2nN_1}\\
    (\bar{Y}_{2N_3\cdot \cdot\cdot} - \bar{Y}_{\cdot \cdot \cdot\cdot}^C)\, 1_{2nN_1} 
\end{pmatrix} \; \Vert^2\\
&=\; \frac{1}{\lambda_0^2}\,N_1\, \dis \sum_{i=1}^{2N_3}\sum_{g=1}^2 \sum_{j=1}^n \left( \bar{Y}_{igj\cdot} - \bar{Y}_{\cdot\cdot \cdot\cdot} \right)^2 \;+\; \frac{1}{\lambda_1^2}\,N_1\, \dis \sum_{i=1}^{2N_3}\sum_{g=1}^2 \sum_{j=1}^n \left( \bar{Y}_{igj\cdot} - \bar{Y}_{ig \cdot\cdot} \right)^2\\
& \qquad +\; \frac{1}{2\lambda_2^2}\; nN_1\,  \Big[ \dis \sum_{i=1}^{N_3} \left\{ (\bar{Y}_{i1\cdot\cdot} - \bar{Y}_{i2\cdot\cdot}) - (\bar{Y}_{\cdot 1 \cdot\cdot}^T - \bar{Y}_{\cdot 2 \cdot\cdot}^T) \right\}^2 \;+\; \sum_{i=N_3+1}^{2N_3} \left\{ (\bar{Y}_{i1\cdot\cdot} - \bar{Y}_{i2\cdot\cdot}) - (\bar{Y}_{\cdot 1 \cdot\cdot}^C - \bar{Y}_{\cdot 2 \cdot\cdot}^C) \right\}^2 \Big]\\
& \qquad +\; \frac{1}{\lambda_3^2}\; 2n\,N_1\, \left[ \dis \sum_{i=1}^{N_3} (\bar{Y}_{i\cdot\cdot\cdot} - \bar{Y}_{\cdot\cdot\cdot\cdot}^T)^2 \,+\, \dis \sum_{i=N_3+1}^{2N_3} (\bar{Y}_{i\cdot\cdot\cdot} - \bar{Y}_{\cdot\cdot\cdot\cdot}^C)^2  \right]\\
&=\; \frac{SS_0}{\lambda_0^2} \;+\; \frac{SS_1}{\lambda_1^2} \;+\; \frac{SS_2}{2\lambda_2^2} \;+\; \frac{SS_3}{\lambda_3^2}\,,
\end{align*}
}
where
\begin{align}\label{eq3.2.1.8}
    SS_0 \;:=\; N_1\, \dis \sum_{i=1}^{2N_3}\sum_{g=1}^2 \sum_{j=1}^n \left( \bar{Y}_{igj\cdot} - \bar{Y}_{\cdot\cdot \cdot\cdot} \right)^2\,.
\end{align}
To satisfy (\ref{eq3.2.1.1}), we need to have
{
\allowdisplaybreaks
\begin{align*}
\textrm{tr}\left( V^{-1} Z_0 Z_0^T \right) \;&=\; \Vert  Z_0^T V^{-1} (Y-X\beta) \Vert^2\,.  
\end{align*}
}
This yields
{
\allowdisplaybreaks
\begin{align*}
&\frac{1}{\lambda_0}\, 2N_3\, 2n\, (N_1 -1) \;+\; \frac{1}{\lambda_1}\, 2N_3\, 2\,(n-1) \;+\; \frac{1}{\lambda_2}\, 2N_3 \;+\; \frac{1}{\lambda_3}\, 2N_3\\
&=\; \frac{SS_0}{\lambda_0^2} \;+\; \frac{SS_1}{\lambda_1^2} \;+\; \frac{SS_2}{2\lambda_2^2} \;+\; \frac{SS_3}{\lambda_3^2}\,.
\end{align*}
}
With this, and combining (\ref{eq3.2.1.3}), (\ref{eq3.2.1.5}) and (\ref{eq3.2.1.7}), we have
\begin{align}\label{eq3.2.1.9}
\frac{2N_3\,2n\,(N_1-1)}{\lambda_0} \;=\; \frac{SS_0}{\lambda_0^2}\,, \qquad \textrm{i.e.,} \qquad \lambda_0 \;=\; \frac{1}{2N_3\,2n\,(N_1-1)}\,SS_0.
\end{align}
Thus the ML estimates of the variance components are:
{
\allowdisplaybreaks
\begin{align}
\hat{\lambda}_0 \;&=\; \hat{\sigma}_e^2 \;=\;\frac{1}{2N_3\,2n\,(N_1-1)}\,SS_0 \;=\; \frac{1}{2N_3\,2n\,(N_1-1)}\; N_1\, \dis \sum_{i=1}^{2N_3}\sum_{g=1}^2 \sum_{j=1}^n \left( \bar{Y}_{igj\cdot} - \bar{Y}_{\cdot\cdot \cdot\cdot} \right)^2\,, \label{eq3.2.1.10} \\
\hat{\lambda}_1 \;&=\; \hat{\sigma}_e^2 + N_1\, \hat{\sigma}_1^2 \;=\; \frac{1}{2N_3\,2(n-1)}\,SS_1 \;=\; \frac{1}{2N_3\,2(n-1)}\, N_1\,\dis \sum_{i=1}^{2N_3}\sum_{g=1}^2 \sum_{j=1}^n \left( \bar{Y}_{igj\cdot} - \bar{Y}_{ig \cdot\cdot} \right)^2\,, \label{eq3.2.1.11} \\
\hat{\lambda}_2 \;&=\; \hat{\sigma}_e^2 + N_1\, \hat{\sigma}_1^2 + nN_1\, \hat{\sigma}_2^2 \;=\; \frac{1}{2N_3\,2}\,SS_2 \notag \\
&=\; \frac{1}{2N_3\,2}\; nN_1\, \left[ \dis \sum_{i=1}^{N_3} \left\{ (\bar{Y}_{i1\cdot\cdot} - \bar{Y}_{i2\cdot\cdot}) - (\bar{Y}_{\cdot 1 \cdot\cdot}^T - \bar{Y}_{\cdot 2 \cdot\cdot}^T) \right\}^2 \;+\; \sum_{i=N_3+1}^{2N_3} \left\{ (\bar{Y}_{i1\cdot\cdot} - \bar{Y}_{i2\cdot\cdot}) - (\bar{Y}_{\cdot 1 \cdot\cdot}^C - \bar{Y}_{\cdot 2 \cdot\cdot}^C) \right\}^2 \right] \label{eq3.2.1.12} \\
\hat{\lambda}_3 \;&=\; \hat{\sigma}_e^2 + N_1\, \hat{\sigma}_1^2 + nN_1\, \hat{\sigma}_2^2 + 2nN_1\, \hat{\sigma}_3^2 \;=\; \frac{1}{2N_3}\,SS_3 \notag \\
&=\; \frac{1}{2N_3}\; 2nN_1\,\left[ \dis \sum_{i=1}^{N_3} (\bar{Y}_{i\cdot\cdot\cdot} - \bar{Y}_{\cdot\cdot\cdot\cdot}^T)^2 \,+\, \dis \sum_{i=N_3+1}^{2N_3} (\bar{Y}_{i\cdot\cdot\cdot} - \bar{Y}_{\cdot\cdot\cdot\cdot}^C)^2  \right]\,. \label{eq3.2.1.13}
\end{align}
}
Subtracting (\ref{eq3.2.1.10}) from (\ref{eq3.2.1.11}), we get
\begin{align}\label{eq3.2.1.14}
  \hat{\sigma}_1^2 \;=\; \frac{1}{N_1}\, \left(  \frac{1}{2N_3\,2(n-1)}\,SS_1 \;-\; \frac{1}{2N_3\,2n\,(N_1-1)}\,SS_0 \right)\,.
\end{align}
Subtracting (\ref{eq3.2.1.11}) from (\ref{eq3.2.1.12}), we get
\begin{align}\label{eq3.2.1.15}
  \hat{\sigma}_2^2 \;=\; \frac{1}{nN_1}\, \left(  \frac{1}{2N_3\,2}\,SS_2 \;-\; \frac{1}{2N_3\,2(n-1)}\,SS_1 \right)\,.
\end{align}
And finally, subtracting (\ref{eq3.2.1.12}) from (\ref{eq3.2.1.13}), we get
\begin{align}\label{eq3.2.1.16}
  \hat{\sigma}_3^2 \;=\; \frac{1}{2nN_1}\, \left(  \frac{1}{2N_3}\,SS_3 \;-\; \frac{1}{2N_3\,2}\,SS_2 \right)\,.
\end{align}
This completes the proof of the theorem.

\end{proof}

\begin{proof}[Proof of Theorem \ref{Thm_consistency_level2}]
We first show that the variance component estimators obtained in Theorem \ref{Thm_MLE_var_comps2} are unbiased of their true parameter counterparts. We will only sketch a proof for $E(\hat{\lambda}_3) = \lambda_3$; similar arguments will show that $\hat{\lambda}_0, \hat{\lambda}_1$ and $\hat{\lambda}_2$ are unbiased as well, which will in turn imply that $\hat{\sigma}_e^2, \hat{\sigma}_1^2, \hat{\sigma}_2^2$ and $\hat{\sigma}_3^2$ are unbiased estimators.

Towards that, note that from (\ref{eq3.2.1.0.6}) we have
{
\allowdisplaybreaks
\begin{align*}
  \frac{2nN_1}{\lambda_3^2}\, SS_3 \;&=\; \Vert  Z_1^T V^{-1} (Y-X\beta)\Vert^2  \;=\; \Vert (Y-X\beta)^T V^{-1} Z_1 \Vert^2\\
  &=\; (Y-X\beta)^T V^{-1} Z_1 Z_1^T V^{-1} (Y-X\beta) \;=\; (Y-X\beta)^T D\, (Y-X\beta)^T\,,
\end{align*}
}
where $D:=V^{-1} Z_1 Z_1^T V^{-1}$. With this, and observing the fact that $Y \sim N(X\beta, V)$, we have
{
\allowdisplaybreaks
\begin{align*}
E\left( \frac{2nN_1}{\lambda_3^2}\, SS_3 \right) \;&=\; E\left( (Y-X\beta)^T D (Y-X\beta)^T\right) \;=\; \textrm{tr}(DV) \;=\; \textrm{tr}(V^{-1} Z_1 Z_1^T)\\
&=\; \frac{2N_3\,2n\,N_1}{\lambda_3}\,,
\end{align*}
}
where the last inequality follows from (\ref{eq3.2.1.0.5}). Some simple algebraic manipulations yield that
\begin{align*}
    E\left( \frac{SS_3}{2N_3}\right) \;=\; E(\hat{\lambda}_3) \;=\; \lambda_3\,,
\end{align*}
which completes the proof that $\hat{\lambda}_3$ is an unbiased estimator of $\lambda_3$.\\

Therefore to show the consistency of the variance component estimators, it is enough to show that the variances of $\hat{\sigma}_e^2, \hat{\sigma}_1^2, \hat{\sigma}_2^2$ and $\hat{\sigma}_3^2$ are asymptotically negligible. Towards that, note that by Section 6.3 of \cite{searle2009}, the asymptotic covariance matrix of $(\hat{\sigma_e}^2, \hat{\sigma_3}^2, \hat{\sigma_2}^2, \hat{\sigma_1}^2)$\, is\, $2B^{-1}$, where 
\begin{align*}
    B \;=\; \begin{pmatrix}
        b_{00} & b_{01} & b_{02} & b_{03}\\
        b_{01} & b_{11} & b_{12} & b_{13}\\
        b_{02} & b_{12} & b_{22} & b_{23}\\
        b_{03} & b_{13} & b_{23} & b_{33}
    \end{pmatrix}\,,
\end{align*}
with $b_{ij} := \textrm{tr}(V^{-1} Z_i Z_i^T\, V^{-1} Z_j Z_j^T)$\, for\, $0\leq i\leq j \leq 3$. Now observe that
{
\allowdisplaybreaks
\begin{align*}
   V^{-1} Z_0 Z_0^T \;&=\; V^{-1} \;=\;  \frac{1}{\lambda_0}\,I_{2N_3} \otimes I_2 \otimes I_n \otimes C_{N_1} \;+\; \frac{1}{\lambda_1}\,I_{2N_3} \otimes I_2 \otimes C_n \otimes \bar{J}_{N_1}\\
   &\qquad \qquad + \frac{1}{\lambda_2}\,I_{2N_3} \otimes C_2 \otimes \bar{J}_n \otimes \bar{J}_{N_1} \;+\; \frac{1}{\lambda_3}\,I_{2N_3} \otimes \bar{J}_2 \otimes \bar{J}_n \otimes \bar{J}_{N_1}\,,\\
   V^{-1} Z_1 Z_1^T \;&=\; \frac{1}{\lambda_3}\,I_{2N_3} \otimes J_2 \otimes J_n \otimes J_{N_1} \,,\\
   V^{-1} Z_2 Z_2^T \;&=\; \frac{1}{\lambda_2}\,I_{2N_3} \otimes C_2 \otimes J_n \otimes J_{N_1} \;+\; \frac{1}{\lambda_3}\,I_{2N_3} \otimes \bar{J}_2 \otimes J_n \otimes J_{N_1} \,,\\
   \textrm{and} \;\; V^{-1} Z_3 Z_3^T \;&=\; \frac{1}{\lambda_1}\,I_{2N_3} \otimes I_2 \otimes C_n \otimes J_{N_1} \;+\; \frac{1}{\lambda_2}\,I_{2N_3} \otimes C_2 \otimes \bar{J}_n \otimes J_{N_1} \;+\; \frac{1}{\lambda_3}\,I_{2N_3} \otimes \bar{J}_2 \otimes \bar{J}_n \otimes J_{N_1} \,.\\
\end{align*}
}
Therefore
{
\allowdisplaybreaks
\begin{align*}
   b_{00} \;&=\; \textrm{tr}(V^{-1} Z_0 Z_0^T\, V^{-1} Z_0 Z_0^T) \;=\; \textrm{tr} \big(\, \frac{1}{\lambda_0^2}\,I_{2N_3} \otimes I_2 \otimes I_n \otimes C_{N_1} \,+\, \frac{1}{\lambda_1^2}\,I_{2N_3} \otimes I_2 \otimes C_n \otimes \bar{J}_{N_1}\\
   &\qquad \qquad \qquad \qquad \qquad \qquad \qquad + \frac{1}{\lambda_2^2}\,I_{2N_3} \otimes C_2 \otimes \bar{J}_n \otimes \bar{J}_{N_1} \,+\, \frac{1}{\lambda_3^2}\,I_{2N_3} \otimes \bar{J}_2 \otimes \bar{J}_n \otimes \bar{J}_{N_1}\,\big)\\
   &=\; \frac{1}{\lambda_0^2}\, 2N_3\, 2n\, (N_1 -1) \;+\; \frac{1}{\lambda_1^2}\, 2N_3\, 2\,(n-1) \;+\; \frac{1}{\lambda_2^2}\, 2N_3 \;+\; \frac{1}{\lambda_3^2}\, 2N_3\,,\\
   b_{11} \;&=\; \textrm{tr}(V^{-1} Z_1 Z_1^T\, V^{-1} Z_1 Z_1^T) \;=\; \textrm{tr} \big(\,\frac{2n\,N_1}{\lambda_3^2}\,I_{2N_3} \otimes J_2 \otimes J_n \otimes J_{N_1}\,\big) \;=\; \frac{1}{\lambda_3^2}\,2N_3\,(2n N_1)^2\,,\\
   b_{22} \;&=\; \textrm{tr}(V^{-1} Z_2 Z_2^T\, V^{-1} Z_2 Z_2^T) \;=\; \textrm{tr} \big(\,\frac{n\,N_1}{\lambda_2^2}\,I_{2N_3} \otimes C_2 \otimes J_n \otimes J_{N_1} \;+\; \frac{n\,N_1}{\lambda_3^2}\,I_{2N_3} \otimes \bar{J}_2 \otimes J_n \otimes J_{N_1}\,\big)\\
   &=\; \frac{1}{\lambda_2^2}\, 2N_3\,(n N_1)^2 \;+\; \frac{1}{\lambda_3^2}\, 2N_3\,(n N_1)^2 \;=\; 2N_3\,(n N_1)^2\, \left(\frac{1}{\lambda_2^2} + \frac{1}{\lambda_3^2}\right)\,,\\
   b_{33} \;&=\; \textrm{tr}(V^{-1} Z_3 Z_3^T\, V^{-1} Z_3 Z_3^T) \;=\; \textrm{tr} \big(\,\frac{N_1}{\lambda_1^2}\,I_{2N_3} \otimes I_2 \otimes C_n \otimes J_{N_1} \;+\; \frac{N_1}{\lambda_2^2}\,I_{2N_3} \otimes C_2 \otimes \bar{J}_n \otimes J_{N_1} \\
   & \qquad \qquad \qquad \qquad \qquad \qquad \qquad \;\; \;+\; \frac{N_1}{\lambda_3^2}\,I_{2N_3} \otimes \bar{J}_2 \otimes \bar{J}_n \otimes J_{N_1}\,\big)\\
   &=\; \frac{1}{\lambda_1^2}\, 2N_3\,2(n-1)\,N_1^2 \;+\; \frac{1}{\lambda_2^2}\, 2N_3\,N_1^2 \;+\; \frac{1}{\lambda_3^2}\, 2N_3 \,N_1^2\,,\\
   b_{01} \;&=\; \textrm{tr}(V^{-1} Z_0 Z_0^T\, V^{-1} Z_1 Z_1^T) \;=\; \textrm{tr} \big(\, 0 + 0 + 0 + \frac{1}{\lambda_3^2}\,I_{2N_3} \otimes J_2 \otimes J_n \otimes J_{N_1}\,\big) \;=\; \frac{1}{\lambda_3^2}\, 2N_3\,2n\,N_1\,,\\
   b_{02} \;&=\; \textrm{tr}(V^{-1} Z_0 Z_0^T\, V^{-1} Z_2 Z_2^T) \;=\; \textrm{tr} \big(\,0 + 0 + \frac{1}{\lambda_2^2}\,I_{2N_3} \otimes C_2 \otimes J_n \otimes J_{N_1} + 0 + 0 + 0\\
   & \qquad \qquad \qquad \qquad \qquad \qquad \qquad \;\; + \frac{1}{\lambda_3^2}\,I_{2N_3} \otimes \bar{J}_2 \otimes J_n \otimes J_{N_1}\,\big)\\
   &=\; \frac{1}{\lambda_2^2}\, 2N_3\,n\,N_1 \;+\; \frac{1}{\lambda_3^2}\, 2N_3\,n\,N_1\,,\\
   b_{03} \;&=\; \textrm{tr}(V^{-1} Z_0 Z_0^T\, V^{-1} Z_3 Z_3^T) \;=\; \textrm{tr} \big(\,\frac{1}{\lambda_1^2}\,I_{2N_3} \otimes I_2 \otimes C_n \otimes J_{N_1} \;+\; \frac{1}{\lambda_2^2}\,I_{2N_3} \otimes C_2 \otimes \bar{J}_n \otimes J_{N_1} \\
   & \qquad \qquad \qquad \qquad \qquad \qquad \qquad \;\; +\; \frac{1}{\lambda_3^2}\,I_{2N_3} \otimes \bar{J}_2 \otimes \bar{J}_n \otimes J_{N_1}\,\big)\\
   &=\; \frac{1}{\lambda_1^2}\, 2N_3\, 2\,(n-1)\,N_1 \;+\; \frac{1}{\lambda_2^2}\, 2N_3\,N_1 \;+\; \frac{1}{\lambda_3^2}\, 2N_3\,N_1\,,\\
   b_{12} \;&=\; \textrm{tr}(V^{-1} Z_1 Z_1^T\, V^{-1} Z_2 Z_2^T) \;=\; \textrm{tr} \big(\,\frac{n N_1}{\lambda_3^2}\,I_{2N_3} \otimes J_2 \otimes J_n \otimes J_{N_1}\,\big) \;=\; \frac{1}{\lambda_3^2}\,2N_3\,2\,(n N_1)^2\,,\\
   b_{13} \;&=\; \textrm{tr}(V^{-1} Z_1 Z_1^T\, V^{-1} Z_3 Z_3^T) \;=\; \textrm{tr} \big(\,\frac{N_1}{\lambda_3^2}\,I_{2N_3} \otimes J_2 \otimes J_n \otimes J_{N_1}\,\big) \;=\; \frac{1}{\lambda_3^2}\,2N_3\,2n\, N_1^2\,,\\
   b_{23} \;&=\; \textrm{tr}(V^{-1} Z_2 Z_2^T\, V^{-1} Z_3 Z_3^T) \;=\; \textrm{tr} \big(\,\frac{N_1}{\lambda_2^2}\,I_{2N_3} \otimes C_2 \otimes J_n \otimes J_{N_1} \;+\; \frac{N_1}{\lambda_3^2}\,I_{2N_3} \otimes \bar{J}_2 \otimes J_n \otimes J_{N_1}\,\big)\\
   &=\; \frac{1}{\lambda_2^2}\,2N_3\,n\,N_1^2 \;+\; \frac{1}{\lambda_3^2}\,2N_3\,n\,N_1^2\,.
\end{align*}
}
Now $B^{-1} = \frac{1}{\vert B \vert}\, adj(B)$. Let 
\begin{align*}
    adj(B) \;&:=\; \begin{pmatrix}
       B_{00} & -B_{01} & B_{02} & -B_{03}\\
        -B_{01} & B_{11} & -B_{12} & B_{13}\\
        B_{02} & -B_{12} & B_{22} & -B_{23}\\
        -B_{03} & B_{13} & -B_{23} & B_{33} 
    \end{pmatrix}\,.
\end{align*}
Then we have
\begin{align}\label{supp_cons_1}
    \vert B \vert \;&=\; b_{00} B_{00} \,-\, b_{01} B_{01}  \,+\, b_{02} B_{02} \,-\, b_{03} B_{03}\,. 
\end{align}
Essentially we need to evaluate $B_{00}, B_{11}, B_{22}, B_{33}$ and $B_{01}, B_{02}, B_{03}$ (to compute $\vert B \vert$), and to establish that $B_{00}/\vert B \vert$, $B_{11}/\vert B \vert$, $B_{22}/\vert B \vert$ and $B_{33}/\vert B \vert$ are all $o(1)$ terms. Towards that, note that
{
\allowdisplaybreaks
\begin{align}\label{supp_cons_2}
    B_{00} \;&=\; \begin{vmatrix}
       b_{11} & b_{12} & b_{13}\\
       b_{12} & b_{22} & b_{23} \\
       b_{13} & b_{23} & b_{33} 
    \end{vmatrix} 
    \;=\; b_{11}\,\left(b_{22} b_{33} - b_{23}^2\right) \;-\; b_{12}\,\left(b_{12} b_{33} - b_{13} b_{23}\right) \;+\; b_{13}\,\left(b_{12} b_{23} - b_{22} b_{13}\right) \notag \\
    &=\; \frac{1}{\lambda_3^2}\,2N_3\,(2n N_1)^2\,\Big[ 2N_3\,(n N_1)^2\,\left(\frac{1}{\lambda_2^2} + \frac{1}{\lambda_3^2}\right) \left\{ \frac{1}{\lambda_1^2}\,2N_3\,2(n-1)\,N_1^2 \,+\, 2N_3\,N_1^2\,\left(\frac{1}{\lambda_2^2} + \frac{1}{\lambda_3^2}\right) \right\} \notag \\
    &\; \qquad \qquad \qquad \qquad \qquad  -\,\left( 2N_3\,nN_1^2 \right)^2\,\left(\frac{1}{\lambda_2^2} + \frac{1}{\lambda_3^2}\right)^2 \Big] \notag \\
    &\qquad -\, \frac{1}{\lambda_3^2}\,2N_3\,2 (n N_1)^2\,\Big[\,\frac{1}{\lambda_3^2}\,2N_3\,2 (n N_1)^2 \, \left\{ \frac{1}{\lambda_1^2}\,2N_3\,2(n-1)\,N_1^2 \,+\, \frac{1}{\lambda_2^2}\,2N_3\,N_1^2 \,+\, \frac{1}{\lambda_3^2}\,2N_3\,N_1^2 \right\} \notag \\
    & \qquad \qquad \qquad \qquad \qquad \qquad - \, \frac{1}{\lambda_3^2}\, 2N_3\,2n\,N_1^2\,\cdot\,2N_3\,n N_1^2 \,\left(\frac{1}{\lambda_2^2} + \frac{1}{\lambda_3^2}\right) \Big] \notag \\
    & \qquad +  \frac{1}{\lambda_3^2}\,2N_3\,2n\,N_1^2\,\Big[\, \frac{1}{\lambda_3^2}\,2N_3\,2 (n N_1)^2 \, \cdot\, 2N_3\,nN_1^2\,\left(\frac{1}{\lambda_2^2} + \frac{1}{\lambda_3^2}\right) \notag \\
    & \qquad \qquad \qquad \qquad \qquad -\, 2N_3\,(nN_1)^2\,\left(\frac{1}{\lambda_2^2} + \frac{1}{\lambda_3^2}\right)\,\cdot\, \frac{1}{\lambda_3^2}\,2N_3\,2n\,N_1^2\,\big] \notag \\
    &=\; \frac{1}{\lambda_3^2}\,2N_3\,(2nN_1)^2\,\cdot 2N_3\,(nN_1)^2\, \cdot\, 2N_3\,2(n-1)\,N_1^2\,\frac{1}{\lambda_1^2}\,\left(\frac{1}{\lambda_2^2} + \frac{1}{\lambda_3^2}\right) \notag \\
    & \qquad -\, \frac{1}{\lambda_3^2}\,2N_3\,2(nN_1)^2\,\cdot\,\frac{1}{\lambda_1^2 \lambda_3^2}\,(2N_3)^2\,2^2\,(nN_1^2)^2\,(n-1)\;+\;0 \notag \\
    &=\; \frac{8}{\lambda_1^2 \lambda_3^2}\,(2N_3)^3 (nN_1)^2 (nN_1^2)^2 (n-1) \left( \frac{1}{\lambda_2^2} + \frac{1}{\lambda_3^2} - \frac{1}{\lambda_3^2} \right)\;=\; \frac{8 (2N_3)^3 n^4 (n-1) N_1^6}{\lambda_1^2 \lambda_2^2 \lambda_3^2}\,,
\end{align}
}

{
\allowdisplaybreaks
\begin{align}\label{supp_cons_3}
    B_{01} \;&=\; \begin{vmatrix}
       b_{01} & b_{12} & b_{13}\\
       b_{02} & b_{22} & b_{23} \\
       b_{03} & b_{23} & b_{33} 
    \end{vmatrix} 
    \;=\; b_{01}\,\left(b_{22} b_{33} - b_{23}^2\right) \;-\; b_{12}\,\left(b_{02} b_{33} - b_{03} b_{23}\right) \;+\; b_{13}\,\left(b_{02} b_{23} - b_{03} b_{22}\right) \notag \\
    &=\; \frac{1}{\lambda_3^2}\,2N_3\,2n\,N_1\,\Big[ 2N_3\,(n N_1)^2\,\left(\frac{1}{\lambda_2^2} + \frac{1}{\lambda_3^2}\right) \left\{ \frac{1}{\lambda_1^2}\,2N_3\,2(n-1)\,N_1^2 \,+\, 2N_3\,N_1^2\,\left(\frac{1}{\lambda_2^2} + \frac{1}{\lambda_3^2}\right) \right\} \notag \\
    &\; \qquad \qquad    -\,\left( 2N_3\,nN_1^2 \right)^2\,\left(\frac{1}{\lambda_2^2} + \frac{1}{\lambda_3^2}\right)^2 \Big] \notag \\
    &\qquad -\, \frac{1}{\lambda_3^2}\,2N_3\,2 (n N_1)^2\,\Big[\,2N_3\,n N_1\,\left(\frac{1}{\lambda_2^2} + \frac{1}{\lambda_3^2}\right)\, \left\{ \frac{1}{\lambda_1^2}\,2N_3\,2(n-1)\,N_1^2 \,+\, 2N_3\,N_1^2\,\left(\frac{1}{\lambda_2^2} + \frac{1}{\lambda_3^2}\right) \right\} \notag \\
    & \qquad \qquad \qquad   - \, \left\{ \frac{1}{\lambda_1^2}\,2N_3\,2(n-1)\,N_1 \,+\, 2N_3\,N_1\,\left(\frac{1}{\lambda_2^2} + \frac{1}{\lambda_3^2}\right) \right\}\,2N_3\,n N_1^2\,\left(\frac{1}{\lambda_2^2} + \frac{1}{\lambda_3^2}\right) \Big] \notag \\
    & \qquad +  \frac{1}{\lambda_3^2}\,2N_3\,2n\,N_1^2\,\Big[\, 2N_3\,n N_1\,\left(\frac{1}{\lambda_2^2} + \frac{1}{\lambda_3^2}\right) \, \cdot\, 2N_3\,n N_1^2\,\left(\frac{1}{\lambda_2^2} + \frac{1}{\lambda_3^2}\right) \notag  \\
    & \qquad \qquad \qquad  -\, \left\{ \frac{1}{\lambda_1^2}\,2N_3\,2(n-1)\,N_1 \,+\, 2N_3\,N_1\,\left(\frac{1}{\lambda_2^2} + \frac{1}{\lambda_3^2}\right) \right\}\,2N_3\,(nN_1)^2\,\left(\frac{1}{\lambda_2^2} + \frac{1}{\lambda_3^2}\right) \big] \notag \\
    &=\;\frac{1}{\lambda_3^2}\,2N_3\,2n\,N_1\,\cdot \, \frac{1}{\lambda_1^2}\,(2N_3)^2\,2(n-1)\,n^2 N_1^4\,\left(\frac{1}{\lambda_2^2} + \frac{1}{\lambda_3^2}\right) \notag \\
    & \;\; -\, \frac{1}{\lambda_3^2}\,2N_3\,2 (n N_1)^2\,\Big[\,(2N_3)^2\,2(n-1)\,nN_1^3\,\frac{1}{\lambda_1^2}\left(\frac{1}{\lambda_2^2} + \frac{1}{\lambda_3^2}\right) \,+\, (2N_3)^2\,nN_1^3\,\left(\frac{1}{\lambda_2^2} + \frac{1}{\lambda_3^2}\right)^2 \notag \\
    & \qquad \qquad \qquad  -\, (2N_3)^2\,2(n-1)\,nN_1^3\,\frac{1}{\lambda_1^2}\left(\frac{1}{\lambda_2^2} + \frac{1}{\lambda_3^2}\right) \,-\, (2N_3)^2\,nN_1^3\,\left(\frac{1}{\lambda_2^2} + \frac{1}{\lambda_3^2}\right)^2 \Big] \notag \\
    & \;\; -\, (2N_3)^3\,2(n-1)\,2 n^3 N_1^5\,\frac{1}{\lambda_1^2} \frac{1}{\lambda_3^2}\left(\frac{1}{\lambda_2^2} + \frac{1}{\lambda_3^2}\right) \notag \\
    &=\; 0\,,
\end{align}
}

{
\allowdisplaybreaks
\begin{align}\label{supp_cons_4}
    B_{02} \;&=\; \begin{vmatrix}
       b_{01} & b_{11} & b_{13}\\
       b_{02} & b_{12} & b_{23} \\
       b_{03} & b_{13} & b_{33} 
    \end{vmatrix} 
    \;=\; b_{01}\,\left(b_{12} b_{33} - b_{13} b_{23}\right) \;-\; b_{11}\,\left(b_{02} b_{33} - b_{03} b_{23}\right) \;+\; b_{13}\,\left(b_{02} b_{13} - b_{03} b_{12}\right) \notag \\
    &=\; \frac{1}{\lambda_3^2}\,2N_3\,2n\,N_1\,\Big[ \frac{1}{\lambda_3^2}\,2N_3\,2(n N_1)^2\, \left\{ \frac{1}{\lambda_1^2}\,2N_3\,2(n-1)\,N_1^2 \,+\, 2N_3\,N_1^2\,\left(\frac{1}{\lambda_2^2} + \frac{1}{\lambda_3^2}\right) \right\} \notag \\
    &\; \qquad \qquad    -\,\frac{1}{\lambda_3^2}\,2N_3\,2n N_1^2\,\cdot 2N_3\,n N_1^2\,\left(\frac{1}{\lambda_2^2} + \frac{1}{\lambda_3^2}\right) \Big] \notag \\
    &\;\; -\, \frac{1}{\lambda_3^2}\,2N_3\,(2n N_1)^2\,\Big[\,2N_3\,n N_1\,\left(\frac{1}{\lambda_2^2} + \frac{1}{\lambda_3^2}\right)\, \left\{ \frac{1}{\lambda_1^2}\,2N_3\,2(n-1)\,N_1^2 \,+\, 2N_3\,N_1^2\,\left(\frac{1}{\lambda_2^2} + \frac{1}{\lambda_3^2}\right) \right\} \notag \\
    & \qquad \qquad \qquad   - \, \left\{ \frac{1}{\lambda_1^2}\,2N_3\,2(n-1)\,N_1 \,+\, 2N_3\,N_1\,\left(\frac{1}{\lambda_2^2} + \frac{1}{\lambda_3^2}\right) \right\}\,2N_3\,n N_1^2\,\left(\frac{1}{\lambda_2^2} + \frac{1}{\lambda_3^2}\right) \Big] \notag \\
    & \; +  \frac{1}{\lambda_3^2}\,2N_3\,2n\,N_1^2\,\Big[\, 2N_3\,n N_1\,\left(\frac{1}{\lambda_2^2} + \frac{1}{\lambda_3^2}\right) \, \cdot\, \frac{1}{\lambda_3^2}\, 2N_3\,2n N_1^2 \notag  \\
    & \qquad \qquad  -\, \left\{ \frac{1}{\lambda_1^2}\,2N_3\,2(n-1)\,N_1 \,+\, 2N_3\,N_1\,\left(\frac{1}{\lambda_2^2} + \frac{1}{\lambda_3^2}\right) \right\}\,\frac{1}{\lambda_3^2}\,2N_3\,2(nN_1)^2 \big] \notag  \\
    &=\;\frac{1}{\lambda_1^2 \lambda_3^4}\,(2N_3)^3 2^3 (n-1) n^3 N_1^5 \,-\,\frac{1}{\lambda_1^2 \lambda_3^4}\,(2N_3)^3 2^3 (n-1) n^3 N_1^5 \notag  \\
    &=\; 0\,,
\end{align}
}

{
\allowdisplaybreaks
\begin{align}\label{supp_cons_5}
    B_{03} \;&=\; \begin{vmatrix}
       b_{01} & b_{11} & b_{12}\\
       b_{02} & b_{12} & b_{22} \\
       b_{03} & b_{13} & b_{23} 
    \end{vmatrix} 
    \;=\; b_{01}\,\left(b_{12} b_{23} - b_{22} b_{13}\right) \;-\; b_{11}\,\left(b_{02} b_{23} - b_{22} b_{03}\right) \;+\; b_{12}\,\left(b_{02} b_{13} - b_{03} b_{12}\right) \notag \\
    &=\; \frac{1}{\lambda_3^2}\,2N_3\,2n\,N_1\,\Big[ \frac{1}{\lambda_3^2}\,2N_3\,2(n N_1)^2\,\cdot\,2N_3\,n N_1^2\,\left(\frac{1}{\lambda_2^2} + \frac{1}{\lambda_3^2}\right) \notag \\
    & \qquad \qquad -\, 2N_3\,(nN_1)^2\,\left(\frac{1}{\lambda_2^2} + \frac{1}{\lambda_3^2}\right)\,\cdot\,\frac{1}{\lambda_3^2}\,2N_3\,2n N_1^2 \Big] \notag \\
    &\;\; -\, \frac{1}{\lambda_3^2}\,2N_3\,(2n N_1)^2\,\Big[\,2N_3\,n N_1\,\left(\frac{1}{\lambda_2^2} + \frac{1}{\lambda_3^2}\right)\, \cdot\,2N_3\,n N_1^2\,\left(\frac{1}{\lambda_2^2} + \frac{1}{\lambda_3^2}\right) \notag \\
    & \qquad \qquad   - \, 2N_3\,(n N_1)^2 \left(\frac{1}{\lambda_2^2} + \frac{1}{\lambda_3^2}\right) \,\cdot\, \left\{\frac{1}{\lambda_1^2}\, 2N_3\,2(n-1)\,N_1 \,+\, 2N_3 N_1\,\left(\frac{1}{\lambda_2^2} + \frac{1}{\lambda_3^2}\right) \right\} \Big] \notag \\
    & \; +  \frac{1}{\lambda_3^2}\,2N_3\,2(n N_1)^2\,\Big[\, 2N_3\,n N_1\,\left(\frac{1}{\lambda_2^2} + \frac{1}{\lambda_3^2}\right) \, \cdot\, \frac{1}{\lambda_3^2}\, 2N_3\,2n N_1^2 \notag \\
    & \qquad \qquad \qquad \qquad -\, \left\{ \frac{1}{\lambda_1^2}\,2N_3\,2(n-1)\,N_1 \,+\, 2N_3\,N_1\,\left(\frac{1}{\lambda_2^2} + \frac{1}{\lambda_3^2}\right) \right\}\,\frac{1}{\lambda_3^2}\,2N_3\,2(nN_1)^2 \big] \notag \\
    &=\;0 + \frac{1}{\lambda_1^2 \lambda_3^2}\,\left(\frac{1}{\lambda_2^2} + \frac{1}{\lambda_3^2}\right)\,(2N_3)^3 2^3 (n-1) n^4 N_1^5 \,-\,\frac{1}{\lambda_1^2 \lambda_3^4}\,(2N_3)^3 2^3 (n-1) n^4 N_1^5 \notag  \\
    &=\; \frac{1}{\lambda_1^2 \lambda_2^2 \lambda_3^2}\,(2N_3)^3 2^3 (n-1) n^4 N_1^5\,.
\end{align}
}
With the above, we have from (\ref{supp_cons_1})
\begin{align}\label{supp_cons_6}
    \vert B \vert \;&=\; b_{00}B_{00} \,-\, b_{03} B_{03} \notag \\
    &=\; \left( \frac{1}{\lambda_0^2}\, 2N_3\, 2n\, (N_1 -1) \;+\; \frac{1}{\lambda_1^2}\, 2N_3\, 2\,(n-1) \;+\; \frac{1}{\lambda_2^2}\, 2N_3 \;+\; \frac{1}{\lambda_3^2}\, 2N_3 \right)\, \frac{8 (2N_3)^3 n^4 (n-1) N_1^6}{\lambda_1^2 \lambda_2^2 \lambda_3^2} \notag \\
    & \qquad -\, \left( \frac{1}{\lambda_1^2}\, 2N_3\, 2\,(n-1)\,N_1 \;+\; \frac{1}{\lambda_2^2}\, 2N_3\,N_1 \;+\; \frac{1}{\lambda_3^2}\, 2N_3\,N_1 \right)\, \frac{1}{\lambda_1^2 \lambda_2^2 \lambda_3^2}\,(2N_3)^3 2^3 (n-1) n^4 N_1^5 \notag \\
    &=\; \frac{1}{\lambda_0^2 \lambda_1^2 \lambda_2^2 \lambda_3^2}\,16\,(2N_3)^4 (n-1)\,n^5\,(N_1-1)\,N_1^6\,.
\end{align}
From (\ref{supp_cons_2}) and (\ref{supp_cons_6}), we have using some straightforward algebra
\begin{align}\label{supp_cons_7}
    \frac{B_{00}}{\vert B \vert} \;=\; \frac{\lambda_0^2}{4 N_3\,n\,(N_1-1)} \;=\; O\left(\frac{1}{N_3\,n\,N_1}\right)\,,
\end{align}
as $\lambda_0^2$ is a constant. Likewise, observe that

{
\allowdisplaybreaks
\begin{align}\label{supp_cons_8}
    B_{11} \;&=\; \begin{vmatrix}
       b_{00} & b_{02} & b_{03}\\
       b_{02} & b_{22} & b_{23} \\
       b_{03} & b_{23} & b_{33} 
    \end{vmatrix} 
    \;=\; b_{00}\,\left(b_{22} b_{33} - b_{23}^2\right) \;-\; b_{02}\,\left(b_{02} b_{33} - b_{03} b_{23}\right) \;+\; b_{03}\,\left(b_{02} b_{23} - b_{03} b_{22}\right) \notag \\
    &=\; \left\{\frac{1}{\lambda_0^2}\,2N_3\,2n\,(N_1-1) \,+\, \frac{1}{\lambda_1^2}\,2N_3\,2(n-1) \,+\, 2N_3\,\left(\frac{1}{\lambda_2^2} + \frac{1}{\lambda_3^2}\right) \right\} \notag \\
    & \qquad \Big[ 2N_3\,(n N_1)^2\, \,\left(\frac{1}{\lambda_2^2} + \frac{1}{\lambda_3^2}\right)\,\left\{ \frac{1}{\lambda_1^2}\,2N_3\,2(n-1)\,N_1^2 \,+\, 2N_3\,N_1^2\,\left(\frac{1}{\lambda_2^2} + \frac{1}{\lambda_3^2}\right) \right\} \notag \\
    &\; \qquad \qquad    -\,(2N_3\,n N_1^2)^2\,\left(\frac{1}{\lambda_2^2} + \frac{1}{\lambda_3^2}\right)^2 \Big] \notag \\
    &\;\; -\, 2N_3\,n N_1\,\left(\frac{1}{\lambda_2^2} + \frac{1}{\lambda_3^2}\right)\,\Big[\,2N_3\,n N_1\,\left(\frac{1}{\lambda_2^2} + \frac{1}{\lambda_3^2}\right)\, \left\{ \frac{1}{\lambda_1^2}\,2N_3\,2(n-1)\,N_1^2 \,+\, 2N_3\,N_1^2\,\left(\frac{1}{\lambda_2^2} + \frac{1}{\lambda_3^2}\right) \right\} \notag \\
    & \qquad \qquad \qquad   - \, \left\{ \frac{1}{\lambda_1^2}\,2N_3\,2(n-1)\,N_1 \,+\, 2N_3\,N_1\,\left(\frac{1}{\lambda_2^2} + \frac{1}{\lambda_3^2}\right) \right\}\,2N_3\,n N_1^2\,\left(\frac{1}{\lambda_2^2} + \frac{1}{\lambda_3^2}\right) \Big] \notag \\
    & \;\; +  \left\{ \frac{1}{\lambda_1^2}\,2N_3\,2(n-1)\,N_1 \,+\, 2N_3\,N_1\,\left(\frac{1}{\lambda_2^2} + \frac{1}{\lambda_3^2}\right) \right\}\,\Big[\, 2N_3\,n N_1\,\left(\frac{1}{\lambda_2^2} + \frac{1}{\lambda_3^2}\right) \, 2N_3\,n N_1^2\,\left(\frac{1}{\lambda_2^2} + \frac{1}{\lambda_3^2}\right) \notag \\
    & \qquad \qquad  -\, \left\{ \frac{1}{\lambda_1^2}\,2N_3\,2(n-1)\,N_1 \,+\, 2N_3\,N_1\,\left(\frac{1}{\lambda_2^2} + \frac{1}{\lambda_3^2}\right) \right\}\,2N_3\,(nN_1)^2\,\left(\frac{1}{\lambda_2^2} + \frac{1}{\lambda_3^2}\right) \big] \notag \\
    &=\;  \left\{\frac{1}{\lambda_0^2}\,2N_3\,2n\,(N_1-1) \,+\, \frac{1}{\lambda_1^2}\,2N_3\,2(n-1) \,+\, 2N_3\,\left(\frac{1}{\lambda_2^2} + \frac{1}{\lambda_3^2}\right) \right\} \,(2N_3)^2\, 2(n-1)\,n^2N_1^4\,\frac{1}{\lambda_1^2}\,\left(\frac{1}{\lambda_2^2} + \frac{1}{\lambda_3^2}\right) \notag \\
    &\;\;- \left\{ \frac{1}{\lambda_1^2}\,2N_3\,2(n-1)\,N_1^2 \,+\, 2N_3\,N_1^2\,\left(\frac{1}{\lambda_2^2} + \frac{1}{\lambda_3^2}\right) \right\}\, \,(2N_3)^2\, 2(n-1)\,n^2 N_1^4\,\frac{1}{\lambda_1^2}\,\left(\frac{1}{\lambda_2^2} + \frac{1}{\lambda_3^2}\right) \notag \\
    &=\;  \,(2N_3)^3\, 4(n-1)\,n^3 N_1^4\,\frac{1}{\lambda_0^2 \lambda_1^2}\,\left(\frac{1}{\lambda_2^2} + \frac{1}{\lambda_3^2}\right)\,.
\end{align}
}
From (\ref{supp_cons_8}) and (\ref{supp_cons_6}), we have using some straightforward algebra
\begin{align}\label{supp_cons_9}
\frac{B_{11}}{\vert B \vert} \;=\; \frac{\left( \lambda_2^2 \,+\, \lambda_3^2 \right)}{8\,N_3\,n^2\,N_1^2}  \;=\; O\left(\frac{n^2\,N_1^2}{N_3\,n^2\,N_1^2}\right)\;=\; O\left( 
\frac{1}{N_3} \right)\,. 
\end{align}
Similarly we have
{
\allowdisplaybreaks
\begin{align}\label{supp_cons_10}
    B_{22} \;&=\; \begin{vmatrix}
       b_{00} & b_{01} & b_{03}\\
       b_{01} & b_{11} & b_{13} \\
       b_{03} & b_{13} & b_{33} 
    \end{vmatrix} 
    \;=\; b_{00}\,\left(b_{11} b_{33} - b_{13}^2\right) \;-\; b_{01}\,\left(b_{01} b_{33} - b_{03} b_{13}\right) \;+\; b_{03}\,\left(b_{01} b_{13} - b_{03} b_{11}\right) \notag \\
    &=\; \left\{\frac{1}{\lambda_0^2}\,2N_3\,2n\,(N_1-1) \,+\, \frac{1}{\lambda_1^2}\,2N_3\,2(n-1) \,+\, 2N_3\,\left(\frac{1}{\lambda_2^2} + \frac{1}{\lambda_3^2}\right) \right\} \notag \\
    & \qquad \Big[ \frac{1}{\lambda_3^2}\,2N_3\,(2n N_1)^2\, \,\left\{ \frac{1}{\lambda_1^2}\,2N_3\,2(n-1)\,N_1^2 \,+\, 2N_3\,N_1^2\,\left(\frac{1}{\lambda_2^2} + \frac{1}{\lambda_3^2}\right) \right\} \,-\,\left(\frac{1}{\lambda_3^2}\,2N_3\,2n N_1^2\right)^2 \Big] \notag \\
    &\;\; -\, \frac{1}{\lambda_3^2}\,2N_3\,2n N_1\,\Big[\frac{1}{\lambda_3^2}\,\,2N_3\,2n N_1\, \left\{ \frac{1}{\lambda_1^2}\,2N_3\,2(n-1)\,N_1^2 \,+\, 2N_3\,N_1^2\,\left(\frac{1}{\lambda_2^2} + \frac{1}{\lambda_3^2}\right) \right\} \notag \\
    & \qquad \qquad \qquad   - \, \left\{ \frac{1}{\lambda_1^2}\,2N_3\,2(n-1)\,N_1 \,+\, 2N_3\,N_1\,\left(\frac{1}{\lambda_2^2} + \frac{1}{\lambda_3^2}\right) \right\}\,\frac{1}{\lambda_3^2}\,2N_3\,2n N_1^2 \Big] \notag \\
    & \;\; +  \left\{ \frac{1}{\lambda_1^2}\,2N_3\,2(n-1)\,N_1 \,+\, 2N_3\,N_1\,\left(\frac{1}{\lambda_2^2} + \frac{1}{\lambda_3^2}\right) \right\}\,\Big[\, \frac{1}{\lambda_3^2}\,2N_3\,2n N_1 \,\cdot\,\frac{1}{\lambda_3^2}\, 2N_3\,2n N_1^2 \notag \\
    & \qquad \qquad  -\, \left\{ \frac{1}{\lambda_1^2}\,2N_3\,2(n-1)\,N_1 \,+\, 2N_3\,N_1\,\left(\frac{1}{\lambda_2^2} + \frac{1}{\lambda_3^2}\right) \right\}\,\frac{1}{\lambda_3^2}\,2N_3\,(2nN_1)^2 \big] \notag \\
    &=\;  \left\{\frac{1}{\lambda_0^2}\,2N_3\,2n\,(N_1-1) \,+\, \frac{1}{\lambda_1^2}\,2N_3\,2(n-1) \,+\, 2N_3\,\left(\frac{1}{\lambda_2^2} + \frac{1}{\lambda_3^2}\right) \right\} \,\frac{1}{\lambda_3^2}\,2N_3\,(2nN_1)^2 \notag \\
    & \qquad \qquad \cdot \,\left\{ \frac{1}{\lambda_1^2}\,2N_3\,2(n-1)\,N_1^2 \,+\, \frac{1}{\lambda_2^2}\,2N_3\,N_1^2 \right\} \notag \\
    &\;\;- \left\{ \frac{1}{\lambda_1^2}\,2N_3\,2(n-1)\,N_1 \,+\, 2N_3\,N_1\,\left(\frac{1}{\lambda_2^2} + \frac{1}{\lambda_3^2}\right) \right\}\, \frac{1}{\lambda_3^2}\,2N_3\,(2nN_1)^2 \,\left\{ \frac{1}{\lambda_1^2}\,2N_3\,2(n-1)\,N_1 \,+\, \frac{1}{\lambda_2^2}\,2N_3\,N_1 \right\}  \notag \\
    &=\;  \,(2N_3)^3\, 4(n-1)\,n^3 N_1^4\,\frac{1}{\lambda_0^2 \lambda_1^2}\,\left(\frac{1}{\lambda_2^2} + \frac{1}{\lambda_3^2}\right) \notag \\
    &=\; \frac{2N_3\,2n\,(N_1-1)}{\lambda_0^2}\,\cdot\, \frac{2N_3\,(2nN_1)^2}{\lambda_3^2}\,2N_3N_1^2\,\left\{ \frac{2(n-1)}{\lambda_1^2} + \frac{1}{\lambda_2^2} \right\} \notag \\
    &=\; \frac{(2N_3)^3\,8n^3\,(N_1-1)N_1^4}{\lambda_0^2 \lambda_1^2 \lambda_2^2 \lambda_3^2}\,\left( 2(n-1)\lambda_2^2 + \lambda_1^2 \right)\,.
\end{align}
}
With some straightforward algebraic manipulations, (\ref{supp_cons_10}) and (\ref{supp_cons_6}) together yield
\begin{align}\label{supp_cons_11}
\frac{B_{22}}{\vert B \vert} \;=\; \frac{\left( 2(n-1)\lambda_2^2 \,+\, \lambda_1^2 \right)}{4\,N_3\,n^2(n-1)\,N_1^2}  \;=\; O\left(\frac{n^3\,N_1^2}{N_3\,n^3\,N_1^2}\right)\;=\; O\left( 
\frac{1}{N_3} \right)\,. 
\end{align}

\noindent And finally, 
{
\allowdisplaybreaks
\begin{align}
    B_{33} \;&=\; \begin{vmatrix}
       b_{00} & b_{01} & b_{02}\\
       b_{01} & b_{11} & b_{12} \\
       b_{02} & b_{12} & b_{22} 
    \end{vmatrix} 
    \;=\; b_{00}\,\left(b_{11} b_{22} - b_{12}^2\right) \;-\; b_{01}\,\left(b_{01} b_{22} - b_{02} b_{12}\right) \;+\; b_{02}\,\left(b_{01} b_{12} - b_{02} b_{11}\right) \notag \\
    &=\; \left\{\frac{1}{\lambda_0^2}\,2N_3\,2n\,(N_1-1) \,+\, \frac{1}{\lambda_1^2}\,2N_3\,2(n-1) \,+\, 2N_3\,\left(\frac{1}{\lambda_2^2} + \frac{1}{\lambda_3^2}\right) \right\} \notag \\
    & \qquad \Big[\, \frac{1}{\lambda_3^2}\,2N_3\,(2n N_1)^2\, \, 2N_3\,(nN_1)^2\,\left(\frac{1}{\lambda_2^2} + \frac{1}{\lambda_3^2}\right)  \,-\,\left(\frac{1}{\lambda_3^2}\,2N_3\,2(n N_1)^2\right)^2 \Big] \notag \\
    & \;\;\;\; -\; \frac{2N_3\,2n\,N_1}{\lambda_3^2} \Big[\,\frac{2N_3\,2n\,N_1}{\lambda_3^2}\, 2N_3\,(nN_1)^2\,\left(\frac{1}{\lambda_2^2} + \frac{1}{\lambda_3^2}\right) \,-\, 2N_3\,nN_1\,\left(\frac{1}{\lambda_2^2} + \frac{1}{\lambda_3^2}\right)\,\frac{2N_3\,2\,(n N_1)^2}{\lambda_3^2}  \Big] \notag \\
    &\;\;\;\; +\; 2N_3\,nN_1\,\left(\frac{1}{\lambda_2^2} + \frac{1}{\lambda_3^2}\right)\, \Big[\,\frac{2N_3\,2\,n N_1}{\lambda_3^2}\,\cdot\, \frac{2N_3\,2\,(n N_1)^2}{\lambda_3^2}\;-\; \frac{2N_3\,(2n N_1)^2}{\lambda_3^2}\,\cdot\,2N_3\,nN_1\,\left(\frac{1}{\lambda_2^2} + \frac{1}{\lambda_3^2}\right) \Big] \notag \\
    &=\; \left\{\frac{1}{\lambda_0^2}\,2N_3\,2n\,(N_1-1) \,+\, \frac{1}{\lambda_1^2}\,2N_3\,2(n-1) \,+\, 2N_3\,\left(\frac{1}{\lambda_2^2} + \frac{1}{\lambda_3^2}\right) \right\}\,\frac{2N_3\,(2nN_1)^2}{\lambda_3^2}\,\cdot\,\frac{2N_3\,(n N_1)^2}{\lambda_2^2}\, \notag \\
    & \;\;\;\; -\; 2N_3\,nN_1\,\left(\frac{1}{\lambda_2^2} + \frac{1}{\lambda_3^2}\right)\,\cdot\, \frac{2N_3\,(2n N_1)^2}{\lambda_3^2}\,\cdot\, \frac{2N_3\,nN_1}{\lambda_2^2} \notag \\
    &=\; 2N_3\,2\left\{ \frac{n(N_1-1)}{\lambda_0^2} + \frac{n-1}{\lambda_1^2} \right\}\,\frac{(2N_3)^2\,4(nN_1)^4}{\lambda_2^2\,\lambda_3^2} \notag \\
    &=\; \frac{(2N_3)^3\,8\,n^4 N_1^4}{\lambda_0^2\,\lambda_1^2\,\lambda_2^2\,\lambda_3^2}\,\left\{ n(N_1-1)\,\lambda_1^2 + (n-1)\,\lambda_0^2\right\} \label{supp_cons_13}
\end{align}
}
With some straightforward algebraic manipulations, (\ref{supp_cons_13}) and (\ref{supp_cons_6}) together yield
\begin{align}\label{supp_cons_14}
\frac{B_{33}}{\vert B \vert} \;=\; O\left(\frac{n\,N_1^3}{N_3\,n^2\,N_1^3}\right)\;=\; O\left( 
\frac{1}{N_3\,n} \right)\,. 
\end{align}
Combining (\ref{supp_cons_7}), (\ref{supp_cons_9}), (\ref{supp_cons_11}) and (\ref{supp_cons_14}), we have that $B_{00}/\vert B \vert$, $B_{11}/\vert B \vert$, $B_{22}/\vert B \vert$ and $B_{33}/\vert B \vert$ are all $o(1)$ quantities as $N_3, n$ and $N_1 \to \infty$. This completes the proof that $\hat{\sigma}_e^2, \hat{\sigma}_1^2, \hat{\sigma}_2^2$ and $\hat{\sigma}_3^2$ are consistent estimators of their respective population parameters.

\end{proof}

\begin{proof}[Proof of Theorem \ref{Thm_asymp_T_level2}]
   Define $Z_{igjk} := Y_{i'gjk}$\, for $i=i' - N_3,\, i'=N_3+1, \dots, 2N_3$. Then from Proposition \ref{Prop_MLE_delta2} we have
   {
    \allowdisplaybreaks
   \begin{align}\label{supp_cons_12}
   \begin{split}
     \hat{\delta} \;&=\; \frac{1}{N_3 n N_1} \dis \sum_{i=1}^{N_3} \sum_{j=1}^n \sum_{k=1}^{N_1} (Y_{i1jk} - Y_{i2jk}) \,-\,  \frac{1}{N_3 n N_1} \dis \sum_{i=N_3+1}^{2N_3} \sum_{j=1}^n \sum_{k=1}^{N_1} (Y_{i1jk} - Y_{i2jk}) \\
     &=\; \frac{1}{N_3 n N_1} \dis \sum_{i=1}^{N_3} \sum_{j=1}^n \sum_{k=1}^{N_1} \left((Y_{i1jk} - Y_{i2jk}) - (Z_{i1jk} - Z_{i2jk})\right)\;=:\; \frac{1}{N_3} \dis \sum_{i=1}^{N_3} U_i\,,\\
     \textrm{and} \qquad \hat{\delta} - \delta \;&=\; \frac{1}{N_3} \dis \sum_{i=1}^{N_3} (U_i - \delta)\,,
     \end{split}
   \end{align}
   }
   where\, $U_i := \frac{1}{n N_1} \sum_{j=1}^n \sum_{k=1}^{N_1} \left((Y_{i1jk} - Y_{i2jk}) - (Z_{i1jk} - Z_{i2jk})\right)$ for $i=1, \dots, N_3$. Now note that
   {
    \allowdisplaybreaks
   \begin{align*}
     var(U_i) \;&=\; \frac{1}{(nN_1)^2} \dis \sum_{j,j'=1}^n \sum_{k,k'=1}^{N_1} cov \left( (Y_{i1jk} - Y_{i2jk}) - (Z_{i1jk} - Z_{i2jk})\,,\, (Y_{i1j'k'} - Y_{i2j'k'}) - (Z_{i1j'k'} - Z_{i2j'k'}) \right)\\
     &=\; \frac{4}{(nN_1)^2}\,\left[ nN_1 (\sigma_e^2 + \sigma_1^2 + \sigma_2^2)\,+\, nN_1(N_1-1)(\sigma_1^2 + \sigma_2^2) \,+\, n(n-1)N_1^2 \sigma_2^2  \right]\\
     &=\; \frac{4}{nN_1}\,(\sigma_e^2 + N_1 \sigma_1^2 + nN_1 \sigma_2^2) \\
     &=:\; \Delta\,.
   \end{align*}
    }
    Define $V_i := \frac{U_i - \delta}{\sqrt{\Delta}}$ for $i=1, \dots, N_3$. Clearly $E(V_i)=0$ and $var(V_i)=1$. Combining this and (\ref{supp_cons_12}), we have $\frac{\hat{\delta} - \delta}{\sqrt{\Delta}} = \frac{1}{N_3} \sum_{i=1}^{N_3} V_i$. Then recalling the expression of $\Var(\hat{\delta})$ from Proposition \ref{Prop_var_delta2}, we have using the central limit theorem
    {
    \allowdisplaybreaks
   \begin{align*}
     \frac{\sqrt{N_3}\left(\hat{\delta} - \delta\right)}{\sqrt{\Delta}} \;=\; \frac{\hat{\delta} - \delta}{\sqrt{\Var(\hat{\delta})}} \;=\; \sqrt{N_3}\,\frac{1}{N_3} \dis \sum_{i=1}^{N_3} V_i \; \overset{d}{\longrightarrow}\; N(0,1)
   \end{align*}
   }
as $N_3 \to \infty$, which completes the proof of the theorem. 

\end{proof}

\begin{proof}[Proof of Proposition \ref{Prop_samp_size2}]
  Note that from Proposition \ref{Prop_var_delta2}, we have  
\begin{align}\label{Prop2_eq1}
    \begin{split}
        \Var(\hat{\delta}) \;&=\; \frac{4}{N_3 \,n\, N_1}\,\left( \sigma_e^2 \,+\, N_1\,\sigma_1^2 \,+\, n\,N_1\,\sigma_2^2 \right)\\
        &=\; \frac{1}{N_1}\,A' \,\,+ \, B'\,,
    \end{split}
\end{align}
where\, $A' = \frac{4\,\sigma_e^2}{N_3\,n}$\, and\, $B'= \frac{4\,\left(\sigma_1^2 + n\,\sigma_2^2\right)}{N_3\,n}$. Combining (\ref{power_ind42}) and (\ref{Prop2_eq1}), we have 
\begin{align*}
    \frac{1}{N_1}\,A' \,\,+ \, B' \;\leq\; \frac{\delta^2}{\big(z_{1-\alpha/2} \,+\, z_{1-\beta} \big)^2}\;,
\end{align*}
which implies
\begin{align*}
        N_1 \;\geq\; A'\,\left( \frac{\delta^2}{\big(z_{1-\alpha/2} \,+\, z_{1-\beta} \big)^2} \,-\, B'\, \right)^{-1}\,.
\end{align*}

\end{proof}

\begin{proof}[Proof of Proposition \ref{Prop_var_delta1}]
From Proposition \ref{Prop_MLE_delta1}, we have
\begin{align*}
    \begin{split}
        \hat{\delta} \;&=\; \left( \frac{1}{N_3 N_2 \,n} \dis \sum_{i=1}^{N_3} \sum_{j=1}^{N_2} \sum_{k=1}^{n} Y_{ij1k}\,\mathbbm{1}(X_i=1) \;-\; \frac{1}{N_3 N_2 \,n} \dis \sum_{i=1}^{N_3} \sum_{j=1}^{N_2} \sum_{k=1}^{n} Y_{ij1k}\,\mathbbm{1}(X_i=0) \right)\\
        & \qquad - \; \left( \frac{1}{N_3 N_2\, n} \dis \sum_{i=1}^{N_3} \sum_{j=1}^{N_2} \sum_{k=1}^{n} Y_{ij2k}\,\mathbbm{1}(X_i=1) \;-\; \frac{1}{N_3 N_2 \,n} \dis \sum_{i=1}^{N_3} \sum_{j=1}^{N_2} \sum_{k=1}^{n} Y_{ij2k}\,\mathbbm{1}(X_i=0) \right)\\
        &=:\; \left(I - II\right) \;-\; \left(III - IV\right)\,.
    \end{split}
\end{align*}
Observe that 
\begin{align}\label{var1}
    \begin{split}
        \Var(I) \;&=\; \frac{1}{(N_3 N_2\, n)^2}\,\times \, \Big[ N_3 N_2\, n\; \sigma^2 \,+\, \sum_{i=1}^{N_3} \sum_{j=1}^{N_2} \sum_{1\leq k\neq k' \leq n} \Cov(Y_{ij1k}, Y_{ij1k'}) \\
        & \qquad +\, \sum_{i=1}^{N_3} \sum_{1\leq j\neq j' \leq N_2} \sum_{k,k'=1}^{n} \Cov(Y_{ij1k}, Y_{ij'1k'}) \Big]\\
        &=\; \frac{1}{(N_3 N_2\, n)^2}\,\times \, \Big[ N_3 N_2\, n\; \sigma^2 \,+\,  N_3 N_2\, n (n -1)\, \sigma^2\,\rho_1 \,+\, N_3 N_2 (N_2 -1) n^2 \, \sigma^2\,\rho_{(2)} \Big]\\
        &= \; \frac{\sigma^2}{N_3 N_2\, n} \,\times \, \Big[ 1 \,+\, \rho_1\,(n-1) \,+\, \rho_{(2)}\,n\, (N_2 -1) \Big] \;=:\;  \frac{\sigma^2\, f}{N_3 N_2\, n} \,.
    \end{split}
\end{align}
Moreover, due to randomization at level three, the terms I and II are independent. This yields 
\begin{align}\label{var_MLE_1}
    \Var(I-II)\;=\; \frac{2\,\sigma^2\, f}{N_3 N_2\, n}\,.
\end{align}
On similar lines,  
\begin{align}\label{var_MLE_2}
    \Var(III-IV)\;=\; \frac{2\,\sigma^2\, f}{N_3 N_2\, n}\,.
\end{align}
Thus, we have
\begin{align}\label{var_delta_hat}
    \begin{split}
        \Var(\hat{\delta}) \;&=\; \Var(I-II) \,+\, \Var(III-IV) \,-\,\, 2\,\Cov(I-II,\, III-IV)\\
        &=\; \frac{2\,\sigma^2\, f}{N_3 N_2\, n}\,+\, \frac{2\,\sigma^2\, f}{N_3 N_2\, n} \,\,-\,\, 2\,\Cov(I,III) \,-\, 2\,\Cov(II,IV)\,,
    \end{split}
\end{align}
where
\begin{align}\label{cov_terms}
    \begin{split}
        \Cov(I,III) \;&=\; \frac{1}{N_3^2 N_2^2\, n^2}\,\left[ \sum_{i=1}^{N_3} \sum_{j=1}^{N_2} \sum_{k,k'=1}^{n} \Cov(Y_{ij1k},Y_{ij2k'})\,+\, \sum_{i=1}^{N_3} \sum_{1\leq j\neq j'\leq N_2} \sum_{k,k'=1}^n \Cov(Y_{ij1k},Y_{ij'2k'}) \right]\\
        &=\; \frac{1}{N_3^2 N_2^2\, n^2}\,\left( N_3 N_2\, n^2\,\,\sigma^2 \rho_{(1)} \,+\, N_3 N_2 (N_2 -1) \,n^2\,\,\sigma^2 \rho_2  \right)\\
        &=\; \frac{\sigma^2}{N_3 N_2} \left(\, \rho_{(1)} \,+\, (N_2-1)\rho_2\,\right)\,.
    \end{split}
\end{align}
Along similar lines, it can be shown that $\Cov(II,IV) = \Cov(I,III)$. Thus, combining (\ref{var_delta_hat}) and (\ref{cov_terms}), we have 
\begin{align}\label{var_final}
    \begin{split}
        \Var(\hat{\delta}) \;&=\; \frac{4\,\sigma^2}{N_3 N_2\, n}\,\left( 1 + \rho_1\,(n-1) + \rho_{(2)}\,n\, (N_2 -1) \right)\; - \; \frac{4\sigma^2}{N_3 N_2} \left(\, \rho_{(1)} \,+\, (N_2-1)\rho_2\,\right)\\
        &=\; \frac{4}{N_3 N_2\, n}\, \left( \sigma^2 (1-\rho_1) \,+\, n\,\sigma^2\, (\rho_1 - \rho_{(1)}) \,+\, n\,(N_2-1)\,\sigma^2\,(\rho_{(2)} - \rho_2) \right)\\
        &=\; \frac{4}{N_3 N_2\, n}\, \left( \sigma_e^2 \,+\, n\,\sigma_{\text{grp}}^2 \,+\, \,+\, n\,(N_2-1)\,\sigma_{\text{grp}}^2 \right)\\
        &=\; \frac{4}{N_3 N_2\, n}\, \left( \sigma_e^2 \,+\, n\,N_2\,\sigma_{\text{grp}}^2 \right)\,.
    \end{split}
\end{align}

\end{proof}

\begin{proof}[Proof of Proposition \ref{Prop_samp_size1}]
Note that with $N_{11} = N_{12} = n$, we have from Proposition \ref{Prop_var_delta1}
\begin{align} \label{Prop_samp_size1_eq1}
    \begin{split}
        \Var(\hat{\delta}) \;&=\; \frac{4}{N_3 N_2\, n}\, \left( \sigma_e^2 \,+\,n\,N_2\,\sigma_{\text{grp}}^2 \right) \;   =:\; \frac{1}{n}\, A \;+\; B\,.
    \end{split}
\end{align}
Combining (\ref{power_ind4}) and (\ref{Prop_samp_size1_eq1}), we have 
\begin{align*}
 \frac{1}{n}\,A \;+\; B \;\leq\;  \frac{\delta^2}{\big(z_{1-\alpha/2} \,+\, z_{1-\beta} \big)^2}\;,  
\end{align*}
which implies
\begin{align*}
        n \;\geq\; A\,\left( \frac{\delta^2}{\big(z_{1-\alpha/2} \,+\, z_{1-\beta} \big)^2} \,-\, B\, \right)^{-1}\,.
    \end{align*}
    
\end{proof}

\end{document}